\documentclass[aps,prb,notitlepage,nofootinbib]{revtex4-2}
\usepackage{graphicx}
\usepackage{amsfonts}
\usepackage{amssymb}
\usepackage{amsmath}
\usepackage{mathtools}
\usepackage{afterpage}
\usepackage[mathscr]{euscript}
\usepackage{braket}
\usepackage{subfigure}
\usepackage{hyperref}
\hypersetup{
    colorlinks=true, %set true if you want colored links
    linkcolor=blue,  %choose some color if you want links to stand out
}

\let\savecorresponds\corresponds
\let\corresponds\relax
\let\stdcup\cup
\usepackage{mathabx}
\let\corresponds\savecorresponds

\usepackage{amsthm}
\newtheorem{theorem}{Theorem}[section]

\newtheorem{conjecture}[theorem]{Conjecture} 
\newtheorem{corollary}[theorem]{Corollary} 
 
\newtheorem{proposition}[theorem]{Proposition} 

\usepackage{multirow}
\usepackage{float}
\usepackage{qcircuit}
\usepackage{caption}
\usepackage{leftidx}
\usepackage[utf8]{inputenc} %useful to type directly diacritic characters
\usepackage{yfonts}
\usepackage{dsfont}
\usepackage{cancel}
\usepackage{booktabs}
\usepackage[scr=boondox]{mathalpha}
\usepackage{pst-all}

\newcommand{\coho}[1]{\textswab{#1}}
\newcommand{\cohosub}[1]{\protect\scalebox{0.7}{\textswab{#1}}}

\begin{document}
	\def\U{\mathrm{U}(1)}
	\def\SO{\mathrm{SO}}
	\def\SU{\mathrm{SU}}
	\def\Sp{\mathrm{Sp}}
	\def\H{\mathcal{H}}
	\def\F{\mathcal{F}}
	\def\E{\mathbb{E}}
	\def\TT{\mathsf{T}}
	\def\A{\mathcal{A}}
	\def\L{\mathcal{L}}
	\def\w{\mathfrak{w}}
	\def\O{\coho{O}}
	\def\C{\widecheck{\mathcal{C}}}
	
\newcommand{\Aut}{\text{Aut}}
\newcommand{\Z}{\mathbb Z}
\newcommand{\R}{\mathbb R}

	%% Comment
\newcommand{\personalcomment}[1]{\textcolor{red}{[#1]}}
	\newcommand{\commentdb}[1]{\textcolor{blue}{[DB: #1]}}
		\newcommand{\commentmb}[1]{\textcolor{red}{[MB: #1]}}

\def\Sq{\mathop{\mathrm{Sq}}\nolimits}

\newcommand{\mathbbm}[1]{\text{\usefont{U}{bbm}{m}{n}#1}}

    \title{Anomaly cascade in (2+1)D fermionic topological phases}
        
    \author{Daniel Bulmash}
    \author{Maissam Barkeshli}
    \affiliation{Condensed Matter Theory Center and Joint Quantum Institute, Department of Physics, University of Maryland, College Park, Maryland 20472 USA}

    \begin{abstract}
        We develop a theory of anomalies of fermionic topological phases of matter in (2+1)D with a general fermionic symmetry group $G_f$. In general, $G_f$ can be a non-trivial central extension of the bosonic symmetry group $G_b$ by fermion parity $(-1)^F$. We encounter four layers of obstructions to gauging the $G_f$ symmetry, which we dub the anomaly cascade: (i) An $\mathcal{H}^1(G_b,\Z_{\bf T})$ obstruction to extending the symmetry permutations on the anyons to the fermion parity gauged theory, (ii) An $\mathcal{H}^2(G_b, \ker r)$ obstruction to extending the $G_b$ group structure of the symmetry permutations to the fermion parity gauged theory, where $r$ is a map that restricts symmetries of the fermion parity gauged theory to the anyon theory, (iii) An $\mathcal{H}^3(G_b, \Z_2)$ obstruction to extending the symmetry fractionalization class to the fermion parity gauged theory, and (iv) the well-known $\mathcal{H}^4(G_b, U(1))$ obstruction to developing a consistent theory of $G_b$ symmetry defects for the fermion parity gauged theory. We describe how the $\mathcal{H}^2$ obstruction can be canceled by anomaly inflow from a bulk (3+1)D symmetry-protected topological state (SPT) and also its relation to the Arf invariant of spin structures on a torus. If any anomaly in the above sequence is non-trivial, the subsequent ones become relative anomalies. A number of conjectures regarding symmetry actions on super-modular categories, guided by general expectations of anomalies in physics, are also presented.  
    \end{abstract}

    \maketitle
    \tableofcontents
        
    \section{Introduction}
        
        Topological phases of matter in (2+1)D can exhibit the remarkable phenomenon of symmetry fractionalization: in the presence of a global symmetry group $G$, topologically non-trivial quasi-particles can carry fractional quantum numbers under $G$. Such phases are called symmetry-enriched topological phases (SETs) and are described in full generality using the mathematics of $G$-crossed modular tensor categories~\cite{barkeshli2019,Tarantino_SET,turaev2000,ENO2010}. Remarkably, some patterns of symmetry fractionalization yield SETs that, while they are mathematically consistent in the absence of a background $G$ gauge field, cannot be consistently coupled to a $G$ gauge field. The physical interpretation of this inconsistency is that such SETs cannot occur in a purely (2+1)D system with the symmetry generated on-site~\cite{barkeshli2019,Chen2014,kapustin2014b,barkeshli2019rel}. In the language of quantum field theory, such SETs are said to have a 't Hooft anomaly.
        
        Anomalies provide crucial information about a phase of matter; they are renormalization group invariants which connect microscopic and low-energy physics, leading, for example, to generalizations of the famed Lieb-Schulz-Mattis theorem~\cite{cheng2016lsm}. Anomalies are subject to the bulk-boundary correspondence, in the sense that the classification of anomalies in $d$ space-time dimensions is believed to be equivalent to the classification of invertible topological phases in $(d+1)$ space-time dimensions~\cite{freed2014}. In particular, anomalous (2+1)D SETs can exist on the surface of (3+1)D symmetry-protected topological phases (SPTs)~\cite{vishwanath2013,wang2013b,MetlitskiPRB2013,MetlitskiPRB2015,wang2013b,chen2014b,Chen2014,fidkowski2013,bonderson2013sto,wang2014,metlitski2014,ChoPRB2014,kapustin2014b,seiberg2016gapped,hermele2016,YangPRL2015,song2017,bulmash2020,tata2021}. This remarkable fact begs the question of understanding anomaly inflow - can we determine physical processes on the boundary which are sensitive to the anomaly and then see how the non-trivial bulk cancels that anomaly?
        
        For bosonic SETs, the 't Hooft anomaly is relatively well-understood \cite{barkeshli2019,Chen2014,kapustin2014b,barkeshli2019rel,bulmash2020,ENO2010,Cui2016}. A bosonic (3+1)D $G$-SPT is (partially) specified by an element of the cohomology group $\H^4(G,\U)$, and this element dictates the inconsistency of fusion and braiding of symmetry defects together with anyons in the (2+1)D boundary SET. 
        
        In contrast, the data specifying a (3+1)D fermionic SPT (FSPT) has much more structure \cite{kapustin2015Cobordism,kapustin2017,freed2016,WangGu}, which involve the cohomology groups $\H^1(G_b, \Z_{\bf T})$, $\H^2(G_b,\Z_2)$, $\H^3(G_b,\Z_2)$, and $\H^4(G_b,\U)$, where $G_b$ is the ``bosonic" symmetry group obtained by modding out fermion parity symmetry. Specifically, (3+1)D FSPTs form an Abelian group that corresponds to a group extension involving the groups $\H^1(G_b, \Z_{\bf T})$, $\H^2(G_b,\Z_2)/\Gamma^2$, $\H^3(G_b,\Z_2)/\Gamma^3$, and $\H^4(G_b,\U)/\Gamma^4$, where $\Gamma^i$ are certain subgroups of the $\H^i$. 
        
        In particular, recently \cite{WangGu} have shown that (3+1)D FSPTs can in general be characterized by a set of data $(n_1, n_2, n_3,\nu_4) \in Z^1(G_b, \Z_{\bf T}) \times C^2(G_b, \Z_2) \times C^3(G_b, \Z_2) \times C^4(G_b, \U)$, where $C^k$ denotes $k$-cochains, $Z^k$ denotes $k$-cocycles, and $\Z_{\bf T}$ refers to the integers with an action of $G_b$ according to whether group elements are anti-unitary, which can be thought of as involving time-reversal symmetry. The data $(n_1, n_2, n_3,\nu_4)$ satisfy a complicated set of consistency equations and equivalences. These data imply that a (3+1)D FSPT generically determines an element in $[n_1] \in \H^1(G_b, \Z_{\bf T})$. Moreover, when we can set $n_1 = 0$, a (3+1)D FSPT determines an element $[n_2] \in \H^2(G_b, \Z_2)/\Gamma^2$. When we can set $n_1, n_2 = 0$ a (3+1)D FSPT determines an element $[n_3] \in \H^3(G_b, \Z_2)/\Gamma^3$. Finally when we can set $n_1,n_2,n_3 = 0$, a (3+1)D FSPT determines an element of $[\nu_4] \in \H^4(G_b, \U)/\Gamma^4$. The subgroups $\Gamma^i$ can be explicitly determined in general \cite{WangGu}. 
        
        The main aim of this paper is to obtain a detailed and general understanding of the 't Hooft anomaly of any (2+1)D fermionic SET (FSET), and in particular the appearance of the above cohomology groups. The most general description of the anomaly is that it is an obstruction to gauging the full global symmetry group $G_f$. Rather than gauging $G_f$ all at once, our approach is to gauge fermion parity first, which is always possible~\cite{johnsonFreydModularExt}. If the $G_b=G_f/\Z_2^f$ symmetry is to be preserved after gauging, the data specifying symmetry fractionalization must be lifted to the parity-gauged theory. We break the process of finding a consistent lift into a sequence of physically meaningful steps, each of which may be obstructed; the obstruction to each step is given by a piece of data 
        The appearance of FSPT data as an inconsistency of the (2+1)D boundary SET is only partially understood \cite{fidkowski2018,delmastro2021}.  Ref.~\cite{fidkowski2018} gives a partial understanding of the $\H^3$ anomaly of FSETs, which we will make fully general in this paper. Ref.~\cite{delmastro2021} gives examples, mostly with $G_b = \Z_2^{\bf T}$, explaining how ``layers" of anomalies of fermionic phases in general dimensions appear in the action of symmetries on the Hilbert space on a spatial torus, but does not extract any cohomological obstructions for the (2+1)D theory. Our formalism is fully general and starts from the general algebraic data characterizing the FSET rather than the Hilbert space on a torus, although we will make direct contact with Ref.~\cite{delmastro2021} in Sec.~\ref{subsec:Arf}. Ref.~\cite{tata2021} recently showed how in general, given any (2+1)D FSET, one can identify the bulk SPT using a state sum construction. However the results of \cite{tata2021} do not directly explain the appearance of the above cohomology groups in terms of an inconsistency of the boundary (2+1)D FSET. 
        
        \subsection{Summary of main results}
        
        We start with a super-modular category $\mathcal{C}$, which is a unitary braided fusion category (UBFC) that captures the braiding and fusion properties of the anyons. A super-modular category contains a single ``invisible" particle $\psi$, which braids trivially with all other particles and which physically corresponds to the local fermion of the system. We then consider the minimal modular extensions $\C_\nu$ of the super-modular category. A minimal modular extension is a unitary modular tensor category (UMTC) that characterizes the phase obtained by gauging fermion parity. In particular $\C_\nu$ characterizes properties of the anyons, the local fermion $\psi$, and the fermion parity vortices. It was recently proved that every super-modular category admits a minimal modular extension~\cite{johnsonFreydModularExt}, and therefore, according to the ``16-fold way" theorem~\cite{bruillard2017a}, admits exactly 16 distinct minimal modular extensions labeled by $\nu = 0, \cdots, 15$.
        
        In the absence of symmetry, it is expected that a (2+1)D fermionic topological phase can be fully specified by either a choice of $(\mathcal{C}, c_-)$ or $(\C_\nu,c_-)$, where $c_-$ is the chiral central charge of the theory which, physically, determines the system's thermal Hall conductivity. $\mathcal{C}$ determines the theory modulo a fermionic invertible phase (e.g. up to stacking with $p+ip$ superconductors), and therefore determines $c_- \mod 1/2$. In contrast, $\C_\nu$ determines the phase modulo a bosonic invertible phase, and therefore determines $c_- \mod 8$. Two different minimal modular extensions $\C_{\nu}$ and $\C_{\nu'}$ have central charges that differ by $(\nu - \nu')/2 \text{ mod 8}$.
        
        Fermionic systems always have a $\Z_2$ fermion parity symmetry $(-1)^F$, which generates a symmetry group we call $\Z_2^f$. The full symmetry group $G_f$ of a fermionic system is in general a central extension of $G_b = G_f/\Z_2^f$ by $\Z_2^f$, characterized by a cocycle
        \begin{equation}
            \omega_2 \in Z^2(G_b,\Z_2).
        \end{equation}
        
        In this paper, we assume that we are given $G_f$ symmetry fractionalization data on $\mathcal{C}$, as detailed in \cite{bulmash2021,aasen21ferm}, and which we review in detail in Sec.~\ref{subsec:fermionicSymmFrac}. Briefly, this amounts to a map
        \begin{align}
            [\rho]: G_b \rightarrow \text{Aut}_{LR}(\mathcal{C}),
        \end{align}
        where $\text{Aut}_{LR}(\mathcal{C})$ is the group of ``locality-respecting" braided auto-equivalences of $\mathcal{C}$. $\text{Aut}_{LR}(\mathcal{C})$ is similar to the group $\text{Aut}(\mathcal{C})$ of braided auto-equivalences, except it takes the locality of the fermion into account by restricting the classes of maps that are considered trivial to those that act trivially on the fermion. Here the representative maps $\rho_{\bf g}$ satisfy 
        \begin{align}
            \rho_{\bf g h} = \kappa_{\bf g,h} \rho_{\bf g} \rho_{\bf h},
        \end{align}
        where $\kappa_{\bf g,h}$ is a natural isomorphism \cite{barkeshli2019}. 
        Once a representative set of maps $\rho_{\bf g}$ are chosen, symmetry fractionalization is characterized by a set of phases $\eta_a({\bf g}, {\bf h}) \in \U$ for each anyon $a$. These phases $\eta_a({\bf g}, {\bf h})$ are subject to a series of consistency equations, constraints, and gauge equivalences. 
        
        Given a minimal modular extension $\C_\nu$, our aim is to lift the given symmetry fractionalization data from $\mathcal{C}$ to $\C_\nu$. That is, we wish to gauge fermion parity while preserving the $G_b$ symmetry. We perform this procedure systematically and characterize the cascade of obstructions that appear along the way.
        
        \subsubsection{Anomaly cascade: first layer}
        
        The first layer of the anomaly cascade is an obstruction to lifting the maps $[\rho_{\bf g}]$ to autoequivalences of the fermion parity-gauged theory $\C$.
        
        More precisely, the first step is to define a lifted ``topological symmetry" of $\C$, that is, a map 
        \begin{align}
            [\widecheck{\rho}] : G_b \rightarrow \Aut_{LR}(\C_\nu).
        \end{align}
        
        For a particular $\nu$, generically not all elements of $\Aut_{LR}(\mathcal{C})$ can be lifted to elements of $\Aut_{LR}(\C_\nu)$. Thus we encounter the first possible obstruction, which concerns whether there exists some $\nu$ such that for every ${\bf g} \in G_b$, the autoequivalence $[\rho_{\bf g}] \in \Aut_{LR}(\mathcal{C})$ can be lifted to an element $[\widecheck{\rho}_{\bf g}] \in \Aut_{LR}(\C_\nu)$.  
        
        In the case that there does not exist a $\nu$ such that one can lift every element $[\rho_{\bf g}]$ to $[\widecheck{\rho}_{\bf g}] \in \Aut_{LR}(\C_\nu)$, a weaker lift may be possible. In particular, it may be possible that there exists at least one pair $\nu_1$ and $\nu_2$, such that a representative of each $[\rho_{\bf g}]$ lifts to a map
        \begin{align}
        \label{liftPermute}
            \widecheck{\rho}_{\bf g} : \C_{\nu_1} \rightarrow \C_{\nu_2}
        \end{align}
        
        By comparing with the classification of (3+1)D fermion SPTs, which defines the anomaly of (2+1)D fermionic topological phases, we conjecture that the existence of such a weak lift is unobstructed:
        \begin{conjecture}
        Given a super-modular category $\mathcal{C}$ with an action $[\rho_{\bf g}]: G_b \rightarrow \text{Aut}_{LR}(\mathcal{C})$, there always exists at least one pair $\nu_1$ and $\nu_2$ such that there exists a map 
        $\widecheck{\rho}_{\bf g} : \C_{\nu_1} \rightarrow \C_{\nu_2}$ and $\widecheck{\rho}_{\bf g}$ restricts to a representative $\rho_{\bf g}$ on $\mathcal{C}$. 
        \end{conjecture}
        Next, regarding lifts of the form Eq.~\ref{liftPermute}, we prove the following:
       
        \begin{theorem}
        Let $\rho_{\bf g} : \mathcal{C} \rightarrow \mathcal{C}$ be a unitary map. Then any pair $(\nu_1, \nu_2)$ satisfying Eq.~\ref{liftPermute} must have $\nu_1=\nu_2$.
        \end{theorem}
        This follows from the fact that the number of anyons, quantum dimensions, and topological twists must be invariant under the map $\widecheck{\rho}_{\bf g}$, which fixes $\nu_1 = \nu_2$ by the Gauss sum. Furthermore, we have

        \begin{theorem}
        Let $\rho_{\bf g} : \mathcal{C} \rightarrow \mathcal{C}$ be an anti-unitary map. Then the pair $(\nu_1, \nu_2)$ satisfying Eq.~\ref{liftPermute} must necessarily satisfy $e^{2\pi i c_{\nu_2}/8} = e^{-2\pi i c_{\nu_1}/8} = e^{2\pi i (c_{\nu_1} + o_1({\bf g})/2)/8}$, for some integer $o_1({\bf g})$. If there are multiple pairs $(\nu_{i,1}, \nu_{i,2})$ satisfying Eq.~\ref{liftPermute}, then $e^{2\pi i c_{\nu_{i,2}} / 8} = e^{-2\pi i c_{\nu_{i,1}} / 8} = e^{2\pi i (c_{\nu_{i,1}} + o_1({\bf g})/2)/8}$, where $o_1({\bf g})$ is independent of $i$, modulo $2$.
        \label{thm:H1}
        \end{theorem}
        
        A corollary of the above is that:
        \begin{corollary}
        A collection of maps $[\rho_{\bf g}] \in \Aut_{LR}(\mathcal{C})$ for each ${\bf g} \in G_b$ defines an element $[o_1] \in \mathcal{H}^1(G_b, \Z_{\bf T})$. 
        \end{corollary}
        We conjecture that $[o_1]$ is really the obstruction to lifting the symmetry the action, in the following sense:
        \begin{conjecture}
        Suppose that for each ${\bf g} \in G_b$ we are given a lift $\widecheck{\rho}^{(1)}_{\bf g}:\C_{\nu_1}\rightarrow\C_{\nu_2}$ of $\rho_{\bf g}:\mathcal{C} \rightarrow \mathcal{C}$. These maps define $[o_1] \in \H^1(G_b,\Z_{\bf T})$. Then there exists a minimal modular extension $\C_{\nu_0}$ with lifts $\widecheck{\rho}^{(0)}_{\bf g}:\C_{\nu_0}\rightarrow \C_{\nu_0}$ for every ${\bf g} \in G_b$ if and only if $[o_1]=0$.
        \label{conj:H1ActualObstruction}
        \end{conjecture}

        A well-known example of the above $\mathcal{H}^1$ anomaly occurs for the $\SO(3)_3$ super-modular category, which was studied in \cite{fidkowski2013} as an example of a (2+1)D surface theory for an odd index topological superconductor in the class DIII. (Here $\SO(3)_3$ consists of the integer spin representations of the $\SU(2)_6$ affine Kac-Moody algebra, which arises in $\SO(3)_3$ Chern-Simons theory.) All minimal modular extensions of this theory have central charge $c = 1/4 \mod 1/2$, and no minimal modular extension can be compatible with time-reversal symmetry. This is because a UMTC that is compatible with time-reversal symmetry must have central charge $c_- = 0 \mod 4$. 
        
        \subsubsection{Anomaly cascade: second layer}
        
        The second layer of the anomaly cascade is an obstruction to choosing the autoequivalences of the fermion parity-gauged theory so that they compose in an appropriate way, determined by the group $G_f$ and details of the fermion parity-gauged theory.
        
        Suppose that the $\mathcal{H}^1$ anomaly vanishes. Then, according to the discussion above, for each ${\bf g} \in G_b$, we assume we have an invertible map 
        \begin{align}
            \widecheck{\rho}_{\bf g} : \C_{\nu} \rightarrow \C_{\nu},
            \label{eqn:extendedRho}
        \end{align}
        for at least some subset of the possible values of $\nu$. Unless otherwise stated, we now fix a particular choice of $\nu$ and omit it from our notation.
        
        Below we will briefly summarize the obstruction theory for the case where the extension $\C_\nu$ contains an Abelian fermion parity vortex. The more complicated cases will be discussed in the main text. 
        
        When $\C_\nu$ contains an Abelian fermion parity vortex, $[\rho]$ is a group homomorphism, and we need to require that the lift $[\widecheck{\rho}_{\bf g}]$ is also a group homomorphism $G_b \rightarrow \Aut_{LR}(\C)$. We show that there is an obstruction $[o_2] \in \mathcal{H}^2(G_b, \ker r)$ to $[\widecheck{\rho}_{\bf g}]$ defining a group homomorphism. Here $r$ is the restriction map, 
        \begin{align}
            r : \Aut_{LR}(\C) \rightarrow \Aut_{LR}(\C)|_{\mathcal{C}} = \Aut_{LR}(\mathcal{C})
        \end{align}
        which has a non-trivial kernel. 
        
        We characterize $\ker r$ in the following way. First, note that a fermion parity vortex $x$ is called $v$-type if $x \times \psi \neq x$.
        \begin{theorem}
        Suppose that the modular $S$-matrix of $\C$, when restricted to its block involving only $v$-type vortices, is block diagonal with $k$ decoupled blocks. Then if $[\rho] \in \ker r$ and $x$ is a fermion parity vortex, $\rho(x) = x \times \psi^{q(x)}$, where $q(x) \in \{0,1\}$ is independent of $x$ within a block if $x$ is a $v$-type vortex.
        \label{thm:kerRPerm}
        \end{theorem}
        The proof of Theorem \ref{thm:kerRPerm} is given in Section \ref{proofI6}. 
        Theorem \ref{thm:kerRPerm} fully characterizes all possible ways that elements of $\ker r$ can permute the vortices. In Theorem~\ref{thm:kerR}, we fully characterize $\ker r \subset \Aut_{LR}(\C)$ in cases where permutation actions of the anyons and vortices uniquely determine the $\Aut$ groups. In Section \ref{subsec:kerR}, we conjecture that in general $\ker r = \Z_2$ (see Conjecture \ref{conjecture:kerRIsActuallyZ2}). In Section \ref{subsec:k>1}, we provide an explicit example of a theory where $k = 2$. 
        
        In Section \ref{subsec:kerR}, we will define a special element $[\alpha_\psi] \in \ker r$, such that $\alpha_\psi(\phi) = \phi \times \psi$ for all $v$-type fermion parity vortices $\phi$. 
        
        In cases where $\ker r = \Z_2 \simeq \{[1], [\alpha_\psi]\}$, the anomaly $[o_2] \in \mathcal{H}^2(G_b,\Z_2)$ can be directly related to the physical origin of $\mathcal{H}^2(G_b,\Z_2)$ in the classification of (3+1)D FSPT phases. The data $n_2 \in C^2(G_b, \Z_2)$ that is used in specifying an FSPT can be understood as decorating each codimension-2 trijunction of $G_b$ domain walls ${\bf g}$, ${\bf h}$, ${\bf gh}$ in the (3+1)D bulk with a $(1+1)D$ Kitaev chain if $n_2({\bf g}, {\bf h})$ is non-trivial. We show how the anomaly $[o_2]$ arising in the problem of extending the group homomorphism $\rho$ to $\widecheck{\rho}$ is related to the presence of Kitaev chains on domain walls in the bulk (3+1)D FSPT. 
        
        Recently, \cite{delmastro2021} has shown that in the case $G_b = \Z_2^{\bf T}$, the $\mathcal{H}^2(\Z_2^{\bf T},\Z_2)$ anomaly corresponds to an anomaly in the action of ${\bf T}^2$ on the torus Hilbert space of the fermionic topological phase. 
        We generalize their results and show how, for $\ker r = \Z_2$, the non-trivial $[o_2]$ that we find implies that the action of the symmetry operators $\widecheck{\rho}_{\bf g}$ for ${\bf g} \in G_b$ get extended on the torus. Letting $\ket{\Psi}_s$ be a state in the Hilbert space of the topological quantum field theory on the torus with a fixed spin structure $s$, we find
        \begin{align}
            \widecheck{\rho}_{\bf g} \widecheck{\rho}_{\bf h}\ket{\Psi}_s = \widecheck{\rho}_{\bf gh}\left(\omega_2({\bf g,h})\right)^F (-1)^{\tilde{o}_2({\bf g}, {\bf h}) \text{Arf}(s)}\ket{\Psi}_s
        \end{align}
        where $\tilde{o}_2$ means we are interpreting $\tilde{o}_2 \in \Z_2 \simeq \{0,1\}$ instead of $o_2 \in \Z_2 \simeq \ker r$. In fact, we show that the element $[\alpha_\psi] \in \ker r$ discussed above allows us to in general change the symmetry action $\widecheck{\rho}_{\bf g} \rightarrow \widecheck{\rho}_{\bf g} \alpha_\psi$, which has the effect of changing the symmetry action on the torus Hilbert space by 
        \begin{align}
        \widecheck{\rho}_{\bf g}\ket{\Psi}_s \rightarrow \widecheck{\rho}_{\bf g} (-1)^{\text{Arf}(s)}\ket{\Psi}_s
        \end{align}
        
        As an example, we show that the semion-fermion theory, $U(1)_2 \times U(1)_{-1}$ which exists at the surface of a $\nu = 2$ fermionic topological superconductor in Class DIII, possesses the above $\mathcal{H}^2(G_b, \Z_2)$ obstruction. 
        
        We note that the $[o_2]$ anomaly defined above can, a priori, depend on which particular modular extension $\nu$ we consider. To highlight this dependence we write $[o_2^{(\nu)}]$. However, we expect that the 't Hooft anomaly of the (2+1)D theory depends only on $\mathcal{C}$ and the symmetry fractionalization data and should be independent of $\nu$. 
        In reviewing (3+1)D FSPTs in Sec. \ref{subsec:fermionSPT}, we will explicitly define a subgroup $\Gamma^2 \subset \H^2(G_b,\Z_2)$, which leads to the natural homomorphism $q_{\Gamma^2}: \H^2(G_b,\Z_2) \rightarrow \H^2(G_b,\Z_2)/\Gamma^2$. In Sec.~\ref{sec:o2independence} we will discuss how $q( [o_2^{(\nu)} ])$ is expected to determine the 't Hooft anomaly of the theory and results regarding its independence of $\nu$. 
        
        Finally, we note that there is still a remnant of the $\mathcal{H}^2$ anomaly if the $\mathcal{H}^1$ anomaly is non-trivial. Specifically, one can always define a relative $\mathcal{H}^2$ anomaly for two theories $\mathcal{C}_1$ and $\mathcal{C}_2$ that possess the same $\mathcal{H}^1$ anomaly. This can be done by considering the theory $\mathcal{C}_{12}$ obtained by stacking $\mathcal{C}_1$ and $\overline{\mathcal{C}_2}$, so that the $\mathcal{H}^1$ anomaly of the stacked theory vanishes, and then computing the resulting $\mathcal{H}^2$ anomaly of $\mathcal{C}_{12}$. The $\mathcal{H}^2$ anomaly of $\mathcal{C}_{12}$ defines the relative $\mathcal{H}^2$ anomaly between $\mathcal{C}_1$ and $\mathcal{C}_2$. This agrees with the (3+1)D fermion SPT classification, where the (3+1)D fermion SPTs form a torsor over $\mathcal{H}^2$ if the $\mathcal{H}^1$ piece is non-trivial. 
        
        \subsubsection{Anomaly cascade: third layer}
        
        The third layer of the anomaly cascade asks whether the symmetry fractionalization data can be extended from the anyons to the fermion parity-gauged theory.
        
        Assume that the first and second layer obstructions vanish, so that the map $[\widecheck{\rho}_{\bf g}]$ defined in Eq.~\ref{eqn:extendedRho} satisfies the appropriate group structure. 
        
        The next task is to determine if the symmetry fractionalization class for $\mathcal{C}$, defined by a set of $\U$ phases $\eta_a : G_b \times G_b \rightarrow \U$ for each anyon $a$, can be extended to the full modular extension $\C$. In \cite{fidkowski2018}, it was shown that there is an anomaly $[o_3] \in \mathcal{H}^3(G_b,\Z_2)$ which quantifies the obstruction to such an extension. The analysis of \cite{fidkowski2018} was restricted to the special case where $G_f = G_b \times \Z_2^f$, with some additional technical assumptions on the symmetry fractionalization on $\mathcal{C}$. Here we provide a completely general discussion which applies to arbitrary group extensions $G_f$ and symmetry fractionalization. 
        
        As an example, we compute the $\mathcal{H}^3(G_b, \Z_2)$ obstruction for the case of the doubled semion-fermion theory $[U(1)_2 \times U(1)_{-1}]^2$, which exists at the surface of a $\nu = 4$ topological superconductor in class DIII, where ${\bf T}^2 = (-1)^F$ so that $G_b = \Z_2^{\bf T}$,
        $G_f = \Z_4^{{\bf T},f}$. We show that this $[o_3]$ class is non-trivial, which matches the expectation from the bulk-boundary correspondence. 
        
        As in the $\H^2$ case, a priori there may be multiple different modular extensions $\C_\nu$ for which we have a valid group homomorphism $[\widecheck{\rho}_{\bf g}]$, and multiple possible choices of $[\widecheck{\rho}_{\bf g}]$. Therefore in general $[o_3]$ depends on $\nu$ and $\widecheck{\rho}$; to highlight this dependence we can write $[o_3^{(\nu, \widecheck{\rho})}]$. 
        
        As in the case of the $\H^2$ layer, we can define the natural homomorphism $q_{\Gamma^3} : \H^3(G_b, \Z_2) \rightarrow \H^3(G_b, \Z_2)/\Gamma^3$, and we expect that the 't Hooft anomaly is determined by $q_{\Gamma^3}( [o_3^{(\nu, \widecheck{\rho})}])$. In particular, since we expect that the 't Hooft anomaly is determined by $\mathcal{C}$ and its symmetry fractionalization data, we expect that $q_{\Gamma^3}( [o_3^{(\nu, \widecheck{\rho})}])$ is independent of the valid choices of $\nu$, $\widecheck{\rho}$. We discuss this expectation in detail in Sec.~\ref{subsec:o3ChoiceDependence}.
        
        Just like the $\mathcal{H}^2$ anomaly, we can always define a relative $\mathcal{H}^3$ anomaly between two theories $\mathcal{C}_1$ and $\mathcal{C}_2$ with identical $\mathcal{H}^1$ and $\mathcal{H}^2$ anomalies. That is, we consider $\mathcal{C}_1$ and $\mathcal{C}_2$ to have identical $\mathcal{H}^1$ anomaly, and vanishing relative $\mathcal{H}^2$ anomaly. Then, we can consider the stacked theory $\mathcal{C}_{12}$, which has vanishing $\mathcal{H}^1$ and $\mathcal{H}^2$ anomaly, and for which we can define an $\mathcal{H}^3$ anomaly. The $\mathcal{H}^3$ anomaly of $\mathcal{C}_{12}$ defines the relative $\mathcal{H}^3$ anomaly between $\mathcal{C}_1$ and $\mathcal{C}_2$. 
        
        \subsubsection{Anomaly cascade: fourth layer}
        
        The fourth layer of the anomaly cascade asks whether it is possible to define a theory of symmetry defects for $G_b$ consistent with the symmetry fractionalization data on the anyons and fermion parity vortices.
        
        If all the preceding obstructions vanish, we can define a fermion parity gauged theory described by the modular extension $\C$, and a notion of symmetry fractionalization on $\C$. Thus we can define a fully bosonic topological phase with a symmetry fractionalization class. We can then refer to the theory of $G_b$-crossed modular categories, for which there is an $\mathcal{H}^4(G_b, U(1))$ obstruction to gauging $G_b$ \cite{ENO2010,barkeshli2019,barkeshli2019rel,Cui2016}, which is now known how to explicitly compute in general \cite{barkeshli2019rel,bulmash2020}. 
        
        One may think of the first three layers of obstructions as determining the mixed anomaly between fermion parity and $G_b$ - they are the obstructions to gauging fermion parity while preserving $G_b$ symmetry. If the mixed anomaly vanishes, then this fourth layer characterizes the remaining pure $G_b$ anomaly.
        
        In general, to define the $[o_4]$ obstruction, there was a choice of $\nu$, $\widecheck{\rho}$, and symmetry fractionalization data $\widecheck{\eta}$. To highlight this dependence we write $[o_4^{(\nu, \widecheck{\rho}, \widecheck{\eta})}]$. 
        As in the case of the second and third layers, we have a map $q_{\Gamma^4} : \H^4(G_b, \Z_2) \rightarrow \H^4(G_b, \Z_2)/\Gamma^4$. Physically $\Gamma^4 \subset \H^4(G_b, \U)$ can be thought of as the subgroup of (3+1)D boson SPTs which become trivial upon introducing fermions with $G_f$ symmetry. We expect then that the 't Hooft anomaly of the theory is $q( [o_4^{(\nu, \widecheck{\rho}, \widecheck{\eta})}])$, and that this is independent of the choices $\nu, \widecheck{\rho}, \widecheck{\eta}$, as summarized in Conjecture \ref{H4conj}. We will in particular prove independence of $q_{\Gamma^4}( [o_4^{(\nu, \widecheck{\rho}, \widecheck{\eta})}])$ under changes of $\widecheck{\eta}$ for fixed $\nu, \widecheck{\rho}$.
        
        Just as in the previous cases, even if a theory possesses a non-trivial $\mathcal{H}^i$ anomaly with $i < 4$, we can always define a relative $\mathcal{H}^4$ anomaly between two theories $\mathcal{C}_1$ and $\mathcal{C}_2$ that have vanishing relative $\mathcal{H}^1$, $\mathcal{H}^2$, and $\mathcal{H}^3$ anomalies.

        \subsubsection{Organization of paper}

        The rest of this paper is structured as follows. In Sec.~\ref{sec:supermodular}, we review some basic facts about super-modular and spin modular categories and the Hilbert space of a spin modular theory on a torus. In Sec.~\ref{sec:bosonicSymmFrac}, we give a brief review of symmetry fractionalization in bosonic systems. In Sec.~\ref{sec:fermionicSymm} we summarize symmetries and symmetry fractionalization in fermionic systems and also review in detail the classification of (3+1)D FSPTs due to \cite{WangGu}. We then consider each level of the anomaly cascade in order. We prove Theorem~\ref{thm:H1} about the $\H^1(G_b,\Z_{\bf T})$ obstruction in Sec.~\ref{sec:H1}. Sec.~\ref{sec:H2} is devoted to a thorough discussion of the $\H^2(G_b,\ker r)$ obstruction and its relation to the $\H^2(G_b,\Z_2)$ part of the 't Hooft anomaly. In Sec.~\ref{sec:H3}, we give a fully general discussion of the $\H^3(G_b,\Z_2)$ obstruction. In Sec. ~\ref{H4sec} we discuss the $\H^4(G_b, \U)$ obstruction, in particular its dependence on the various choices made in its definition. Sec.~\ref{sec:examples} contains additional interesting examples, and some summary and discussion appears in Sec.~\ref{sec:discussion}.

    \section{Super-modular and spin modular categories and fermionic topological phases of matter}
    \label{sec:supermodular}
    
        There are two equivalent descriptions of a fermionic topological phase of matter. One description uses a super-modular tensor category, denoted $\mathcal{C}$, along with a chiral central charge $c_-$. In this description, the super-modular tensor category determines $c_- \mod 1/2$. A super-modular tensor category~\cite{bruillard2017a,bruillard2017b,bonderson2018} is a unitary braided fusion category with exactly one nontrivial invisible particle $\psi$, which is a fermion and satisfies $\Z_2$ fusion rules. ``Invisible" means that the double braid
        \begin{equation}
            M_{a,\psi} = +1
        \end{equation}
        for all $a \in \mathcal{C}$, ``fermion" means that the topological twist $\theta_\psi=-1$, and the $\Z_2$ fusion rules means $\psi \times \psi = 1$. Physically, this description tracks the topologically non-trivial quasi-particle content of the phase. The presence of $\psi$ is used to track the presence of a fermion which is topologically trivial in the sense that the fermion can be created or annihilated by a local fermion operator. Different fermionic phases with the same quasi-particle content are distinguished by stacking with $p+ip$ superconductors, each of which changes $c_-$ by 1/2.
        
        The alternate description is via a spin modular category, denoted $\C$, and a chiral central charge $c_-$ which is determined modulo $8$ by $\C$. A spin modular category is a UMTC together with a preferred choice of fermion $\psi$ which has $\Z_2$ fusion rules. Physically, the spin modular category describes the phase after gauging fermion parity, that is, it describes the quasi-particle content of the phase along with the fermion parity vortices. More precisely, $\C$ possesses a natural $\Z_2$ grading determined by double braids with $\psi$:
        \begin{align}
            \C &= \C_0 \oplus \C_1\\
            \C_0 &\simeq \mathcal{C}
        \end{align}
        where $M_{a,\psi}=+1$ if $a \in \C_0$ and $M_{a,\psi}=-1$ if $a \in \C_1$. Fusion respects this grading. The objects in $\C_1$ are physically interpreted as fermion parity vortices, or equivalently symmetry defects of fermion parity symmetry.
        
        The sector $\C_1$ can be further decomposed as follows:
        \begin{equation}
            \C_1 = \C_v \oplus \C_\sigma
        \end{equation}
        where
        \begin{align}
            \C_v &= \{a \in \C_1 | a \times \psi \neq a\}\\
            \C_\sigma &= \{a \in \C_1 | a \times \psi = a\} .
        \end{align}
        
        The two descriptions are related as follows. Every super-modular category $\mathcal{C}$ admits~\cite{johnsonFreydModularExt} a minimal modular extension $\C$, which means that $\C$ contains $\mathcal{C}$ as a subcategory, the preferred fermion of $\C$ is the invisible fermion of $\mathcal{C}$, and $\C$ has minimal possible total quantum dimension
        \begin{equation}
            \mathcal{D}_{\C}^2 = 2\mathcal{D}_{\mathcal{C}}^2.
        \end{equation}
        Every super-modular category has precisely 16 distinct minimal modular extensions $\C_{\nu}$, for $\nu=0,\ldots,15$, and whose chiral central charges modulo 8 differ by $\nu/2$. A fermionic topological phase is described either by $(\mathcal{C},c_-)$ or by $(\C_{\nu},c_-)$ where the chiral central charge of $\C_{\nu}$ is equal to $c_-$ modulo 8.
        
        We will make use of the following fact~\cite{bruillard2017b} about the topological $S$-matrix of a minimal modular extension $\C$:
        \begin{equation}
            S_{x,y} = \begin{cases}
              S_{x \times \psi, y} & y \in \C_0\\
              -S_{x \times \psi, y} & y \in \C_1
            \end{cases}
            \label{eqn:SChangeByPsi}
        \end{equation}
        Note that the second line implies that $S_{\sigma,y} = 0$ for $\sigma \in \C_\sigma$, $y \in \C_1$.
        
        As we will see now, the full spin modular category, not just the super-modular tensor category, is required to describe the fermionic system on non-trivial surfaces and with arbitrary spin structures.
    
        \subsection{Torus degeneracy and spin structure}
        \label{subsec:torusDegeneracy}
        
        The spin modular category determines the Hilbert space and the action of the mapping class group of the fermionic topological phase of matter on a topologically non-trivial surface. 
        
        Consider the (spatial) torus $T^2$ with nontrivial cycles $\alpha$ and $\beta$. Let us label the states of the spin modular category as $|a \rangle_\alpha$, for $a \in \C$ an anyon. This means that the topological charge as \textit{measured} through the loop $\alpha$ is $a$. Below for ease of notation we drop the subscript $\alpha$ and keep it implicit in the definition of the state $|a \rangle$. 
        
        Now suppose we have a choice of spin structure $(\mu_{\alpha}, \nu_{\beta})$ where $\mu = 0$ corresponds to Neveu-Schwarz (anti-periodic) boundary conditions on the loop $\alpha$ and $\mu = 1$ corresponds to Ramond (periodic) boundary conditions on $\alpha$, with $\nu_{\beta}$ similar for the $\beta$ loop\footnote{Our convention is slightly unusual because the modular $T$ transformation does not act linearly on the $\mu$ and $\nu$ indices. However, the formulas relevant to us like Eq.~\ref{eqn:WpsiEigenvalues} are more natural.}. We drop the $\alpha$ and $\beta$ labels in what follows. Let $\H_{\mu, \nu}$ denote the Hilbert space of the fermionic topological phase on the torus with the chosen spin structure, $|\Psi\rangle_{\mu, \nu} \in \H_{\mu,\nu}$, and let $W_x(\gamma)$ denote the Wilson loop of a particle $x$ around the loop $\gamma$.
        
        The defining distinction between the different sectors $(\mu, \nu)$ is in the fermion boundary conditions, which defines the eigenvalue of the fermion Wilson loop $W_\psi$:
        \begin{align}
            W_\psi(\alpha) | \Psi \rangle_{\mu, \nu} &= (-1)^{\mu}  | \Psi \rangle_{\mu,\nu}
            \nonumber \\
            W_\psi(\beta) | \Psi \rangle _{\mu,\nu}&= (-1)^{\nu} | \Psi\rangle_{\mu,\nu}. \label{eqn:WpsiEigenvalues}
        \end{align}
        Note that $W_\psi(\alpha)$ has $-1$ eigenvalue for the sector with periodic $(\mu = 1)$ boundary conditions along the $\alpha$ cycle. This is because periodic boundary conditions occur when there is a fermion parity vortex threading the conjugate cycle, as can be derived by studying the modular matrices carefully. The $-1$ then arises due to the mutual statistics between $\psi$ and $v$. 
        
        A basis of states for the fermionic topological phase on a torus is as follows (see e.g.~\cite{delmastro2021} for a recent discussion):
        \begin{align}
            \H_{0,0}:  &\ket{a}_{00} = \frac{1}{\sqrt{2}} \left( |a \rangle + | a \times \psi \rangle \right) 
            \nonumber \\
            \mathcal{H}_{1,0}:  &
            \begin{cases}
              \ket{v}_{10}=\frac{1}{\sqrt{2}} ( |v \rangle + |v \times \psi \rangle )  \\    
              \ket{\sigma}_{10}=|\sigma \rangle 
            \end{cases}
            \nonumber \\
            \mathcal{H}_{0,1}: & \ket{a}_{01}=\frac{1}{\sqrt{2}} \left( |a \rangle - | a \times \psi \rangle \right) 
            \nonumber \\
            \mathcal{H}_{1,1}: &
            \begin{cases}
              \ket{v}_{11}=\frac{1}{\sqrt{2}} ( |v \rangle - |v \times \psi \rangle )  \\    
              \ket{\sigma}_{11}=|\sigma; \psi \rangle 
            \end{cases},
            \label{eqn:torusBasis}
        \end{align}
        where $a \in \C_0$, $v \in \C_v$, and $\sigma \in \C_\sigma$. Furthermore, $|\sigma; \psi\rangle$ denotes the state on a torus with a puncture labeled $\psi$. 
        
        These states on the torus can be built by gluing together states on the 3-punctured sphere, i.e. splitting spaces $V_c^{ab}$ and their duals $V_{ab}^c$. Eq.~\ref{eqn:torusBasis} can be re-expressed in this language as
        \begin{align}
            \H_{00} &= \left(\bigoplus_{a \in \C_0} V_{1}^{a \overline{a}} \otimes V_{a \overline{a}}^1\right)\bigg|_{\text{symm}} \nonumber \\
            \H_{01} &= \left(\bigoplus_{a \in \C_0} V_{1}^{a \overline{a}} \otimes V_{a \overline{a}}^1\right)\bigg|_{\text{anti-symm}} \nonumber\\
            \H_{1,0} &= \left(\bigoplus_{v \in \C_v} V_{1}^{v \overline{v}} \otimes V_{v \overline{v}}^1\right)\bigg|_{\text{symm}} \oplus \bigoplus_{\sigma \in \C_\sigma} \left(V_{1}^{\sigma \overline{\sigma}} \otimes V_{\sigma\overline{\sigma}}^{1}\right) \nonumber \\
            \H_{1,1}&= \left(\bigoplus_{v \in \C_v} V_{1}^{v \overline{v}} \otimes V_{v \overline{v}}^1\right)\bigg|_{\text{anti-symm}} \oplus \bigoplus_{\sigma \in \C_\sigma} \left(V_{1}^{\sigma \overline{\sigma}} \otimes V_{\sigma\overline{\sigma}}^{\psi} \right)
            \label{eqn:torusStates}
        \end{align}
        The ``(anti)-symm" notation means we restrict to the subspace which consists of (anti)-symmetric sums of states in the $a$ and $a \times \psi$ sectors.

    \section{Review of symmetry fractionalization in bosonic systems}
        \label{sec:bosonicSymmFrac}
        
        We briefly review the formalism for symmetry fractionalization in bosonic systems.
        
        Consider a UMTC $\mathcal{B}$ with global symmetry group $G$. The basic data required to define symmetry fractionalization is the following. First, we define a group homomorphism
        \begin{equation}
            [\rho_{\bf g}]: G \rightarrow \Aut(\mathcal{B})
        \end{equation}
        where $\Aut(\mathcal{B})$ is the group of braided autoequivalences of $\mathcal{B}$, modulo a set of gauge equivalences called natural isomorphisms\footnote{We define natural isomorphisms using Eq.~\ref{eqn:natIso}. There exists a more abstract mathematical definition, and it is unclear if this definition is equivalent to ours. In considering equivalence classes of braided autoequivalences, we will restrict our attention to natural isomorphisms in our definition, since this appears to naturally describe SETs \cite{barkeshli2019}.}. A braided autoequivalence, or autoequivalence for short, is a map from $\mathcal{B}$ to itself which preserves the data of the theory up to a gauge transformation. We will also use the term ``autoequivalence" to refer to braided anti-autoequivalences, which, up to a gauge transformation, complex conjugate the data of the theory. A natural isomorphism $\Upsilon$ is an autoequivalence which acts on fusion vertices as
        \begin{equation}
            \Upsilon\left(\ket{a,b;c;\mu}\right) = \frac{\gamma_a \gamma_b}{\gamma_c}\ket{a,b;c;\mu}
            \label{eqn:natIso}
        \end{equation}
        where $\gamma_a \in \U$. Natural isomorphisms have a redundancy
        \begin{equation}
            \gamma_a \rightarrow \gamma_a \zeta_a
            \label{eqn:zetaRedundancy}
        \end{equation}
        where $\zeta_a \in \U$ obeys the fusion rules in the sense that $\zeta_a \zeta_b = \zeta_c$ whenever $N_{ab}^c > 0$. A representative $\rho_{\bf g}$ of the equivalence class $[\rho_{\bf g}]$ determines a permutation of the anyons $a \rightarrow \,^{\bf g}a$ and a set of unitary matrices $U_{\bf g}(a,b;c)$ as follows:
        \begin{equation}
            \rho_{\bf g}\left(\ket{a,b;c;\mu}\right)=\sum_{\nu} U_{\bf g}(\,^{\bf g}a,\,^{\bf g}b;\,^{\bf g}c)_{\mu \nu}\ket{\,^{\bf g}a,\,^{\bf g}b;\,^{\bf g}c;\nu}.
            \label{eqn:Udef}
        \end{equation}
        The statement that $\rho_{\bf g}$ is an autoequivalence means that the $F$- and $R$-symbols are preserved (up to complex conjugation) according to the following consistency conditions:
        \begin{align}
            \left(\left[F^{abc}_{d}\right]_{(e,\alpha,\beta),(f,\mu,\nu)}\right)^{\sigma({\bf g})} &= \sum_{\alpha',\beta',\mu',\nu'} U_{\bf g}(\,^{\bf g}a, \,^{\bf g}b;\,^{\bf g}e)_{\alpha \alpha'} U_{\bf g}(\,^{\bf g}e, \,^{\bf g}c;\,^{\bf g}d)_{\beta \beta'}\left[F^{\,^{\bf g}a\,^{\bf g}b\,^{\bf g}c}_{\,^{\bf g}d}\right]_{(\,^{\bf g}e,\alpha',\beta'),(\,^{\bf g}f,\mu',\nu')} \times \nonumber  \label{eqn:UFConsistency} \\
            &\hspace{1in}\times \left(U_{\bf g}(\,^{\bf g}b, \,^{\bf g}c;\,^{\bf g}f\right)^{-1}_{\mu' \mu} \left(U_{\bf g}(\,^{\bf g}a, \,^{\bf g}f;\,^{\bf g}d\right)^{-1}_{\nu' \nu} \\
            \left(\left[R^{ab}_c\right]_{\mu \nu}\right)^{\sigma({\bf g})} &= U_{\bf g}(\,^{\bf g}b, \,^{\bf g}a; \,^{\bf g}c)_{\mu \mu'} \left(R^{\,^{\bf g}a\,^{\bf g}b}_{\,^{\bf g}c}\right)_{\mu' \nu'}\left[U_{\bf g}(\,^{\bf g}a,\,^{\bf g}b;\,^{\bf g}c)^{-1}\right]_{\nu' \nu} \label{eqn:URConsistency}
        \end{align}
        Here
        \begin{equation}
            \sigma({\bf g}) = \begin{cases}
              1 & {\bf g} \text{ unitary}\\
              \ast & {\bf g} \text{ anti-unitary}
            \end{cases}.
        \end{equation}
        We will also use
        \begin{equation}
            s_1({\bf g}) = \begin{cases}
              0 & {\bf g} \text{ unitary}\\
              1 & {\bf g} \text{ anti-unitary} 
            \end{cases}
            \label{eqn:s1Def}
        \end{equation}
        The map $s_1$ is a group homomorphism, i.e. $s_1 \in Z^1(G_b,\Z_2)$. 
        
        The maps $\rho_{\bf g}$ define the natural isomorphisms
        \begin{equation}
            \kappa_{\bf g,h} = \rho_{\bf gh}\rho_{\bf h}^{-1}\rho_{\bf g}^{-1}
        \end{equation}
        which have actions on fusion vertices given by
        \begin{equation}
            \kappa_{\bf g,h}\ket{a,b;c;\mu}=\kappa_{\bf g,h}(a,b;c)\ket{a,b;c; \mu}.
        \end{equation}
        Here
        \begin{equation}
            \kappa_{\bf g,h}(a,b;c) = \frac{\beta_a({\bf g,h})\beta_b({\bf g,h})}{\beta_c{\bf g,h}} =U_{\bf g}^{-1}(a,b;c)U_{\bf h}(\,^{\overline{\bf g}}a,\,^{\overline{\bf g}}b ; \,^{\overline{\bf g}}c)^{-\sigma({\bf g})}U_{\bf gh}(a,b;c)
        \end{equation}
        and $\beta_a({\bf g,h})$ are phases defining $\kappa_{\bf g,h}$ as a natural isomorphism. The $\beta_a$ define the phases
        \begin{equation}
            \Omega_a({\bf g,h,k}) = \frac{\beta^{\sigma({\bf g})}_{\,^{\overline{\bf g}}a}({\bf h,k}) \beta_a({\bf g,hk})}{\beta_a({\bf g,h})\beta_a({\bf gh,k})}
        \end{equation}
        which can be shown to obey the fusion rules, in the sense
        \begin{equation}
            \Omega_a \Omega_b = \Omega_c \text{ whenever } N_{ab}^c>0.
        \end{equation}
        These phases define, for modular $\mathcal{B}$, an obstruction $[\coho{O}]\in \H^3(G,\A)$ to localizing $G$ on the anyons, where $\A \subset \mathcal{B}$ is the group of Abelian anyons of the theory; see standard references, e.g.~\cite{barkeshli2019}, for further details.
        
        If said obstruction vanishes, then one can define symmetry fractionalization on $\mathcal{B}$, which amounts to a choice of phases $\eta_a({\bf g,h}) \in \U$ which satisfy the following consistency conditions:
        \begin{align}
            \eta_{\,^{\overline{\bf g}}a}({\bf h,k})^{\sigma({\bf g})}\eta_a({\bf g,hk})&= \eta_a({\bf g,h})\eta_a({\bf gh,k}) \label{eqn:etaConsistency}\\
            \frac{\eta_c({\bf k,l})}{\eta_a({\bf k,l})\eta_b({\bf k,l})} U_{\bf kl}(a,b;c)_{\mu \nu} &= \sum_\lambda U_{\bf l}(\,^{\overline{\bf k}}a,\,^{\overline{\bf k}}b;\,^{\overline{\bf k}}c)_{\mu \lambda}U_{\bf k}(a,b;c)_{\lambda \nu} \label{eqn:UEtaConsistency}
        \end{align}
        
        Symmetry fractionalization can equivalently be specified by a set of phases $\omega_a({\bf g,h})$ which obey the fusion rules and satisfy
        \begin{equation}
            \Omega_a({\bf g,h,k}) = \frac{\omega^{\sigma({\bf g})}_{\,^{\overline{\bf g}}a}({\bf h,k}) \omega_a({\bf g,hk})}{\omega_a({\bf g,h})\omega_a({\bf gh,k})} \equiv (d\omega)_a({\bf g,h,k})
            \label{eqn:OmegaEqualsDomega}.
        \end{equation}
        The relationship between the $\eta_a$ and $\omega_a$ descriptions is
        \begin{equation}
            \eta_a({\bf g,h}) = \frac{\beta_a({\bf g,h})}{\omega_a({\bf g,h})}.
            \label{eqn:etaDef}
        \end{equation}
        
        In all of the above, we have taken a fixed representative $\rho_{\bf g}$ of the class $[\rho_{\bf g}]$. If $\rho_{\bf g}$ is modified by a natural isomorphism given by the phases $\gamma_a({\bf g})$, then we obtain gauge-equivalent data
        \begin{align}
            U'_{\bf g}(a,b;c)_{\mu \nu} &= \frac{\gamma_a({\bf g}) \gamma_b({\bf g})}{\gamma_c({\bf g})}U_{\bf g}(a,b;c)_{\mu \nu}\\
            \eta'_a({\bf g,h}) &= \frac{\gamma_a({\bf gh})}{\left[\gamma_{\,^{\overline{\bf g}}a}({\bf h})\right]^{\sigma({\bf g})}\gamma_a({\bf g})}\eta_a({\bf g,h}).
        \end{align}
        Separately, there is gauge freedom
        \begin{align}
            \beta_a({\bf g,h}) & \rightarrow \beta_a({\bf g,h})\nu_a({\bf g,h}) \nonumber \\
            \Omega_a({\bf g,h,k}) &\rightarrow \Omega_a({\bf g,h,k}) (d\nu)_a({\bf g,h,k}) \nonumber \\
            \omega_a({\bf g,h}) &\rightarrow \omega_a({\bf g,h})\nu_a({\bf g,h})
            \label{eqn:OmegaomegaGaugeFreedom}
        \end{align}
        where $\nu_a({\bf g,h})$ is a phase obeying the fusion rules.
        
        For a fixed map $[\rho_{\bf g}]$, the set of symmetry fractionalization patterns form a torsor over $\H^2(G,\A)$. Specifically, given consistent symmetry fractionalization data $\eta_a({\bf g,h})$ and an element $\coho{t}({\bf g,h})\in Z^2(G,\A)$, then one obtains a new symmetry fractionalization pattern
        \begin{equation}
            \widehat{\eta}_a({\bf g,h})=M_{a,\cohosub{t}({\bf g,h})}\eta_a({\bf g,h}).
        \end{equation}
        One can check that up to gauge transformations, $\widehat{\eta}$ depends only on the cohomology class $[\coho{t}]\in \H^2(G,\A)$ and that different cohomology classes produce gauge-inequivalent symmetry fractionalization patterns.

    \section{Symmetries in fermionic topological phases}
    \label{sec:fermionicSymm}
        
            In this section, we briefly review fermionic symmetries, the classification of (3+1)D fermionic SPTs, and the results of~\cite{bulmash2021} on fermionic symmetry fractionalization in fermionic topological phases, and also define a map $\widecheck{\Upsilon}_\psi$ in Eq.~\ref{eqn:checkUpsilonPsiDef} which will be useful later.\footnote{The results on fermionic symmetry fractionalization reviewed here also appeared in~\cite{aasen21ferm}, the first version of which appeared on the arXiv at the same time as both the first version of this paper and \cite{bulmash2021}. } 
        
        We assume that we are describing a system whose Hilbert space decomposes into a tensor product of local Hilbert spaces which include fermionic degrees of freedom. Further, we assume the dynamics of the system are given by a local Hamiltonian with an energy gap such that the system is in the fermionic topological phase associated to $\mathcal{C}$. In order for the formalism to describe anomalous fermionic SETs, we use the term ``symmetry fractionalization of a fermionic topological phase" to refer only to fractionalization data on the super-modular category $\mathcal{C}$. 
        
        \subsection{Fermionic symmetries}
        \label{subsec:fermionicSymm}
        
        Fermionic systems always have a special symmetry, fermion parity symmetry $(-1)^F$, which generates a central $\Z_2$ subgroup $\Z_2^f$ of the full symmetry group $G_f$. Define the ``bosonic" symmetry group $G_b = G_f / \Z_2^f$; then the symmetry generators restricted to their action on bosonic operators form a representation of $G_b$. An alternate characterization of $G_f$ is as a $\Z_2^f$ central extension of $G_b$ via the short exact sequence
        \begin{equation}
            1 \rightarrow \Z_2^f \rightarrow G_f \rightarrow G_b \rightarrow 1
        \end{equation}
        and a cocycle $\omega_2 \in Z^2(G_b,\Z_2)$.
        
        We will reserve the notation $\omega_2$ for a cocycle valued in $\{\pm 1\}\simeq \Z_2$ and use the notation $\tilde{\omega}_2$ for the $\{0,1\}$-valued additive parameterization of $\omega_2$, i.e. define
        \begin{equation}
            \omega_2({\bf g,h})=(-1)^{\tilde{\omega}_2({\bf g,h})}.
            \label{eqn:w2Def}
        \end{equation}
        
        Physically, we can write a set of operators $R_{\bf g}$ for ${\bf g} \in G_b$ which implement the action of $G_b$ on the microscopic Hilbert space. Since the symmetry group on the full fermionic Hilbert space is actually $G_f$, these operators multiply projectively,
        \begin{equation}
            R_{\bf g}R_{\bf h} = \left(\omega_2({\bf g,h})\right)^F R_{\bf gh}.
            \label{eqn:projectiveRg}
        \end{equation}
        
        The cocycle $\omega_2 \in Z^2(G_b,\Z_2)$ determines a decomposition $G_f = G_b\times \Z_2^f$ as sets, but all choices of cocycle representative in the same cohomology class $[\omega_2]$ lead to isomorphic groups $G_f$. Physically, modifying $\omega_2 \rightarrow \omega_2 \times d\phi$ for $\phi \in C^1(G_b,\Z_2)$ changes
        \begin{equation}
            R_{\bf g} \rightarrow (\phi({\bf g}))^F R_{\bf g},
        \end{equation}
        where we are taking $\phi \in \{\pm 1\}$.
        Such a transformation changes the physical meaning of $R_{\bf g}$; for example, if $G_b = \Z_2^{\bf T}$, the time-reversal operator ${\bf T}$ is physically distinct from $(-1)^F {\bf T}$, and these operators should not be interchanged. Therefore, a microscopic realization of $G_f$ symmetry will in general specify a cocycle representative $\omega_2$, not just its cohomology class $[\omega_2]$.
        
        \subsection{General classification of (3+1)D fermionic SPTs and 't Hooft anomalies in (2+1)D fermionic systems}
        \label{subsec:fermionSPT}
        
        The group of (3+1)D fermionic SPTs, which also defines 't Hooft anomalies for (2+1)D fermionic topological phases, has been classified in recent years through a variety of approaches. One general approach is through the cobordism classification \cite{kapustin2015Cobordism,kapustin2017}, where (3+1)D fermionic SPTs are classified using bordism groups of 4-manifolds equipped with $G_b$ gauge fields and a generalized spin structure. 
        
        An alternative general classification method was recently provided by Wang and Gu \cite{WangGu}, based on decorating symmetry defects of varying codimension with lower dimensional fermionic invertible topological phases. Below we summarize the Wang-Gu consistency equations and equivalence relations. We note that our treatment of the equivalence relations below differs slightly from Wang-Gu, as we will explain. 
        
        The calculation of the bordism groups mentioned above can be performed by various spectral sequence methods. These methods typically give the resulting FSPT classification in terms of a group extension involving subgroups of $\mathcal{H}^1(G_b, \Z_{\bf T})$, $\mathcal{H}^2(G_b, \Z_2)$, $\mathcal{H}^3(G_b, \Z_2)$, and $\mathcal{H}^4(G_b, \U)$. The Wang-Gu classification provides a partial solution to the above spectral sequence computation for general groups $G_b$.\footnote{It is a partial solution because Wang-Gu do not derive the group multiplication law for stacking FSPT phases, and thus do not fully solve the group extension problem.} 
        
        According to Wang and Gu, (3+1)D fermionic SPTs can be specified by four layers of data,
        \begin{align}
            (n_1, n_2, n_3, \nu_4) \in Z^1(G_b, \Z_{\bf T}) \times C^2(G_b, \Z_2) \times C^3(G_b, \Z_2) \times C^4(G_b, U(1)) . 
        \end{align}
        We will refer to the subscript of $n_i$ and $\nu_4$ as the layer index. 
        
        This set of data is subject to consistency conditions:
        \begin{align}
            dn_2 &= \tilde{\omega}_2 \stdcup n_1 + s_1 \stdcup n_1 \stdcup n_1
            \nonumber \\
            dn_3 &= \tilde{\omega}_2 \stdcup n_2 + n_2 \stdcup n_2 + s_1 \stdcup (n_2 \stdcup_1 n_2) . 
            \nonumber \\
            d\nu_4 &= \mathcal{O}_5[n_3]. 
        \end{align}
        with $s_1$ defined in Eq.~\ref{eqn:s1Def}.
        The precise formula for $\mathcal{O}_5$ is unimportant for our purposes and can be found in \cite{WangGu}. Physically, the different layers of data correspond to decorating defects of various codimension with lower-dimensional invertible fermionic topological phases; for example, $n_1$ corresponds to decorating time-reversal domain walls with a $c_- = 1/2$ invertible fermionic phase (e.g. a $p_x + ip_y$ superconductor); $n_2$ corresponds to decorating codimension-2 junctions with Kitaev-Majorana chains; and $n_3$ corresponds to decorating codimension-3 junctions with fermions. 
        
        The main fact about these equations that is of relevance to us in this paper is that data of the $i$th layer is an $i$-cocycle if the data of the lower layers vanish. That is, if $n_1 = 0$, then $n_2 \in Z^2(G_b, \Z_2)$; if $n_1,n_2 = 0$, then $n_3 \in Z^3(G_b, \Z_2)$, and if $n_1,n_2,n_3 = 0$, then $\nu_4 \in Z^4(G_b, U(1))$.

        In addition to the above consistency equations, the data above is also subject to a number of equivalence relations. Each layer can change by a coboundary, and also an additional equivalence, as we explain. \footnote{We note that in the equivalences summarized here for a given layer, we allow the possibility for the higher layer data to change as well. The equivalence relations of Ref. \cite{WangGu} imply that the higher layer data does not change under these equivalences, however Ref. \cite{manjunath21inv} found that in (2+1)D the higher layer data does change under equivalences of a given layer. We expect a similar phenomenon to generally occur in (3+1)D as well.} 
        
        \subsubsection{1st layer equivalence}
        
        We can change $n_1$ by a coboundary:
        \begin{align}
            (n_1, n_2, n_3, \nu_4) \simeq (n_1 + d b_0, n_2', n_3', \nu_4') ,
        \end{align}
        where $b_0 \in C^0(G_b, \Z_{\bf T})$. Note that when $n_1$ changes by a coboundary, the higher layer data may in principle also change. 
        
        If we ignore the higher layer data, then, we see that each (3+1)D FSPT determines an element 
        \begin{align}
            [n_1] \in \mathcal{H}^1(G_b, \Z_{\bf T} ). 
        \end{align}
        
        \subsubsection{2nd layer equivalence and \texorpdfstring{$\Gamma^2$}{Gamma2}}
        
        There are the following 2nd layer equivalences:
        \begin{align}
            (n_1, n_2, n_3, \nu_4) & \simeq (n_1, n_2 + \tilde{\omega}_2, n_3', \nu_4')\\
            &\simeq (n_1, n_2 + db_1, n_3'', \nu_4'')  \text{ if } G_b \text{ unitary},
         \end{align}
        where $b_1 \in C^1(G_b, \Z_2)$ and
        recalling $\omega_2 = (-1)^{\tilde{\omega}_2}$. Note that under changing the $n_2$, the higher level data $n_3$ and $\nu_4$ can change to some other consistent data $n_3', \nu_4'$ or $n_3'', \nu_4''$, whose precise form has not been computed and is not relevant for our purposes. 
        
        It is useful to denote
        \begin{equation}
            \Gamma^2 = \begin{cases}
              \{1,[\omega_2]\} & G_b \text{ unitary}\\
              \{1\} & G_b \text{ contains anti-unitary symmetries}
            \end{cases}
        \end{equation}
        where $\Gamma^2 \subset \mathcal{H}^2(G_b, \Z_2)$. We also define the group homomorphism
        \begin{align}
            q_{\Gamma^2} : \mathcal{H}^2(G_b, \Z_2) \rightarrow \mathcal{H}^2(G_b, \Z_2) / \Gamma^2,
        \end{align}
        which we will use later. 
        
        The implication of the above equivalence is that if $n_1 = 0$ and we forget about $n_3, \nu_4$, we can define an equivalence class
        \begin{align}
        [n_2] \in \mathcal{H}^2(G_b, \Z_2)/\Gamma^2.
        \end{align}
        
        In other words, any (3+1)D FSPT with $n_1 = 0$ defines a class $[n_2] \in \mathcal{H}^2(G_b, \Z_2)/\Gamma^2$. 
        
        \subsubsection{3rd layer equivalence and \texorpdfstring{$\Gamma^3$}{Gamma3}}
        
        There are the following 3rd layer equivalences:
        \begin{align}
            (n_1, n_2, n_3, \nu_4) & \simeq (n_1, n_2, n_3 + \chi_3, \nu_4')
            \simeq (n_1, n_2, n_3 + db_2, \nu_4'') ,
        \end{align}
        where $b_2 \in C^2(G_b, \Z_2)$. 
        Here $\chi_3 \in Z^3(G_b, \Z_2)$ consists of 3-cocycles that satisfy 
        \begin{align}
        \label{Gamma3Eq}
            \chi_3 = \tilde{\omega}_2 \stdcup \lambda_1 + s_1 \stdcup \lambda_1 \stdcup \lambda_1 + \left\lfloor \lambda_0/2 \right\rfloor  ( \tilde{\omega}_2 \stdcup_1 \tilde{\omega}_2) \;\;\;(\text{mod } 2),
        \end{align}
        for any choice of $\lambda_1 \in Z^1(G_b, \Z_2)$ and $\lambda_0 \in \Z$. 
        
        Formally, we can define the group $\Gamma^3 \subset \mathcal{H}^3(G_b, \Z_2)$, where all representative $3$-cocycles have the form $\chi_3$ given in Eq.~\ref{Gamma3Eq}. It is useful to then define the group homomorphism
        \begin{align}
            q_{\Gamma^3} : \mathcal{H}^3(G_b, \Z_2) \rightarrow \mathcal{H}^3(G_b,\Z_2)/\Gamma^3. 
        \end{align}
        
        Similar to the case of the 2nd layer, the implication of the above equivalence is that if $n_1, n_2 = 0$ and we ignore $\nu_4$, then  we can define an equivalence class 
        \begin{align}
        [n_3] \in \mathcal{H}^3(G_b, \Z_2)/\Gamma^3 .
        \end{align}
        In other words, any (3+1)D FSPT with $n_1, n_2 = 0$ defines a class $[n_3] \in \mathcal{H}^3(G_b, \Z_2)/\Gamma^3$. 
        
        \subsubsection{4th layer equivalence and \texorpdfstring{$\Gamma^4$}{Gamma4}}

        There exist the following 4th layer equivalences:        
        \begin{align}
        (n_1, n_2, n_3, \nu_4) & \simeq (n_1, n_2, n_3, \nu_4 \chi_4) \simeq (n_1, n_2, n_3, \nu_4 d \epsilon_3)
        \end{align}
        for $\epsilon_3 \in C^3(G_b, \Z_2)$. 
        Here $\chi_4$ is any 4-cocycle in a group $\Gamma^4$ with a rather involved definition that can be found in~\cite{WangGu} but is unimportant for our present purposes.  It will suffice to know that $\Gamma^4$ contains a subgroup of the form
        \begin{equation}
            \{[\lambda_2] \stdcup [\lambda_2] + [\omega_2] \stdcup [\lambda_2] \; | \; [\lambda_2] \in \H^2(G_b,\Z_2)\} \subset \Gamma^4
            \label{eqn:Gamma4Subgroup}
        \end{equation}
        We define the group homomorphism
        \begin{align}
            q_{\Gamma^4}: \mathcal{H}^4(G_b, \U) \rightarrow \mathcal{H}^4(G_b, \U)/\Gamma^4. 
            \label{eqn:qGamma4}
        \end{align}
        
        Physically, $\Gamma^4$ is the group of bosonic (3+1)D SPTs with $G_b$ symmetry that are trivial when viewed as a fermionic SPT with $G_f$ symmetry. Equivalently, $\Gamma^4$ is the group of ``anomalous" fermionic (2+1)D SPTs; these have the property that the surface of a bosonic (3+1)D SPT characterized by a 4-cocycle in $\Gamma^4$ can have a topologically trivial gapped symmetric gapped (2+1)D surface, if fermions transforming under $G_f$ symmetry are introduced to the surface.  
        
        We see that any (3+1)D FSPT for which $n_1, n_2, n_3 = 0$ defines an element 
        \begin{align}    
        [\nu_4] \in \mathcal{H}^4(G_b, \U) / \Gamma^4. 
        \end{align}
        
        \subsection{Fermionic symmetry fractionalization}
        \label{subsec:fermionicSymmFrac}
        
        The first step of defining $G_f$ symmetry fractionalization is to assign an autoequivalence $[\rho_{\bf g}]$ of $\mathcal{C}$ to each element ${\bf g} \in G_b$. Topological autoequivalences are well-defined for any BFC regardless of modularity, so autoequivalences of a super-modular category $\mathcal{C}$ are also well-defined. This assignment defines the $U$-symbols as given in Eq.~\ref{eqn:Udef}.
        
        There is, however, a physical constraint on the choice of autoequivalence. The symmetry operator $R_{\bf g}$ is defined on the physical, microscopic Hilbert space, and we will be seeking to localize $R_{\bf g}$. Then if $\gamma_{i,{\bf r}}$ is a basis of (Majorana) fermion operators at position ${\bf r}$, the Hilbert space defines matrices $\tilde{U}_{ij}({\bf g},{\bf r})$ such that
        \begin{equation}
            R_{\bf g}\gamma_{i,{\bf r}}R_{\bf g}^{-1}= \sum_j \tilde{U}_{ij}({\bf g, r}) \gamma_{j,{\bf r}},
            \label{eqn:localFermionTransform}
        \end{equation}
        where $i$ and $j$ label elements of the basis of fermionic operators.
        
        As shown in~\cite{bulmash2021}, compatibility of a representative autoequivalence $\rho_{\bf g}$ of $\mathcal{C}$ with Eq.~\ref{eqn:localFermionTransform} constrains
        \begin{equation}
            \rho_{\bf g}(\ket{\psi, \psi;1}) = \ket{\psi, \psi; 1},
        \end{equation}
        or equivalently, for all ${\bf g} \in G_b$,
        \begin{equation}
            U_{\bf g}(\psi, \psi;1) = +1.
            \label{eqn:UpsiConstraint}
        \end{equation}
        
        In the bosonic case, the topological autoequivalence may be redefined by a natural isomorphism of the form Eq.~\ref{eqn:natIso}. Clearly we have less freedom in the fermionic case; maintaining Eq.~\ref{eqn:UpsiConstraint} requires natural isomorphisms to have $\gamma_\psi \in \{\pm 1\}$. In fact, as shown in~\cite{bulmash2021}, modifying $R_{\bf g}$ with a natural isomorphism with $\gamma_\psi = -1$ amounts to redefining $R_{\bf g} \rightarrow (-1)^F R_{\bf g}$, which, as discussed in Sec.~\ref{subsec:fermionicSymm}, is a physical change to the system and not a gauge redundancy. Therefore, the redundancy in topological autoequivalences for fermionic symmetries is not given by arbitrary natural isomorphisms, but instead by ``locality-respecting" natural isomorphisms which have $\gamma_\psi = +1$. We denote the group of topological autoequivalences of $\mathcal{C}$ which obey Eq.~\ref{eqn:UpsiConstraint}, modulo locality-respecting natural isomorphisms, as $\Aut_{LR}(\mathcal{C})$\footnote{It is not hard to show that every element of $\Aut(\mathcal{C})$ has a representative which obeys Eq.~\ref{eqn:UpsiConstraint}. Accordingly, one can check that, depending on details which will be discussed shortly, as a group $\Aut_{LR}(\mathcal{C})$ is either isomorphic to $\Aut(\mathcal{C})$ or double-covers $\Aut(\mathcal{C})$.}.
        
        As such, fermionic symmetry actions are specified by a map
        \begin{equation}
            [\rho_{\bf g}]: G_b \rightarrow \Aut_{LR}(\mathcal{C})
        \end{equation}
        which obeys
        \begin{equation}
            [\kappa_{\bf g,h}] [\rho_{\bf g}] [\rho_{\bf h}] = [\rho_{\bf gh}] 
        \end{equation}
        for some (not-necessarily locality-respecting) natural isomorphism $[\kappa_{\bf g,h}]$. If $\Aut_{LR}(\mathcal{C})$ is isomorphic to $\Aut(\mathcal{C})$, then $[\kappa] = [\text{Id}]$, where Id is the identity map, and $[\rho_{\bf g}]$ is a group homomorphism.
        
        There is some subtlety in demanding that locality-respecting natural isomorphisms have $\gamma_\psi = +1$ because of the redundancy Eq.~\ref{eqn:zetaRedundancy} in natural isomorphisms. Define an Abelian group $K(\mathcal{C})$ consisting of maps $\zeta_a$ that obey the fusion rules, that is,
        \begin{equation}
            K(\mathcal{C}) = \{\zeta: \text{anyon labels} \rightarrow \U \phantom{i} | \phantom{i} \zeta_a \zeta_b = \zeta_c \text{ whenever } N_{ab}^c > 0 \}.
        \end{equation}
        Clearly $\zeta_\psi \in \{\pm 1\}$, which provides a natural $\Z_2$ grading on $K(\mathcal{C})$:
        \begin{equation}
            K(\mathcal{C}) = K_+(\mathcal{C}) \oplus K_-(\mathcal{C})
        \end{equation}
        where $K_\pm(\mathcal{C})$ consists of maps with $\zeta_\psi = \pm 1$. It was proven in~\cite{bulmash2021} that, given any minimal modular extension $\C$ of $\mathcal{C}$, every element $\zeta_a$ of $K(\mathcal{C})$ can be written
        \begin{equation}
            \zeta_a = M_{a,x}
        \end{equation}
        for some $x\in \C$, and in particular, $x \in \A$ if $\zeta \in K_+(\mathcal{C})$. Hence 
        \begin{equation}
        K_+(\mathcal{C}) \simeq \A/\{1,\psi\},
        \end{equation}
        where $\A \subset \mathcal{C}$ is the group of Abelian anyons of $\mathcal{C}$ and $\{1,\psi\}$ is the $\Z_2$ subgroup generated by the transparent fermion $\psi$.
        
        If $K_-(\mathcal{C})$ is nonempty, then every natural isomorphism which preserves Eq.~\ref{eqn:UpsiConstraint} is equivalent to a locality-respecting natural isomorphism.
        
        There is a canonical natural isomorphism $\Upsilon_\psi$ of the same form as Eq.~\ref{eqn:natIso} with
        \begin{equation}
            \Upsilon_\psi(\ket{a,b;c}) = \frac{\gamma_a \gamma_b}{\gamma_c}\ket{a,b;c} \text{ with } \gamma_a = \begin{cases} -1 & a = \psi\\
            1 & \text{ any other }a \in \mathcal{C} 
            \end{cases}
            \label{eqn:upsilonPsiDef}
        \end{equation}
        We immediately see that $\Upsilon_\psi$ respects locality if and only if $K_-(\mathcal{C})$ is nonempty. If $\Upsilon_\psi$ does not respect locality, then $[\Upsilon_\psi] \neq [1]$ as elements of $\Aut_{LR}(\mathcal{C})$.
        
        Many of the results and casework in this paper depend on whether or not $\Upsilon_\psi$ respects locality, so it is useful to summarize some characterizations from~\cite{bulmash2021} of when $\Upsilon_\psi$ respects locality. The following are equivalent:
        \begin{itemize}
            \item $\Upsilon_\psi$ respects locality
            \item $K(\mathcal{C})/K_+(\mathcal{C}) \simeq\Z_2$ (which means $K_-(\mathcal{C})$ is nonempty)
            \item As groups, $\Aut(\mathcal{C})\simeq \Aut_{LR}(\mathcal{C})$
            \item There exists a minimal modular extension of $\mathcal{C}$ which contains an Abelian fermion parity vortex
            \item There exists a set of phases $\zeta_a$ which obey the fusion rules and have $\zeta_\psi = -1$.
        \end{itemize}
        
        Conversely, the following are also equivalent:
        \begin{itemize}
            \item $\Upsilon_\psi$ violates locality
            \item $K(\mathcal{C}) = K_+(\mathcal{C}) \simeq \mathcal{A}/\{1,\psi\}$
            \item As groups, $\Aut_{LR}(\mathcal{C})/\Z_2 \simeq \Aut(\mathcal{C})$
            \item Any set of phases $\zeta_a$ which obey the fusion rules must have $\zeta_\psi = +1$.
        \end{itemize}
        
        Another useful fact is that if $\Upsilon_\psi$ respects locality, then exactly half of the minimal modular extensions $\C$ contain only $v$-type vortices and half contain only $\sigma$-type vortices.
        
        Also, if elements of $\Aut(\mathcal{C})$ are uniquely determined by their permutation action on the anyons (as is true in many theories of physical interest), then the same holds for $\Aut_{LR}(\mathcal{C})$ if and only if $\Upsilon_\psi$ respects locality.
        
        One can also define a natural isomorphism $\widecheck{\Upsilon}_\psi$ on a minimal modular extension $\C$ analagously:
        \begin{equation}
        \widecheck{\Upsilon}_\psi(\ket{a,b;c}) = \frac{\gamma_a \gamma_b}{\gamma_c}\ket{a,b;c} \text{ with } \gamma_a = \begin{cases} -1 & a = \psi\\
            1 & \text{ any other } a \in \C
            \end{cases}
            \label{eqn:checkUpsilonPsiDef}
        \end{equation}
        Similar to the case of $\Upsilon_\psi$, the following are equivalent: (i) The map $\widecheck{\Upsilon}_\psi$ respects locality; (ii) $\C$ contains an Abelian fermion parity vortex; (iii) There exists a set of phases $\widecheck{\zeta}_a$ on $\C$ which obey the fusion rules and have $\widecheck{\zeta}_\psi = -1$.
        
        After specifying the group homomorphism, one must compute the obstruction to symmetry localization. There are two obstructions. The first is the ``bosonic" obstruction to defining any $G_b$ symmetry fractionalization whatsoever on $\mathcal{C}$ and is valued in $\H^3(G_b,K(\mathcal{C}))$; if $\Upsilon_\psi$ violates locality, then we may characterize $K(\mathcal{C}) \simeq \A/\{1,\psi\}$, but if $\Upsilon_\psi$ respects locality, there is nothing further to say in general. The derivation is very similar to the bosonic case; see~\cite{bulmash2021} for details.
        
        At this point one may define symmetry fractionalization data $\eta_a({\bf g,h})$ subject to the usual consistency conditions Eqs.~\ref{eqn:etaConsistency}-\ref{eqn:UEtaConsistency}. However, compatibility of the symmetry localization ansatz with Eq.~\ref{eqn:localFermionTransform} and the full $G_f$ symmetry in Eq.~\ref{eqn:projectiveRg} requires~\cite{bulmash2021} the constraint
        \begin{equation}
            \eta_\psi({\bf g,h})=\omega_2({\bf g,h}) \in Z^2(G_b,\Z_2).
            \label{eqn:etaPsiConstraint}
        \end{equation}
        Regarding autoequivalences as gauge-equivalent only if they differ by \textit{locality-respecting} natural isomorphisms means that gauge transformations preserve Eq.~\ref{eqn:etaPsiConstraint}.
        
        The existence of a symmetry fractionalization pattern obeying Eq.~\ref{eqn:etaPsiConstraint} is subject to a ``fermionic" symmetry localization obstruction which is valued in $Z^2(G_b,\Z_2)$ if $\Upsilon_\psi$ violates locality and is valued in $\H^3(G_b,K_+(\mathcal{C}))=\H^3(G_b,\A/\{1,\psi\})$ in $\Upsilon_\psi$ respects locality.
        
        To summarize, symmetry fractionalization of a fermionic symmetry group $G_f$ on a super-modular category $\mathcal{C}$ is given by a homomorphism $[\rho_{\bf g}]:G_b \rightarrow \Aut_{LR}(\mathcal{C})$ (which defines the $U$-symbols $U_{\bf g}(a,b;c)$) and a choice of data $\eta_a({\bf g,h})$, subject to the same consistency conditions Eqs.~\ref{eqn:UFConsistency},\ref{eqn:URConsistency},\ref{eqn:etaConsistency},\ref{eqn:UEtaConsistency} as in the bosonic case. This data is subject to the constraints
        \begin{align}
            U_{\bf g}(\psi,\psi;1)&=+1\\
            \eta_\psi({\bf g,h})&= \omega_2({\bf g,h}),
        \end{align}
        and symmetry action gauge transformations are restricted to be locality-respecting, in that they must have
        \begin{equation}
            \gamma_\psi({\bf g})=+1.
        \end{equation}
        The existence of consistent symmetry fractionalization requires two obstructions to vanish: a bosonic symmetry localization obstruction $[\coho{O}_b]\in \H^3(G_b,\A/\{1,\psi\})$ and a fermionic symmetry localization obstruction $[\coho{O}_f]$ which is valued in $Z^2(G_b,\Z_2)$ if $\Upsilon_\psi$ violates locality and which is valued in $\H^3(G_b,\A/\{1,\psi\})$ if $\Upsilon_\psi$ respects locality.
        
    \section{\texorpdfstring{$\mathcal{H}^1(G_b,\Z_{\bf T})$}{H1(Gb, ZT)} obstruction}
    \label{sec:H1}
        
        We prove Theorem~\ref{thm:H1}.
        
        \begin{proof}
        Using the Gauss sum,
        \begin{equation}
           c_{\nu_2}=- c_{\nu_1} \mod 8.
        \end{equation}
        We have defined 
        \begin{equation}
        c_{\nu_2} = c_{\nu_1}+o_1({\bf g})/2 \mod 8
        \label{eqn:cRelation}
        \end{equation}
        where $o_1({\bf g})$ must, by the 16-fold way, be an integer. Hence
        \begin{equation}
            o_1({\bf g}) = -4c_{\nu_1} \mod 16.
        \end{equation}
        Hence $c_{\nu_1} = 0 \mod 1/4$. Now, suppose that $(\nu_3,\nu_4)$ also satisfy Eq.~\ref{eqn:cRelation} with a different integer $n'({\bf g})$. Then we could run the same argument to obtain
        \begin{equation}
            o_1'({\bf g}) = -4c_{\nu_3} \mod 16.
        \end{equation}
        But $c_{\nu_3}-c_{\nu_1} \in \Z/2$ by the 16-fold way, which means
        \begin{equation}
            o_1'({\bf g})-o_1({\bf g}) \in 2\Z.
        \end{equation}
        \end{proof}
        
        It immediately follows that if $o_1({\bf g}) = 1 \mod 2$, then it is not possible to have a lift $\rho_{\bf g}:\C_{\nu_0}\rightarrow\C_{\nu_0}$. This is a rather familiar statement because $o_1({\bf g})=1$ implies $c_{\nu_i}\in \Z \pm 1/4$, which can never be left invariant by an anti-unitary symmetry. Hence $[o_1] \in \H^1(G_b,\Z_{\bf T})$ obstructs the ability to lift $\rho_{\bf g}$ to a map which is a true symmetry of a minimal modular extension $\C$, rather than a map between two distinct minimal modular extensions. Conjecture~\ref{conj:H1ActualObstruction} states that $[o_1]$ is the only such obstruction.
        
        \section{\texorpdfstring{$\mathcal{H}^2(G_b, \ker r)$}{H2(Gb, ker r)} obstruction}
        \label{sec:H2}
    
        Now let us suppose that the $\mathcal{H}^1$ obstruction defined in the preceding section vanishes. Here we will find that there is an obstruction to lifting the maps $[\rho_{\bf g}]$ to $[\widecheck{\rho}_{\bf g}]$ in a way which appropriately satisfies the group structure.
        
        In what follows, we fix a choice of minimal modular extension $\C_{\nu}$ and suppress the $\nu$ index. In general the obstruction $o_2$ that we will find depends on $\nu$; we will discuss this dependence in Sec.~\ref{sec:o2independence}.
        
        We will assume for now that $\widecheck{\Upsilon}_\psi$ (defined in Eq. \ref{eqn:checkUpsilonPsiDef}) respects locality. We will discuss cases where it violates locality afterwards in Section \ref{UpsiCheckViolates}. Since $\widecheck{\Upsilon}_\psi$ and therefore  $\Upsilon_\psi$ both respect locality in the first part of our discussion, $[\kappa_{\bf g,h}]$ is always the identity and so the map
        $[\rho_{\bf g}] :G_b \rightarrow \Aut_{LR}(\mathcal{C})$ is a group homomorphism. Since the $\mathcal{H}^1$ obstruction vanishes, for each ${\bf g}\in G_b$, one can define a consistent lift $[\widecheck{\rho}_{\bf g}]$ on a fixed minimal modular extension $\C$, that is, $r([\widecheck{\rho}_{\bf g}])=[\rho_{\bf g}]$, where $r$ is the restriction map
        \begin{align}
            r : \Aut_{LR}(\C) \rightarrow \Aut_{LR}(\C)|_\mathcal{C} \subseteq \Aut_{LR}(\mathcal{C}). 
        \end{align}
        However, the lifted map $[\widecheck{\rho}_{\bf g}]$ may not be a group homomorphism. Our aim is to show that there is an obstruction to finding a lift $[\widecheck{\rho}_{\bf g}]$ which is a group homomorphism $G_b \rightarrow \Aut_{LR}(\C)$, and this obstruction is valued in $\H^2(G_b,\ker r)$. We then relate this obstruction to the $\H^2(G_b,\Z_2)$ part of the 't Hooft anomaly for fermionic SETs.
        
        \subsection{Defining the obstruction}
        \label{subsec:definingObstruction}
        
        Suppose we have a lift $[\widecheck{\rho}_{\bf g}]$ of a general autoequivalence $[\rho_{\bf g}]$ of $\mathcal{C}$. In general $r$ may have a nontrivial kernel, so $\widecheck{\rho}_{\bf g}$ can be composed with any element of $\ker r$ to obtain another, equally valid lift. Although $\rho_{\bf g}$ and $\widecheck{\rho}_{\bf g}$ may in general be antiunitary, elements of $\ker r$ are automatically unitary.
        
        We will show in Sec.~\ref{subsec:kerR} that the permutation action of all elements of $\ker r$ on anyon labels commute with each other, although they in general need not commute with the permutation action of $[\rho_{\bf g}]$. In the case where elements of $\Aut(\C)$ are completely determined by their permutation action on the anyons, then this implies that $\ker r$ is Abelian. We will assume that $\ker r$ is Abelian in general. 
        
        Let us consider
        \begin{align}        
        o_2({\bf g}, {\bf h}) :=  \widecheck{\rho}_{\bf gh}\widecheck{\rho}_{\bf h}^{-1}\widecheck{\rho}_{\bf g}^{-1}.
        \label{eqn:o2Def}
        \end{align}
        By inspection, modifying a representative lift $\widecheck{\rho}_{\bf g}$ by a locality-respecting natural isomorphism modifies $o_2$ by a locality-respecting natural isomorphism, so this equation is also well-defined in $\Aut_{LR}(\C)$.
        
        We warn the reader that we will overload notation so that $o_2({\bf g,h})$ can refer both to a topological autoequivalence and its equivalence class in $\Aut_{LR}(\C)$ (after modding out by natural isomorphisms). We reserve $[o_2]$ for later use as a cohomology class. 
        
        Since $[\rho_{\bf g}]$ is a group homomorphism $G_b \rightarrow \Aut_{LR}(\mathcal{C})$, $o_2$ restricts to a trivial map in $\Aut_{LR}(\mathcal{C})$, i.e., $o_2({\bf g,h}) \in \ker r$. However, it may be a nontrivial element of $\ker r$. In general, $o_2({\bf g,h}) \in C^2(G_b,\ker r)$ defines a $(\ker r)$-valued 2-cochain on $G_b$.
        
        Demanding that the $\widecheck{\rho}_{\bf g}$ be associative, we find by decomposing $\widecheck{\rho}_{\bf ghk}$ in two distinct ways that
        \begin{align}
            \widecheck{\rho}_{\bf ghk} &= o_2({\bf gh,k})\widecheck{\rho}_{\bf gh}\widecheck{\rho}_{\bf k}\\
            &= o_2({\bf gh,k})o_2({\bf g,h)}\widecheck{\rho}_{\bf g}\widecheck{\rho}_{\bf h} \widecheck{\rho}_{\bf k}\\
            &= o_2({\bf g,hk})\widecheck{\rho}_{\bf g}\widecheck{\rho}_{\bf hk}\\
            &= o_2({\bf g,hk})\,^{\bf g}o_2({\bf h,k})\widecheck{\rho}_{\bf g}\widecheck{\rho}_{\bf h}\widecheck{\rho}_{\bf k},
        \end{align}
        where we have defined
        \begin{equation}
            \,^{\bf g}o_2({\bf h,k})=\widecheck{\rho}_{\bf g}o_2({\bf h,k})\widecheck{\rho}_{\bf g}^{-1}.
        \end{equation}
        For these two decompositions of $\widecheck{\rho}_{\bf ghk}$ to be equal, we need $o_2({\bf g,h}) \in Z^2(G_b,\ker r)$.
        
        Clearly $[\widecheck{\rho}_{\bf g}]$ is only a group homomorphism if $o_2 = 1$. This condition is not generically satisfied, but we may obtain another lift by modifying each $[\widecheck{\rho}_{\bf g}]$ by an element of $\ker r$. Such a modification changes $o_2({\bf g,h})$ by a $(\ker r)$-valued 2-coboundary. Therefore, the lift can be modified to obtain a group homomorphism $G_b \rightarrow \Aut_{LR}(\C)$ if and only if $[o_2]\in \H^2(G_b,\ker r)$ is cohomologically trivial. That is, $[o_2]$ is the obstruction to lifting the permutation action of $G_b$ on $\mathcal{C}$ to $\C$.
        
        In the case where $\ker r = \Z_2$, we will see that $\ker r$ commutes with all of $\Aut_{LR}(\C)$. Accordingly, $[o_2]$ defines a group extension $\widecheck{G}_b$ of $G_b$ by $\Z_2$. In fact, if we enlarge the symmetry to  $\widecheck{G}_b$, then there is a consistent lift to $[\widecheck{\rho}_{\bf g}]\in \Aut_{LR}(\C)$. Let $\ker r = \{1, [\alpha_\psi]\}$, and with $\widecheck{G}_b = G_b \times \Z_2$ as sets, define
        \begin{equation}
            \widecheck{\rho}_{({\bf g, p})} = \widecheck{\rho}_{\bf g}\alpha_{\psi}^{\bf p}
        \end{equation}
        with ${\bf g} \in G_b$ and ${\bf p} \in \{0,1\} \simeq \Z_2$. Then
        \begin{align}
              [\widecheck{\rho}_{({\bf g, p})}]  [\widecheck{\rho}_{({\bf h, q})}] &= o_2({\bf g, h})[\alpha_\psi]^{{\bf p}+{\bf q}}[\widecheck{\rho}_{\bf gh}]\\
              &=[\widecheck{\rho}_{({\bf gh},{\bf p}+{\bf q}+\tilde{o}_2({\bf g,h}))}]\\
              &= [\widecheck{\rho}_{({\bf g, p}) \times ({\bf h, q})}],
        \end{align}
        where $\tilde{o}_2 \in \{0, 1\} \simeq \Z_2$ means we are interpreting $o_2 \in \ker r \simeq \Z_2$ as an element of additive $\Z_2$ instead of $\ker r$. Hence these symmetry actions are a group homomorphism $[\widecheck{\rho}_{({\bf g},{\bf p})}]: \widecheck{G}_b \rightarrow \Aut_{LR}(\C)$ as claimed.
        
        As an example, let us consider the semion-fermion theory with $G_b = \Z_2^{\bf T}$. The appropriate modular extension of this theory is $\C= U(1)_2 \times U(1)_{-4}$. The simple objects of $\C$ can be labeled $(a,b)$ for $a =0,1$ and $b = 0,1,2,3$. Here, $v = (0,1)$ is the fermion parity vortex, $\psi = (0,2)$ is the fermion, and $s = (1,0)$ is the semion. It is clear that $\ker r = \Z_2$ in this case; its nontrivial element $[\alpha_\psi]$ takes $v \leftrightarrow v \times \psi$. Under ${\bf T}$, $\mathcal{C}$ transforms as follows:
        \begin{align}
            \,^{\bf T} s &= s \times \psi
            \nonumber \\
            \,^{\bf T} \psi &= \psi 
        \end{align}
        There are two possible lifts $[\widecheck{\rho}_{\bf T}]$ to the modular extension which differ by the action of $[\alpha_\psi]$; we may take either
        \begin{equation}
            \,^{\bf T} v = s \times v \text{ or } \,^{\bf T} v = s \times \psi \times v
        \end{equation}
        For the first choice
        \begin{align}
            \,^{\bf T}(\,^{\bf T} v) = \,^{\bf T}(s \times v) = (s \times \psi) \times (s \times v) = \psi \times v.
        \end{align}
        That is,
        \begin{equation}
           [ \widecheck{\rho}_{\bf T}]^2 = [\alpha_\psi] = o_2({\bf T,T})
        \end{equation}
        One can check straightforwardly that the second choice of $[\widecheck{\rho}_{\bf T}]$ leads to the same $o_2$. Therefore $[o_2] \neq +1 \in \H^2(\Z_2^{\bf T},\Z_2)$. Thus there is no way to have the permutations faithfully act on the modular extension as $G_b = \Z_2^{\bf T}$. However we can have them act as $\widecheck{G}_b = \Z_4^{\bf T}$.
        
        \subsubsection{\texorpdfstring{$\widecheck{\Upsilon}_\psi$ violates locality}{UpsiCheck violates locality}}
        \label{UpsiCheckViolates}
        
        If $\widecheck{\Upsilon}_\psi$ violates locality, the discussion above must be slightly modified.\footnote{The first version of this paper assumed that $[\rho_{\bf g}]$ and $[\widecheck{\rho}_{\bf g}]$ is always a group homomorphism, and did not incorporate the distinction between situations where $\widecheck{\Upsilon}_\psi$ violates locality and preserves it. This was addressed in the first version of Ref.~\cite{aasen21ferm}, which appeared on the arXiv simultaneously with the first version of this paper. We provide an alternate treatment here.}. 
        
        In particular, there are now two possibilities. The first possibility is that $\Upsilon_\psi$ does not respect locality. In this case, if the symmetry fractionalization on $\mathcal{C}$ is unobstructed, $[\rho_{\bf g}]$ generally fails to be a group homomorphism up to factors of $[\Upsilon_\psi]$ \cite{bulmash2021,aasen21ferm}:
        \begin{equation}
            [\Upsilon_\psi]^{\tilde{\omega}_2({\bf g,h})}[\rho_{\bf g}][\rho_{\bf h}] = [\rho_{\bf gh}].
            \label{eqn:modifiedHomomorphism}
        \end{equation}
        where $\tilde{\omega}_2$ is defined by Eq.~\ref{eqn:w2Def}. In this case, we must instead enforce the equation
        \begin{equation}
            [\widecheck{\Upsilon}_\psi]^{\tilde{\omega}_2({\bf g,h})}[\widecheck{\rho}_{\bf g}][\widecheck{\rho}_{\bf h}] = [\widecheck{\rho}_{\bf gh}]
            \label{eqn:modifiedHomomorphismCheck}
        \end{equation}
        because $r([\widecheck{\Upsilon}_\psi])=[\Upsilon_\psi]$ and we need Eq.~\ref{eqn:modifiedHomomorphismCheck} to reduce to Eq.~\ref{eqn:modifiedHomomorphism} upon applying the restriction map $r$. We can thus define a modified obstruction which must vanish if we want Eq.~\ref{eqn:modifiedHomomorphismCheck} to hold:
        \begin{equation}
            o_2({\bf g,h}) := \widecheck{\rho}_{\bf gh}\widecheck{\rho}_{\bf h}^{-1}\widecheck{\rho}_{\bf g}^{-1} \widecheck{\Upsilon}_\psi^{\tilde{\omega}_2({\bf g,h})}
            \label{eqn:modifiedO2}
        \end{equation}
        The rest of the analysis showing that $[o_2] \in \H^2(G_b,\ker r)$ proceeds analogously to the case where $\widecheck{\Upsilon}_\psi$ respects locality, with the additional factors of $\widecheck{\Upsilon}_\psi$ cancelling out at all stages. Proving this statement requires using the fact, discussed in Sec.~\ref{subsec:kerR}, that $\widecheck{\Upsilon}_\psi$ commutes with all topological autoequivalences.
        
        In the case where $\Upsilon_\psi$ respects locality but $\widecheck{\Upsilon}_\psi$ does not, we again need to enforce Eq.~\ref{eqn:modifiedHomomorphismCheck}, but for a slightly more involved reason. We take a lift $\widecheck{\rho}_{\bf g}$ and attempt to enforce the condition
        \begin{equation}
            \widecheck{\kappa}_{\bf g,h} \widecheck{\rho}_{\bf g} \widecheck{\rho}_{\bf h} = \widecheck{\rho}_{\bf gh}
            \label{eqn:checkKappaDef}
        \end{equation}
        with $\widecheck{\kappa}_{\bf g,h}$ a (possibly locality-violating) natural isomorphism. Since $\widecheck{\Upsilon}_\psi$ violates locality, this condition requires that the decomposition $\widecheck{\kappa}_{\bf g,h}(a,b;c)$ into anyon-dependent factors  $\widecheck{\beta}_a$ obey
        \begin{equation}
             \widecheck{\beta}_\psi = \begin{cases}
               1 & [\widecheck{\kappa}_{\bf g,h}]= [1]\\
               -1 & [\widecheck{\kappa}_{\bf g,h}] = [\widecheck{\Upsilon}_\psi]
             \end{cases}.
             \label{eqn:betaPsiFromKappa}
        \end{equation}
        The above equation is gauge-invariant. Now, in order to be compatible with $G_f$, any putative $G_f$ symmetry fractionalization pattern will need to obey  $\widecheck{\eta}_\psi({\bf g,h})=\omega_2({\bf g,h})$. If such a $G_f$ symmetry fractionalization pattern can possibly exist, then
        \begin{equation}
            \widecheck{\omega}_\psi({\bf g,h}) = \frac{\widecheck{\beta}_\psi({\bf g,h})}{\widecheck{\eta}_\psi({\bf g,h})} = +1 \Rightarrow  \widecheck{\beta}_\psi({\bf g,h}) = \omega_2({\bf g,h}) 
            \label{eqn:checkBetaEqualsOmega2}
        \end{equation}
        (see Eq.~\ref{eqn:etaDef}) since $\widecheck{\Upsilon}_\psi$ violates locality and $\widecheck{\omega}_a$ must respect the fusion rules. Combining Eqs.~\ref{eqn:checkKappaDef},~\ref{eqn:betaPsiFromKappa}, and~\ref{eqn:checkBetaEqualsOmega2}, we find that our desired lift should again obey Eq.~\ref{eqn:modifiedHomomorphismCheck}. Accordingly, we should define $o_2$ via Eq.~\ref{eqn:modifiedO2}. The rest of the analysis is unchanged from the case where $\Upsilon_\psi$ violates locality.
        
        Note that in Eq. \ref{eqn:modifiedHomomorphismCheck}, the brackets $[]$ correspond to taking equivalence under \it locality-respecting \rm natural isomorphisms, so that $[\widecheck{\rho}_{\bf g}] \in \Aut_{LR}(\C)$. When we instead consider equivalence under all natural isomorphisms, $\widecheck{\rho}$ will reduce to a group homomorphism, as expected based on the theory of bosonic SETs \cite{barkeshli2019}.
        
        \subsection{Characterizing \texorpdfstring{$\ker r$}{ker r}}
        \label{subsec:kerR}

        In order to understand $[o_2]$ better, we need to characterize $\ker r$. We make the following conjecture:
        \begin{conjecture}
        In all cases, $\ker r = \Z_2$.
        \label{conjecture:kerRIsActuallyZ2}
        \end{conjecture}
        As a consequence, $[o_2] \in \H^2(G_b,\Z_2)$. Although we cannot prove this conjecture in full generality, in the following subsections we will provide a number of concrete results that motivate the conjecture and give partial progress towards proving it. One motivation for the above conjecture is that it is $\H^2(G_b, \Z_2)$ that appears in the characterization of (3+1)D FSPTs, which should classify the 't Hooft anomalies in (2+1)D. 
        
        We begin by discussing a few properties of the map $[\widecheck{\Upsilon}_\psi]$ and define an important map $[\alpha_\psi]$. We will see that $[\alpha_\psi]$ always generates a $\Z_2 \subseteq \ker r$. 
        
        \subsubsection{\texorpdfstring{$[\widecheck{\Upsilon}_\psi]$}{CheckYpsi}}
        
        We defined the map $\widecheck{\Upsilon}_\psi$ by Eq.~\ref{eqn:checkUpsilonPsiDef}. It immediately follows that $r([\widecheck{\Upsilon}_{\psi}])=[\Upsilon_\psi]$, so $[\widecheck{\Upsilon}_\psi] \in \ker r$ if and only if $\Upsilon_\psi$ respects locality.  If $[\widecheck{\Upsilon}_\psi] \in \ker r$ and $\widecheck{\Upsilon}_\psi$ violates locality, then it forms a $\Z_2$ subgroup of $\ker r$. On the other hand, if $\widecheck{\Upsilon}_\psi$ respects locality, then it is a trivial element of $\ker r$.
        
        Since $\psi$ is invariant under all (fermionic) topological autoequivalences, it follows that $\widecheck{\Upsilon}_\psi$ commutes with all topological autoequivalences.
        
        There are three possibilities for how $\widecheck{\Upsilon}_\psi$ and $\Upsilon_\psi$ behave. First, $\widecheck{\Upsilon}_\psi$ may respect locality, in which case, $\Upsilon_\psi$ does as well. It was proven in~\cite{bulmash2021} that all parity vortices are $v$-type in this case (and at least one is Abelian). Second, $\widecheck{\Upsilon}_\psi$ may violate locality but $\Upsilon_\psi$ may respect locality. From the list of properties in Sec.~\ref{sec:fermionicSymm}, we see that this occurs when $\mathcal{C}$ has some minimal modular extension with an Abelian fermion parity vortex, but $\C$ is not such a minimal modular extension. It was proven in~\cite{bulmash2021} that all parity vortices in $\C$ are $\sigma$-type in this case. Finally, both $\widecheck{\Upsilon}_\psi$ and $\Upsilon_\psi$ may violate locality.
        
        One can check that all of the cases we considered above can actually occur. 
        For an example where $\widecheck{\Upsilon}_\psi$ respects locality, we can take $\mathcal{C} = \mathcal{B} \boxtimes \{1, \psi\}$ for any modular $\mathcal{B}$ and $\C = \mathcal{B} \boxtimes D(\Z_2)$, where $D(\Z_2)$ is the quantum double of $\Z_2$, also known as the toric code topological order. Here $\boxtimes$ is the Deligne product and physically corresponds to stacking decoupled topological orders. Instead taking $\C = \mathcal{B} \boxtimes \text{Ising}$ gives an example where $\widecheck{\Upsilon}_\psi$ does not respect locality but $\Upsilon_\psi$ does. For an example where neither $\widecheck{\Upsilon}_\psi$ nor $\Upsilon_\psi$ respects locality, we can take $\mathcal{C} = \mathrm{SO}(3)_3$; an example minimal modular extension is $\C = \mathrm{SU}(2)_6$.
        
        \subsubsection{\texorpdfstring{$[\alpha_\psi]$}{alphaPsi}}
        
        Next, we consider a map $\alpha_\psi$, with the following permutation action:
        \begin{align}
            \alpha_\psi(a) &= a \;\;\;\text{ if } a \in \C_0 \simeq \mathcal{C}
            \nonumber \\
            \alpha_\psi(a) &= a \times \psi \;\;\; \text{ if } a \in \C_1
            \label{eqn:alphaPsiDef}.
        \end{align}
        This permutation preserves the fusion rules, twists, and modular $S$-matrix of the theory. One can check in a range of examples that there indeed exists a braided autoequivalence $\alpha_\psi$ of $\C$ with this permutation action. In fact, Ref.~\cite{aasen21ferm} (which appeared simultaneously on the arXiv with the first version of this work) gave an explicit formula for the $U$-symbols of exactly such a braided autoequivalence as follows:
        \begin{equation}
            U_{\alpha_\psi}(a',b';c';\mu,\nu) = \sum_{\lambda} \left[\left(F^{a, \psi^{f_b},b'}_{c}\right)^{-1}\right]_{(b,\mu),(a \times \psi^{f_b},\lambda)} R^{\psi^{f_b},a}\left[\left(F^{\psi^{f_c},\psi^{f_b}\times a, b'}_{c'}\right)^{-1}\right]_{(c,\lambda),(a',\nu)}
            \label{eqn:alphaPsiUSymbols}
        \end{equation}
        where we use the shorthand $a' = \alpha_\psi(a)$ and define
        \begin{equation}
            f_x = \begin{cases}
              0 & x \in \C_0\\
              1 & x \in \C_1
            \end{cases}.
        \end{equation}
        
        Consider the case where (i) every permutation of simple objects of $\C$ which preserves the modular data corresponds to a unique element $\Aut(\C)$, and similarly (ii) every permutation of simple objects in $\mathcal{C}$ preserving the modular data corresponds to a unique element of $\Aut(\mathcal{C})$.
        Then $[\alpha_\psi]$ generates a $\Z_2$ subgroup of  $\ker r$ as long as $\C_v$ is non-empty. In more general situations, Ref.~\cite{aasen21ferm} demonstrated that as a braided autoequivalence, it is always true that $[\alpha_\psi] \in \ker r$, commutes with all of $\Aut_{LR}(\C)$, and squares to the identity. 
        
        If $\C_v$ is non-empty, then $[\alpha_\psi]$ is clearly non-trivial. If $\C_\sigma$ is non-empty, then we can calculate the gauge-invariant (in $\Aut_{LR}(\C)$) quantity
        \begin{align}
            U_{\alpha_\psi}(\sigma,\psi;\sigma) &= \left(F^{\sigma 1 \psi} \right)^{-1}R^{1 \sigma} \left(F^{\psi \sigma \psi}\right)^{-1}\\
            &= \left(F^{\psi \sigma \psi}\right)^{-1} = -1,
            \label{eqn:UsigmaPsiSigmaMinus1}
        \end{align}
        where the last equality follows from a straightforward use of the hexagon equation. Hence $[\alpha_\psi]$ is non-trivial in this case as well, and so $[\alpha_\psi]$ always generates a central $\Z_2$ subgroup of $\ker r$. Note that Eq.~\ref{eqn:alphaPsiUSymbols} provides a definition of a non-trivial $[\alpha_\psi]$ even when $\C_v$ is empty, in which case $\alpha_\psi$ always has trivial permutation action on the objects in $\C$. 
        
        For a familiar example, consider $\mathcal{C}=\{1,\psi\}$; then there is a modular extension $\C = \{1,\psi,e,m\}$ which is equivalent to $\Z_2$ gauge theory, where we are viewing $\C_v = \{e,m\}$. Then the map $\alpha_\psi$ permutes $e \leftrightarrow m$, implementing electric-magnetic duality.

        If $\C_v$ is empty, then $[\alpha_\psi]$ may or may not equal $[\widecheck{\Upsilon}_\psi]$. If $\Upsilon_\psi$ violates locality, then these maps cannot be equal since $[\alpha_\psi]$ is in $\ker r$ but $[\widecheck{\Upsilon}_\psi]$ is not. In a version of Ref.~\cite{aasen21ferm} posted after the second version of this paper, it was proven that $[\alpha_\psi] = [\widecheck{\Upsilon}_\psi]$ if and only if $\widecheck{\Upsilon}_\psi$ violates locality but $\Upsilon_\psi$ respects locality. In other words, any time $[\widecheck{\Upsilon}_\psi]$ is a non-trivial element in $\ker r$, it is equal to $[\alpha_\psi]$.
        
        \subsubsection{Permutation actions of elements of \texorpdfstring{$\ker r$}{ker r}}
        \label{proofI6}
        
        Fully characterizing $\ker r$ is a non-trivial task in general. We can, however, determine the allowed \textit{permutation action} of all elements of $\ker r$ by proving Theorem~\ref{thm:kerRPerm}:
        
        \begin{proof}
        We may apply the Verlinde formula to $\C$, which is modular:
        \begin{align}
            N_{v,\psi}^{\widecheck{\rho}(v)} + N_{v,1}^{\widecheck{\rho}(v)} &= \sum_{x \in \C}\frac{\left(S_{\psi,x} + S_{1x}\right)S_{vx}S^{\ast}_{\widecheck{\rho}(v)x}}{S_{1x}}\\
            &= \sum_{x \in \C_0}\frac{\left(S_{\psi,x} + S_{1x}\right)S_{vx}S^{\ast}_{\widecheck{\rho}(v)x}}{S_{1x}} + \sum_{x \in \C_1}\frac{\left(S_{\psi,x} + S_{1x}\right)S_{vx}S^{\ast}_{\widecheck{\rho}(v)x}}{S_{1x}}.
        \end{align}
        By Eq.~\ref{eqn:SChangeByPsi}, $S_{\psi,x} = \pm S_{1,x}$ with the upper sign for $x \in \C_0$ and the lower sign for $x \in \C_1$. Therefore,
        \begin{align}
             N_{v,\psi}^{\widecheck{\rho}(v)} + N_{v,1}^{\widecheck{\rho}(v)} &= \sum_{x \in \C_0}2S_{vx}S^{\ast}_{\widecheck{\rho}(v)x} + 0\\
             &= \sum_{x \in \C_0}2|S_{vx}|^2 > 0,
        \end{align}
        where we have used the fact that $S$ is invariant under $\widecheck{\rho}$ and that $\widecheck{\rho}(x)=x$ if $x \in \C_0$.
        Therefore either $\widecheck{\rho}(v)=v$ or $\widecheck{\rho}(v) = \psi \times v$, and in particular, if $v \in \C_\sigma$, $\widecheck{\rho}(v) = v$. 
        
        Next, suppose $v_1, v_2 \in \C_v$. Then
        \begin{equation}
            S_{v_1,v_2}= \widecheck{\rho}(S_{v_1,v_2})=S_{\widecheck{\rho}(v_1),\widecheck{\rho}(v_2)} = S_{v_1,v_2} (-1)^{m_1 + m_2}
        \end{equation}
        where
        \begin{equation}
            m_i = \begin{cases}
              0 & \widecheck{\rho}(v_i) = v_i\\
              1 & \widecheck{\rho}(v_i) = v_i \times \psi
            \end{cases}
        \end{equation}
        Hence if $S_{v_1,v_2} \neq 0$, then $m_1=m_2$, that is, $\widecheck{\rho}$ changes the fermion parity of both $v_1$ and $v_2$ or of neither. If $S_{v_1,v_2} = 0$ but $v_1,v_2$ are in the same one of the $k$ blocks of the $S_{a_v,b_v}$ part of the $S$-matrix, then there exists a sequence of $v$-type vortices $a_1,a_2,\ldots a_p$ such that $S_{v_1,a_1} \neq 0$,$S_{a_1,a_2}\neq 0$, $\cdots$, $S_{a_p,v_2} \neq 0$. Applying the above argument to each pair in the sequence, we conclude that if $v_1,v_2$ belong to the same block of the $S$-matrix, $\widecheck{\rho}$ acts the same way on both vortices, i.e. it either changes the fermion parity of both or of neither.
        \end{proof}

        \subsubsection{$\ker r$ when permutations determine $\Aut(\C)$ and $\Aut(\mathcal{C})$}
        
        In many well-studied examples, every permutation of the simple objects of a BFC $\mathcal{B}$ uniquely determines an element of $\Aut(\mathcal{B})$. We can fully characterize $\ker r$ as long as $\C$ and $\mathcal{C}$ obey slightly weaker properties:  
        \begin{theorem}
        Let $k$ be as in Theorem~\ref{thm:kerRPerm}. Suppose that every permutation of vortices in $\C$ given in Theorem~\ref{thm:kerRPerm} uniquely determines an element of $\Aut(\C)$ (but not necessarily a unique element of $\Aut_{LR}(\C)$), and further suppose that there is a unique element of $\Aut(\mathcal{C})$ which does not permute anyons. Then $\ker r = \Z_2^{\max(k,1)}$.
        \label{thm:kerR}
        \end{theorem}
        
        \begin{proof}
        Let $[\widecheck{\rho}] \in \ker r$. According to Theorem~\ref{thm:kerRPerm}, there are $2^{\max(k,1)}$ possible permutation actions for $[\widecheck{\rho}]$, given by $k$ independent choices of whether or not $[\widecheck{\rho}]$ changes the fermion parity of the $v$-type vortices in each block, and each permutation action squares to the identity. If $k=0$, then $[\widecheck{\rho}]$ acts as the identity permutation.
        
        We first claim that all of these possible permutation actions commute with each other. This statement is only nontrivial for $k>1$. To prove this claim, it suffices to show that $v$ and $v \times \psi$ belong to the same block. 
        
        Suppose first that for some $w \neq v$, $v$ and $w$ belong to the same one of the $k$ blocks of the $S$-matrix. Then there is a sequence of $v$-type vortices $a_1,a_2,\ldots a_p$ such that $S_{v,a_1} \neq 0$,$S_{a_1,a_2}\neq 0$, $\cdots$, $S_{a_p,w} \neq 0$. But $S_{v\times \psi,a_1} = -S_{v,a_1} \neq 0$ as well. Hence $v \times \psi$ is also in the same block as $w$, so $v$ and $v \times \psi$ are in the same block, namely the block containing $w$. 
        
        We claim that such a $w$ must exist. Suppose by way of contradiction that no such $w$ exists, i.e., $S_{v,w}=0$ for all $w \in \C_1$, in which case $v$ is in a block by itself. It follows that $S_{v\times \psi, w} = -S_{v,w} =0$ for all $w \in \C_1$, and since $S_{v\times \psi,a}=S_{v,a}$ for all $a \in \C_0$, we must have $S_{v,a} = S_{v\times \psi,a}$ for all $a \in \C$. Hence $S$ has two identical rows and is not invertible, which is a contradiction since $\C$ is modular. This proves that all of the aforementioned permutations commute.
        
        By assumption, the permutation action of $\widecheck{\rho}$ defines a unique element of $\Aut(\C)$, but may or may not uniquely determine an element of $\Aut_{LR}(\C)$. Whether or not it does depends on the properties of $[\widecheck{\Upsilon}_\psi]$, which is the one possibly-nontrivial element of $\Aut_{LR}(\C)$ which does not permute anyons.
        
        We need to consider three possible cases, depending on whether $\Upsilon_\psi$ and $\widecheck{\Upsilon}_\psi$ respect locality.
        
        \underline{Case 1:} $\widecheck{\Upsilon}_\psi$ respects locality in $\C$. Then $\Upsilon_\psi$ respects locality in $\mathcal{C}$, and also $\C$ contains an Abelian fermion parity vortex. Hence $\C_v$ is nonempty, $k>0$, $\Aut_{LR}(\C)=\Aut(\C)$, and $\Aut_{LR}(\mathcal{C})=\Aut(\mathcal{C})$. Therefore, each of the $2^k$ anyon permutations defined above determines a unique element of $\Aut_{LR}(\C)$ which restricts to the identity in $\Aut_{LR}(\mathcal{C})$ (since the restricted permutation action is trivial). Furthermore, if $[\widecheck{\rho}]\in \ker r$, $[\widecheck{\rho}]^2=[1]$ because $[\widecheck{\rho}]^2$ does not permute anyons. Thus $\ker r = \Z_2^k$.
        
        \underline{Case 2}: $\widecheck{\Upsilon}_\psi$ does not respect locality in $\C$, but $\Upsilon_\psi$ respects locality in $\mathcal{C}$. Then some minimal modular extension of $\mathcal{C}$ contains an Abelian fermion parity vortex, but $\C$ does not; according 
        to the list of properties in Sec.~\ref{subsec:fermionicSymmFrac}, this implies that $\C$ has no $v$-type vortices, i.e., $k=0$. According to Theorem~\ref{thm:kerRPerm}, $[\widecheck{\rho}]$ must therefore act as the identity permutation. Exactly two elements of $\Aut_{LR}(\C)$ implement the trivial permutation, namely $[\widecheck{\Upsilon}_\psi]\neq [1]$  in $\C$. In particular, $[\widecheck{\Upsilon}_\psi]$ restricts to $[\Upsilon_\psi] = [1]$ in $\mathcal{C}$. Hence $[\widecheck{\Upsilon}_\psi] \in \ker r$, and $\ker r = \mathbb{Z}_2^{\max(k,1)}=\Z_2$.
        
        \underline{Case 3}: Neither $\widecheck{\Upsilon}_\psi$ nor $\Upsilon_\psi$ respect locality in their respective categories. 
        In this case, $r([\widecheck{\Upsilon}_\psi])=[\Upsilon_\psi] \neq [1]$, so $[\widecheck{\Upsilon}_\psi] \not\in \ker r$. But $[\alpha_\psi]$ is a non-trivial element of $\ker r$, and thus must not equal $[\widecheck{\Upsilon}_\psi]$. Therefore, $[\alpha_\psi]$ has a non-trivial permutation action, that is, $\C_v$ is non-empty and $k>0$. Now consider any braided autoequivalence $[\widecheck{\rho}] \in \Aut_{LR}(\C)$ which implements one of the anyon permutation actions given above. $[\widecheck{\rho}]$ must restrict to either $[1]$ or $[\Upsilon_\psi]$ on $\mathcal{C}$, and $[\widecheck{\Upsilon}_\psi \widecheck{\rho}]$ will restrict to $[\Upsilon_\psi]$ or $[1]$ respectively. Hence exactly one of $[\widecheck{\rho}]$ and $[\widecheck{\Upsilon}_\psi \widecheck{\rho}]$ is in $\ker r$, that is, each of the $2^{k}$ permutation actions above defines a unique element of $\ker r$. If $[\widecheck{\rho}]\in \ker r$, then $r([\widecheck{\rho}]^2)=[1]$, so $[\widecheck{\rho}]^2=1$ as well. Hence $\ker r = \Z_2^k$, concluding our proof.

        \end{proof}

        It is true but not immediately obvious that one can have $k>1$ decoupled blocks of the $v-v$ part of the $S$-matrix; we consider an explicit example with $k=2$ in Sec.~\ref{subsec:k>1}. 
        
        \subsubsection{$\ker r$ when permutations do not uniquely determine $\Aut(\C)$ and $\Aut(\mathcal{C})$}
        \label{subsubsec:kerRNonUnique}
        
        If there is not a one-to-one correspondence between elements of $\Aut(\C)$ (resp. $\Aut(\mathcal{C})$) and anyon permutations which preserve the fusion rules and modular data of $\C$ (resp. $\mathcal{C}$), then we are not generally able to give any further characterization of $\ker r$ as a group beyond the results of the previous subsections.
        
        We do know that Theorem~\ref{thm:kerRPerm} still applies; the anyon permutations of all elements of $\ker r$ are still restricted. However, three possibilities could further complicate the analysis of $\ker r$:
        
        \begin{enumerate}
            \item  There could be an anyon permutation which preserves the modular data of $\C$ but does not correspond to any braided autoequivalence of $\C$. In this case, there must be a map $\widecheck{\rho}$ from $\C$ to an inequivalent theory with the same modular data as $\C$, that is, $\C$ must have a ``modular isotope." This can occur in general~\cite{WenBeyondModularData,BondersonBeyondModular}. In this situation, $\ker r$ may contain fewer than $2^k$ distinct permutation actions.
            
            \item There could be a non-trivial but non-permuting element of $\Aut(\C)$. If such an autoequivalence exists and has a representative in $\ker r$, then each allowed permutation action may determine many elements of $\ker r$, and each allowed permutation action may give rise to a subgroup of $\ker r$ which is larger than $\Z_2$. UMTCs can in general have non-trivial non-permuting autoequivalences~\cite{davydov2014}.
            
            \item There could be a non-trivial but non-permuting element of $\Aut(\mathcal{C})$. If such an autoequivalence exists, then there may be some permutation actions allowed by Theorem~\ref{thm:kerRPerm} which, nevertheless, do not determine an element of $\ker r$ because any braided autoequivalence which implements such a permutation action necessarily restricts to a non-trivial but non-permuting element of $\Aut_{LR}(\mathcal{C})$. In this situation, $\ker r$ may again contain fewer than $2^k$ distinct permutation actions.
        \end{enumerate}
        
        A version of Ref.~\cite{aasen21ferm} posted after the second version of the present paper showed that whenever $\Upsilon_\psi$ respects locality, $\ker r = \{1, [\alpha_\psi]\} \simeq \Z_2$. Conjecture~\ref{conjecture:kerRIsActuallyZ2} remains open in the case where $\widecheck{\Upsilon}_\psi$ violates locality.
        
        We have not found any example where we know that a particular one of the above possibilities is relevant for characterizing $\ker r$. In Section~\ref{subsec:k>1}, we study an example involving two copies of $\SO(3)_3$, for which we suspect one of the above possibilities is occurring.

        \subsection{Relation to Arf invariant on torus}
        \label{subsec:Arf}
        
         In this section, we demonstrate that when $\ker r = \{[1],[\alpha_\psi]\} \simeq \Z_2$, our obstruction $o_2$ leads to a modification of the action of the symmetry on the torus Hilbert space of the SET in a way which is sensitive to the Arf invariant of the spin structure on the torus. Ref.~\cite{delmastro2021} showed that for the special case $G_f= \Z_4^{{\bf T},f}$, such a sensitivity appears in certain systems with a 't Hooft anomaly; where there is overlap, our results agree.
        
        Consider the ground states of the TQFT on a spatial torus, $T^2$. As described in Sec.~\ref{subsec:torusDegeneracy}, the Hilbert space breaks up into four sectors $\mathcal{H}_{\mu,\nu}$, where $\{\mu,\nu\} \in \{0,1\}$ specify a spin structure. Here $\mu = 0$ refers to anti-periodic (Neveu-Schwarz) boundary conditions and $1$ refers to periodic (Ramond) boundary conditions. These are sometimes also referred to as bounding and non-bounding spin structures, respectively.
        
        Here $\mathcal{H}_{\mu,\nu}$ is a $\Z_2$ graded Hilbert space which includes both even and odd fermion number sectors. That is, we can think of this as the TQFT Hilbert space allowing for the possibility of a $\psi$ puncture. We will use the basis Eq.~\ref{eqn:torusBasis} and the description Eq.~\ref{eqn:torusStates} for the torus Hilbert space. Since the torus Hilbert space of the fermionic theory is defined via states of the minimal modular extension $\C$ on the 3-punctured sphere (i.e. its states are defined using simple objects in $\C$), only a (representative) lift $\widecheck{\rho}_{\bf g}$ has a well-defined action on the torus Hilbert space.
        
        Our aim is to show that if there is an $\H^2$ anomaly, then the action of $\widecheck{\rho}_{\bf g}$ on the torus Hilbert space is necessarily deformed by the Arf invariant
        \begin{equation}
            \text{Arf}(\mu,\nu)=\mu \nu
        \end{equation}
        of the spin structure. More precisely, given a state $\ket{\Psi}_{\mu \nu} \in \H_{\mu,\nu}$, the $G_b$ group law of the lift is deformed in the sense
        \begin{equation}
            \widecheck{\rho}_{\bf gh}\widecheck{\rho}_{\bf h}^{-1}\widecheck{\rho}_{\bf g}^{-1}\ket{\Psi}_{\mu \nu} = (\omega_2({\bf g,h}))^F (-1)^{\widetilde{o}_2({\bf g,h})\text{Arf}(\mu,\nu)}\ket{\Psi}_{\mu \nu}.
            \label{eqn:o2Arf}
        \end{equation}
        where $(-1)^F$ is the fermion parity operator.
        The notation $\widetilde{o}_2({\bf g,h}) \in \Z_2 \simeq \{0,1\}$ distinguishes when we are viewing $o_2$ as an element of $\Z_2$ as an additive group from when we view $o_2 \in \Z_2 \simeq \ker r$.
        Eq.~\ref{eqn:o2Arf} extends and sharpens the results of~\cite{delmastro2021}, which considered only the special case $G_f=\Z_4^{{\bf T},f}$.
        
         From the definition,
         \begin{equation}
            \widecheck{\rho}_{\bf gh}\widecheck{\rho}_{\bf g}^{-1}\widecheck{\rho}_{\bf h}^{-1}\ket{\Psi}_{\mu \nu} = \begin{cases}
              o_2({\bf g,h})\ket{\Psi}_{\mu \nu} & \widecheck{\Upsilon}_\psi \text{ respects locality}\\
              o_2({\bf g,h}) \widecheck{\Upsilon}_\psi^{\tilde{\omega}_2({\bf g,h})}\ket{\Psi}_{\mu \nu} & \widecheck{\Upsilon}_\psi \text{ violates locality}
            \end{cases}
         \end{equation}
        Recall that $\tilde{\omega}_2$ is the additive parameterization of $\omega_2$ defined by Eq.~\ref{eqn:w2Def}.
        One can compute directly that $\widecheck{\Upsilon}_\psi$ acts on the torus Hilbert space as fermion parity, that is, it inserts a minus sign on states of the form $\ket{\sigma}_{11}$ and acts as the identity otherwise, so in all other sectors there is no difference whether $\widecheck{\Upsilon}_\psi$ respects or violates locality.
        
        Since $o_2({\bf g,h})$ acts trivially on fusion spaces $V_{ab}^c$ with $a,b,c\in \C_0$, it is immediate that Eq.~\ref{eqn:o2Arf} holds for $\mu = 0$.
        
        Now fix a $v$-type vortex $v$ and consider $\ket{v,\overline{v};1}\in V_1^{v \overline{v}}$; then
        \begin{equation}
            o_2({\bf g,h})\left(\ket{v,\overline{v};1}\right) = U_{o_2({\bf g,h})}\left(v \times \psi^{\widetilde{o}_2},\overline{v} \times \psi^{\widetilde{o}_2};1\right)\ket{v \times \psi^{\widetilde{o}_2},\overline{v} \times \psi^{\widetilde{o}_2};1}.
        \end{equation}
        The $U$ factors obtained from the action of $o_2$ on $V_1^{v \overline{v}}$ and its dual $V_{v \overline{v}}^1$ are complex conjugates and thus will always cancel out, so the action of $o_2({\bf g,h})$ is either trivial or simply interchanges $\left(V_{1}^{v \overline{v}} \otimes V_{v \overline{v}}^1\right)$ and $\left(V_{1}^{v\psi, \overline{v}\psi} \otimes V_{v\psi, \overline{v}\psi}^1\right)$. Said differently,
        \begin{equation}
            o_2({\bf g,h})\ket{v} = \ket{v\times \psi^{\widetilde{o}_2}}
        \end{equation}
        where $\ket{v}$ is a torus state of definite topological charge $v$ piercing the $\alpha$ cycle. Therefore, using Eq.~\ref{eqn:torusBasis}, we find
        \begin{align}
            o_2({\bf g,h})\ket{v}_{1,0}&=\ket{v}_{1,0}\\
            o_2({\bf g,h})\ket{v}_{1,1}&=(-1)^{\widetilde{o}_2({\bf g,h})}\ket{v}_{1,1}
        \end{align}
        which verifies Eq.~\ref{eqn:o2Arf} for the present case.
        
        Finally, we must consider states built from a $\sigma$-type vortex, which  only exist when $\widecheck{\Upsilon}_\psi$ violates locality. The above argument immediately generalizes to the unpunctured states to show that
        \begin{equation}
            o_2({\bf g,h})\ket{\sigma}_{1,0} = \ket{\sigma}_{1,0}
        \end{equation}
        Running a similar argument on the states with a puncture, we find that
        \begin{equation}
            o_2({\bf g,h})\widecheck{\Upsilon}_\psi^{\tilde{\omega}_2({\bf g,h})}\ket{\sigma;\psi}= \omega_2({\bf g,h}) U_{o_2({\bf g,h})}(\sigma,\overline{\sigma};1)U^{\ast}_{o_2({\bf g,h})}(\sigma,\overline{\sigma};\psi)\ket{\sigma;\psi}
            \label{eqn:o2OnPunctured}
        \end{equation}
        Since $o_2$ does not permute $\sigma$, and $\sigma \times \psi = \sigma$, one can show from the consistency conditions that
        \begin{equation}
            U_{o_2({\bf g,h})}(\sigma,\overline{\sigma};1)U^{\ast}_{o_2({\bf g,h})}(\sigma,\overline{\sigma;\psi}) = U_{o_2({\bf g,h})}(\sigma,\psi;\sigma)
            \label{eqn:UsigmaPsiSigmaUU}
        \end{equation}
        
        This last quantity was computed in Eq.~\ref{eqn:UsigmaPsiSigmaMinus1}; combining it with Eqs.~\ref{eqn:o2OnPunctured} and~\ref{eqn:UsigmaPsiSigmaUU}, we obtain Eq.~\ref{eqn:o2Arf} for this last class of states.

        \subsection{Anomaly inflow from bulk (3+1)D FSPT}
        \label{subsec:inflow}
        
        According to the classification of fermion SPTs reviewed in Sec.~\ref{subsec:fermionicSymm}, if an FSPT has $[n_1]=0$, then there is a piece of data $[n_2] \in \H^2(G_b,\Z_2)/\Gamma^2$ in the specification of a (3+1)D FSPT. Let us assume that $\ker r = \Z_2 = \{[1], [\alpha_\psi]\}$; then our obstruction $[o_2]$ is valued in $\H^2(G_b,\Z_2)$. Letting 
        \begin{equation}
            q_{\Gamma^2}:\H^2(G_b,\Z_2) \rightarrow \H^2(G_b,\Z_2)/\Gamma^2,
        \end{equation}
        we argue that the 't Hooft anomaly data $[n_2]$ of the FSET is given by
        \begin{equation}
            [n_2]=q_{\Gamma^2}([o_2])
            \label{eqn:H2Anomaly}
        \end{equation}
        
        In support of this conjecture, we give an argument explicitly describing anomaly inflow using the decorated domain wall construction for FSPTs. This argument will work at the level of cocycles, without any quotient by coboundaries or by $\Gamma^2$; we discuss the $\Gamma^2$ redundancy in the 't Hooft anomaly in Sec.~\ref{sec:o2independence}. We will also discuss the possibility that $\ker r \neq \Z_2$ in Sec.~\ref{subsec:k>1}.
        
        \begin{figure*}
            \centering
            \includegraphics[width=0.7\linewidth]{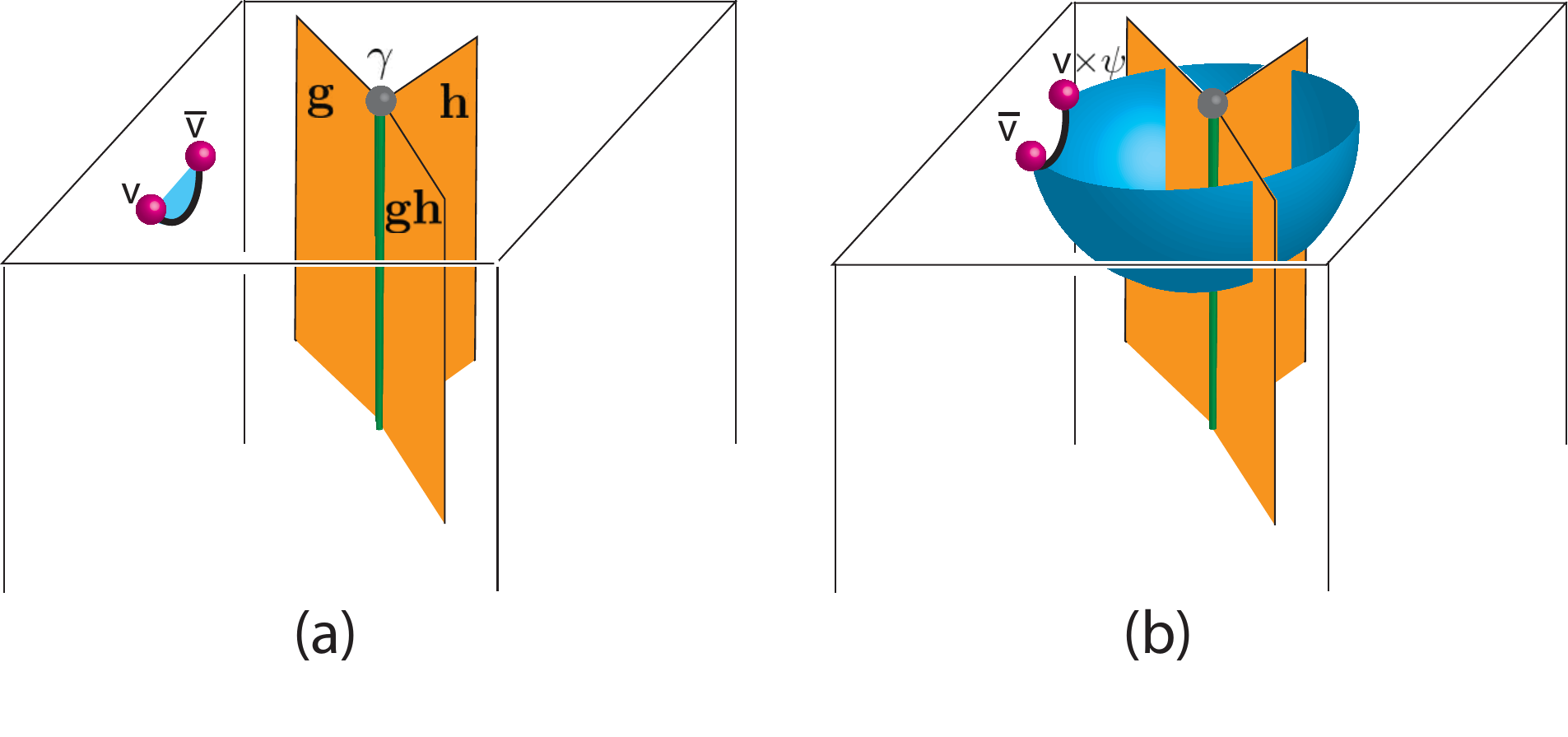}
            \caption{Anomaly inflow for the $\H^2(G_b,\Z_2)$ anomaly. The fermion SET lives on the surface of a (3+1)D fermion SPT. A trijunction of domain walls (orange) with $o_2({\bf g,h}) \neq 0$ is decorated with a Kitaev chain (green) in the bulk with a Majorana fermion $\gamma$ at its endpoint (grey). Braiding a fermion parity vortex $v$ (pink) around the trijunction transforms $v \rightarrow v\times \psi$ according to the surface theory, so re-annihilating the parity vortices leaves behind a fermion. The parity vortices on the surface are endpoints of a fermion parity vortex string that goes through the bulk (thick black line). The parity vortex creation and motion is given by a membrane operator in the bulk (blue surface) which links with the Kitaev chain. The linking of the blue membrane operator with the Kitaev chain induces a $\Z_2^f$ defect in the Kitaev chain, which locally changes the fermion parity, compensating for the additional fermion arising from the transformation $v \rightarrow v \times \psi$. }
            \label{anomaly inflow}
        \end{figure*}
        
        First, we note that the physical meaning of the $2$-cocycle $n_2({\bf g}, {\bf h})$ that characterizes a (3+1)D bulk SPT is as follows. The (3+1)D fermionic SPT state can be considered to be a superposition of all possible networks of codimension-1 domain walls. A non-trivial $n_2({\bf g}, {\bf h})$ means that the codimension-2 junction of three codimension-1 $G_b$ domain walls labeled ${\bf g}$, ${\bf h}$, and ${\bf gh} \in G_b$ are decorated with a (1+1)D Kitaev chain if $n_2({\bf g}, {\bf h})$ is non-trivial.  Therefore, we can consider our system on a 3-dimensional space with boundary, and with a particular choice of domain wall junction, as shown in Fig.~\ref{anomaly inflow}. 
        
        Second, we note that an important property of a (1+1)D Kitaev chain is that it can be coupled to a $\Z_2^f$ gauge field, and we can insert a $\Z_2^f$ symmetry defect (flux). For the Kitaev chain defined on a ring, the $\Z_2^f$ symmetry defect changes the spin structure; that is, it changes
        the boundary conditions from periodic to anti-periodic (or vice versa). It is well-known that changing the spin structure changes the fermion parity of the ground state. 
        To see this concretely, we can consider the Hamiltonian for the Kitaev chain: 
        \begin{align}
            H_K = - i \sum_{j} (u \gamma_{j}' \gamma_{j+1} \sigma_{j,j+1} + v \gamma_{j} \gamma_{j}') . 
        \end{align}
        Here $\gamma_j$, $\gamma_j'$ are independent Majorana fermion operators on site $j$ (i.e. $\{\gamma_j, \gamma_k\} = \{\gamma_j', \gamma_k'\} = 2 \delta_{jk}$ and $\{\gamma_j,\gamma_k'\} = 0$. Note that this system can also be written in terms of complex fermions $c_i = \gamma_i + i \gamma_i'$. Furthermore, we have coupled the system to a $\Z_2$ gauge field $\sigma_{j,j+1} = \pm 1$ on the links $(j,j+1)$. 

        Setting $\sigma_{j,j+1} = 1$ for all $j$, the limits $u/v \ll 1$ and $u/v \gg 1$ realize topologically distinct phases. When the system originates from a model of hopping and pairing of complex fermions $c_i$, it is natural to identify the $u/v \ll 1$ to be the trivial phase and $u/v \gg 1$ to be the topological phase. A hallmark of the topological phase is that the the fermion parity of the ground state changes in the presence of a $\Z_2$ symmetry defect on the link $(j,j+1)$ (which corresponds to setting $\sigma_{j,j+1} = -1$). This can be seen easily in the limit $v = 0$, where the ground state is simply $i \gamma_j' \gamma_{j+1} = \sigma_{j,j+1}$. Therefore taking $\sigma_{j,j+1} \rightarrow - \sigma_{j,j+1}$ changes the local fermion parity $i \gamma_j'\gamma_{j+1}$ by one. 
        
        Next, let us consider fermion parity vortices at the (2+1)D surface, which are endpoints of fermion parity vortex lines that go into the (3+1)D bulk. Let us consider a process, shown in Fig.~\ref{anomaly inflow}, where a fermion parity vortex $v$ at the surface encircles the trijunction where the defects ${\bf g}$, ${\bf h}$, ${\bf gh}$ all meet. Importantly, to come back to the original configuration, the vortex line in the bulk sweeps across a membrane that must necessarily intersect the Kitaev chain on the codimension-2 junction in the bulk. This changes the fermion parity of the Kitaev chain, which must be compensated for by a change in fermion parity on the vortex line. This can then be interpreted as a transformation $v \rightarrow v \times \psi$ in the surface theory. Therefore we see that we indeed obtain the symmetry action $\,^{\overline{\bf g h}} ( \,^{\bf g} (\,^{\bf h} v)) \times \bar{v} = o_2({\bf g}, {\bf h}) \times v \times \bar{v}$, where here we take $o_2({\bf g},{\bf h}) \in \{1, \psi\}$, which is consistent with the action $\,^{\overline{\bf g h}} ( \,^{\bf g} (\,^{\bf h} v)) = o_2({\bf g}, {\bf h}) \times v$. If $v \in \C_\sigma$, this process does not permute the vortices but instead changes the fusion channel of two $\sigma$-type vortices, which is expected because they are braiding with a Majorana zero mode. 
        
        \subsection{Dependence of $[o_2]$ on modular extension}
        \label{sec:o2independence}
        
        To define the obstruction $[o_2] \in \mathcal{H}^2(G_b, \Z_2)$, we have started with symmetry fractionalization defined on $\mathcal{C}$, picked a particular modular extension $\mathcal{C}_\nu$, and attempted to lift the symmetry action $\rho_{\bf g}$ to $\widecheck{\rho}_{\bf g}$. In principle there may be multiple different choices of $\nu$ for which there exists a lift of $\rho_{\bf g}$ to $\widecheck{\rho}_{\bf g}$. As such, $[o_2]$ also has an implicit dependence on $\nu$, which we may write as $[o_2^{(\nu)}]$. This raises the question of how $[o_2^{(\nu)}]$ may change under a valid change of $\nu$, if at all. 
        
        In the case where $\ker r = \Z_2$, then as we discussed in the Sec. \ref{subsec:inflow}, we expect that $q_{\Gamma^2}([o_2]) \in \mathcal{H}^2(G_b, \Z_2)/\Gamma^2$ should be interpreted as the 't Hooft anomaly of the theory. In other words, such a theory exists at the (2+1)D surface of a (3+1)D FSPT characterized by $[n_2] = q_{\Gamma^2}([o_2])$.  Since the modular extension $\nu$ can be changed by layering a (2+1)D invertible topological phase with chiral central charge $c_- = \nu/2$, the bulk (3+1)D system should be independent of $\nu$. This leads to the general expectation that $q_{\Gamma^2}([o_2^{(\nu)}])$ is independent of $\nu$ (for any valid $\nu$).
        
        Recently an updated version of Ref. \cite{aasen21ferm}, motivated by a conjecture of an earlier version of this paper,\footnote{An earlier version of this paper conjectured that $[o_2]$ itself, and not just its image under $q_{\Gamma^2}$, would be independent of $\nu$.} showed that for unitary symmetry groups $G_b$, the following result holds:
         \begin{equation}
            [o_2^{(\nu')}] = [\alpha_\psi]^{(\nu' - \nu)[\omega_2]}[o_2^{(\nu)}],
            \label{eqn:o2ABKResult}
        \end{equation}
        where in the above formula we are interpreting $[\omega_2]$ as $0$ or $1$ depending on whether it is non-trivial in 
        $\H^2(G_b, \Z_2)$.
        This proves the statement that $q_{\Gamma^2}([o_2^{(\nu)}])$ is independent of $\nu$ in the case where $G_b$ is unitary (as long as $\C$ admits a lift of $\rho_{\bf g}$ for all ${\bf g}$).  
        
        In what follows, we will prove that, under some assumptions\footnote{The case we do not prove is when $[\alpha_\psi]$ is non-permuting, i.e., $\widecheck{\Upsilon}_\psi$ violates locality and $\Upsilon_\psi$ respects locality. This last case is proven in v4 of Ref.~\cite{aasen21ferm}, which was posted around the same time as the revision of this paper containing the current version of Thm.~\ref{o2independence_anti}.}, $[o_2^{(\nu)}]$ is independent of modular extension in the case where $G_b$ contains anti-unitary symmetries. We will then provide additional results (already included in the first version of this paper) about how $[o_2^{(\nu)}]$ depends on modular extension when we only restrict attention to how the symmetry actions permute anyons. 
        
        \begin{theorem}
        Suppose $G_b$ contains at least one anti-unitary symmetry. Assume that $\ker r = \Z_2$ and the nontrivial element of $\ker r$ has a nontrivial permutation action. Then the obstruction $[o_2] \in \H^2(G_b,\ker r)$ is identical for every $\nu$ for which there exists, for every ${\bf g} \in G_b$, a lift $\widecheck{\rho}_{\bf g}^{(\nu)}:\C_\nu \rightarrow \C_\nu$ of $\rho_{\bf g}$.
        \label{o2independence_anti}
        \end{theorem}
        
        \begin{proof} 
        Suppose that lifts of $[\rho_{\bf g}]$ exist for two modular extensions $\C_{\nu_1}$ and $\C_{\nu_2}$. Let $r_i$ be the restriction map $r_i : \Aut_{LR}(\C_{\nu_i})\rightarrow \Aut_{LR}(\C_{\nu_i})|_{\mathcal{C}}$. Suppose we have a particular lift $\widecheck{\rho}^{(1)}_{\bf g}$ of a representative $\rho_{\bf g}$ to $\C_{\nu_1}$. Since $G_b$ contains anti-unitary symmetries, the only possible $\C_{\nu_2}$ are obtained from $\C_{\nu_1}$ by stacking the minimal modular extension $\mathcal{I}(8)$ of $\{1,\psi \}$ with central charge $c_- = 4 \mod 8$ and condensing the bound state of the local fermions in the two theories.
        
        The theory $\mathcal{I}(8)$ is equivalent to the $3$-fermion topological order. Stacking and condensing such a phase has the effect of changing the topological twist of the fermion parity vortices by a minus sign,
        \begin{align}
            \theta_x \rightarrow - \theta_x \text{ if } x \in (\C_{\nu_1})_1,
        \end{align}
        while keeping the fusion rules and $S$-matrix invariant.
        
          Therefore, we can use the same anyon labels for $\C_{\nu_1}$ and $\C_{\nu_2}$. In defining the condensation procedure, we need to specify an action of $G_b$ on the 3-fermion topological order $\mathcal{I}(8)$. The only autoequivalences of $\mathcal{I}(8)$ which preserve a choice of physical fermion are the trivial one and the one which permutes its parity vortices. Let
        \begin{equation}
            \lambda_1({\bf g}) = \begin{cases}
              0 & {\bf g} \text{ non-permuting on } \mathcal{I}(8)\\
              1 & {\bf g} \text{ permutes vortices of }\mathcal{I}(8)
            \end{cases}.
        \end{equation}
        Since $\mathcal{I}(8)$ contains an Abelian parity vortex, $\lambda_1$ must be a group homomorphism $\lambda_1: G_b \rightarrow \Z_2$.
        The permutation action of $\widecheck{\rho}_{\bf g}^{(2)}$ on anyon labels is therefore the same as that of $\alpha_\psi^{\lambda_1({\bf g})}\widecheck{\rho}_{\bf g}^{(1)}$. Accordingly, the permutation action of $o_2^{(2)}({\bf g,h})$ is the same as that of $o_2^{(1)}({\bf g,h})\alpha_\psi^{d\lambda_1({\bf g,h})}$. Since we assumed that the only nontrivial element of $\ker r$ is permuting, the permutation action of $o_2^{(2)}$ determines it as an element of $\ker r$. Hence, as cohomology classes, $[o_2^{(2)}]=[o_2^{(1)}]$.
        \end{proof}
        
        In the remaining part of this section, we will provide a number of results about the dependence of $[o_2^{(\nu)}]$ on $\nu$ when we only restrict our attention to how $\rho$ and $o_2$ permute the anyons. Hence we will define $P(\C)$ to be the group of anyon permutations of $\C$ that preserve the modular $S$ and $T$ matrices of $\C$. Furthermore, for $\rho$, $o_2 \in Aut_{LR}(\C)$, we will denote $\underline{\rho}, {\underline{o_2}} \in P(\C)$ their respective permutation actions (and similarly for $P(\mathcal{C})$). Note that if we have a lift $\underline{\widecheck{\rho}} \in P(\C)$ of $\underline{\rho} \in P(\mathcal{C})$ then $\underline{\widecheck{\rho}}$ defines a permutation $\underline{o_2}$ by Eq.~\ref{eqn:o2Def}. 
        
        \begin{proposition}
        Let $\widecheck{\rho}_{\bf g}^{(1)}$ be a lift of $\rho_{\bf g}$ to $\Aut_{LR}(\C_{\nu_1})$ which defines the obstruction $o_2^{(1)} \in Z^2(G_b,\ker r)$. If $G_b$ is unitary, then there is a permutation $\underline{\widecheck{\rho}_{\bf g}^{(2)}} \in P(\C_{\nu_2})$ which lifts $\underline{\rho_{\bf g}}$ for any $\nu_2$. If $\delta \nu = \nu_2 - \nu_1$ is even, then $\underline{o_2^{(2)}} = \underline{o_2^{(1)}}$. If $\delta \nu$ is odd, then $\underline{o_2^{(2)}}$ is the trivial permutation. 
        \label{prop:unitaryH2Trivial}
        \end{proposition}
        
        The proof is quite technical, so we briefly sketch it here and defer the details to Appendix~\ref{app:H2ChangeModularExt}. The main idea is to use the 16-fold way theorem~\cite{bruillard2017a}, which states that a generic minimal modular extension may be derived from a given  $\C_{\nu_1}$ by layering $\C_{\nu_1}$ with a minimal modular extension $\mathcal{I}(\delta \nu)$ of $\{1,\psi\}$ and condensing the bound state of the fermions in the two layers. Using results on anyon condensation~\cite{delmastro2021}, we can derive the modular data of a generic minimal modular extension from the given one $\C_{\nu_1}$. This step is straightforward when $\delta \nu$ is even because $\mathcal{I}(\delta \nu)$ is Abelian, but it is quite involved when $\delta \nu$ is odd so that $\mathcal{I}(\delta \nu)$ is non-Abelian. From the new modular data, we can explicitly construct an element of $P(\C_{\nu_2})$ for each $\widecheck{\rho}_{\bf g}^{(1)}$ and then calculate the permutation action of compositions of these permutations.
        
        Prop.~\ref{prop:unitaryH2Trivial} has the rather surprising corollary:
        \begin{corollary}
        Suppose that in every minimal modular extension of $\mathcal{C}$, every lift of $\underline{\rho_{\bf g}}$ from $P(\mathcal{C})$ to $P(\C)$ defines a lift of the autoequivalence class $[\rho_{\bf g}]$ from $\Aut_{LR}(\mathcal{C})$ to $\Aut_{LR}(\C)$. Further assume that the only possible non-permuting element of $\ker r$ is $[\widecheck{\Upsilon}_\psi]$. Then if $\Upsilon_\psi$ violates locality, $[o_2^{(\nu)}]=+1$ is trivial for all $\nu$. If $\Upsilon_\psi$ respects locality, then $[o_2^{(\nu)}]=+1$ is trivial for the eight minimal modular extensions with only $v$-type vortices.
        \label{corr:H2Independence}
        \end{corollary}
        
        \begin{proof} (Corollary~\ref{corr:H2Independence})
        First consider the case where $\Upsilon_\psi$ violates locality. Then $[\widecheck{\Upsilon}_\psi]$ is not in $\ker r$ for any minimal modular extension $\C_{\nu_2}$. But according to Proposition~\ref{prop:unitaryH2Trivial}, for $\nu_2-\nu_1$ odd, $o_2^{(2)}$ has trivial permutation action, so it must be the identity as an autoequivalence. Now repeat the argument starting from the permutation action $\widecheck{\rho}_{\bf g}^{(2)}$ to see that for $\nu_3 -\nu_2$ odd, that is, $\nu_3-\nu_1$ even, we can again construct a $\widecheck{\rho}_{\bf g}^{(3)}$ such that $o_2^{(3)}$ is also trivial.
        
        If $\Upsilon_\psi$ respects locality, then by a property discussed in Sec.~\ref{subsec:fermionicSymmFrac}, half of the modular extensions of $\mathcal{C}$ contain only $\sigma$-type fermion parity vortices and half contain only $v$-type fermion parity vortices. We may assume without loss of generality that $\C_{\nu_1}$ has only $\sigma$-type fermion parity vortices; if $\C_{\nu_1}$ contains only $v$-type fermion parity vortices, use Prop.~\ref{prop:unitaryH2Trivial} to construct a symmetry action on $\C_{\nu_1+1}$ and reindex $\nu_1 \rightarrow \nu_1 - 1 \mod 16$. Now apply the $\delta \nu$ odd case of Prop~\ref{prop:unitaryH2Trivial}. Then $o_2^{(\nu_1+1)}$ is a nonpermuting element of $\ker r$. Hence $o_2^{(\nu_1+1)}$ is the identity; the reason is that $\C_{\nu_1+1}$ has only $v$-type vortices, so all non-trivial elements of $\ker r$ are permuting (see Sec.~\ref{subsec:kerR}). Using the $\delta \nu$ even case of Prop.~\ref{prop:unitaryH2Trivial}, we obtain a non-anomalous lift for all eight modular extensions $\C_{\nu}$ which have $\nu = \nu_1+1 \mod 2$
        
         The above argument does not extend to $\nu = \nu_1 \mod 2$, since $\ker r$ contains a nontrivial non-permuting element $\widecheck{\Upsilon}_\psi$. 
        \end{proof}

        The interplay of Eq.~\ref{eqn:o2ABKResult} and Corollary~\ref{corr:H2Independence} is quite interesting. Consider a theory $\mathcal{C}$ which obeys the assumptions of Corollary~\ref{corr:H2Independence}. If $\Upsilon_\psi$ respects locality, then Eq.~\ref{eqn:o2ABKResult} and Corollary~\ref{corr:H2Independence} combine to prove that the minimal modular extensions of $\mathcal{C}$ with only $v$-type vortices have no $o_2$ obstruction and the minimal modular extensions with only $\sigma$-type vortices have $[o_2]=[\omega_2]$.
        
        On the other hand, if $\Upsilon_\psi$ violates locality, then these two statements combine to prove that $[\omega_2]=0$. Therefore, if $\Upsilon_\psi$ violates locality and the theory has unobstructed fractionalization with $[\omega_2] \neq 0$, then the assumptions of either Eq.~\ref{eqn:o2ABKResult} (that $\ker r = \Z_2$) or Corollary~\ref{corr:H2Independence} must be violated. We expect in general, given Conjecture~\ref{conjecture:kerRIsActuallyZ2} that $\ker r = \Z_2$, that the latter occurs. 
        
        \subsection{\texorpdfstring{$\SO(3)_3 \times \SO(3)_3$: an example with $k>1$}{SO(3)3xSO(3)3: an example with k>1}}
        \label{subsec:k>1}
    
        As discussed in Sec.~\ref{subsubsec:kerRNonUnique}, we cannot presently exclude the possibility that $\ker r$ is different from $\Z_2$. This might happen because the number of blocks in the $v-v$ portion of the $S$-matrix of $\C$ is $k>1$, or because certain permutations of $\C$ which preserve the modular data do not uniquely determine an element of $\Aut(\C)$. If $\ker r \neq \Z_2$, it is not immediately obvious how to connect the 't Hooft anomaly $[n_2] \in \H^2(G_b,\Z_2)/\Gamma^2$ to the obstruction $[o_2] \in \H^2(G_b, \ker r)$. We do know, however, that $\ker r$ contains a $\Z_2$ subgroup generated by $[\alpha_\psi]$.
        
        Our anomaly arguments in Secs.~\ref{subsec:Arf},\ref{subsec:inflow} are quite general and physical, so we expect that there should be a fully general connection between the 't Hooft anomaly $[n_2]$ and the $\H^2(G_b,\ker r)$ obstruction. However, our arguments only strictly hold when $\ker r = \Z_2$. Assuming Conjecture~\ref{conjecture:kerRIsActuallyZ2} neatly allows $[o_2]$ and $[n_2]$ to always be related. However, if there exists a theory $\C$ has $k>1$, then there is tension; at the level of permutations, Theorem~\ref{thm:kerR} suggests that $\ker r$ may be larger than $\Z_2$. 
        
        To explore this issue, we now present an explicit example with $k=2$ and which has a non-trivial 't Hooft anomaly $[n_2]$. We argue for this example that if Theorem~\ref{thm:kerR} applies, that is, if every permutation allowed by Theorem~\ref{thm:kerRPerm} uniquely defines an element of $\Aut(\C)$ and $\Aut(\mathcal{C})$ contains no non-trivial non-permuting elements, then the obstruction $[o_2] \in \H^2(G_b,\ker r)$ is trivial. Since $[n_2]$ is non-trivial in this example, there would be no connection between $[n_2]$ and $[o_2]$. This is in tension with the physical arguments for Conjecture~\ref{conjecture:kerRIsActuallyZ2}. We give a well-defined but difficult in practice way to test Conjecture~\ref{conjecture:kerRIsActuallyZ2} in this case.
        
         Consider the theory $\mathcal{C} = \SO(3)_3 \boxtimes \SO(3)_3 /\{\psi \psi \sim 1\}$, where the quotient means we condense the bound state of the transparent fermions from each copy of the theory. We take $G_b = \Z_2^{\bf T}$ and $G_f = \Z_4^{{\bf T},f}$. It is easy to check that this theory is super-modular. For this $G_f$, the (3+1)D FSPT classification is given by an element $\mu \in \Z_{16}$, and it is known that $\SO(3)_3$ has a $\mu=3$ anomaly~\cite{tata2021}. Hence $\mathcal{C}$ has a $\mu=6$ anomaly which is associated with the nontrivial element $[n_2] \in \H^2(\Z_2^{\bf T},\Z_2)=\Z_2$. 
        
            \begin{table}
            \centering
\renewcommand{\arraystretch}{1.3}
            \begin{tabular}{@{}clclclc}
            \toprule[2pt]
                Label && $\theta_a$ && $d_a$ &&  $\,^{\bf T}a$ \\ \hline
                $(0,0,0)$ && 1 && 1 &&  $(0,0,0)$ \\
                $(2,0,0)$ && $i$ && $1+\sqrt{2}$ && $(4,0,0)$ \\
                $(0,2,0)$ && $i$ && $1+\sqrt{2}$ && $(0,4,0)$ \\
                $(4,0,0)$ && $-i$ && $1+\sqrt{2}$ && $(2,0,0)$ \\
                $(0,4,0)$ && $-i$ && $1+\sqrt{2}$ && $(0,2,0)$ \\
                $(2,2,0)$ && -1 &&  $(1+\sqrt{2})^2$ && $(2,2,0)$ \\
                $(2,4,0)$ && 1 &&  $(1+\sqrt{2})^2$ && $(2,4,0)$ \\
                $(6,0,0)$ && -1 && 1  && $(6,0,0)$ \\
                $(1,3,\sigma)_+$ && 1 && $2+\sqrt{2}$ && $(3,1,\sigma)_{\pm_1}$ \\
                $(1,3,\sigma)_-$ && 1 &&  $2+\sqrt{2}$ && $(3,1,\sigma)_{\mp_1}$ \\
                $(3,1,\sigma)_+$ && 1 &&  $2+\sqrt{2}$ && $(1,3,\sigma)_{\pm_2}$ \\
                $(3,1,\sigma)_-$ && 1 &&  $2+\sqrt{2}$ && $(1,3,\sigma)_{\mp_2}$ \\
                $(1,1,\sigma)$ && $e^{5\pi i/4}$ &&  $2+2\sqrt{2}$ && $(3,3,\sigma)$ \\
                $(3,3,\sigma)$ && $e^{3\pi i/4}$ && $2+2\sqrt{2}$ && $(1,1,\sigma)$ \\ \bottomrule[2pt]
            \end{tabular}
            \caption{UMTC data and action of time-reversal symmetry for $\C=\mathrm{SU}(2)_6 \times \mathrm{SU}(2)_6 \times \mathrm{Ising}_{-9/2}/\{\psi \psi \sim 1\}$. The parity vortices $(1,3,\sigma)_{\pm}$ are v-type and interchanged by fusion with $\psi$; likewise for $(3,1,\sigma)_{\pm}$. The parity vortices $(1,1,\sigma)$ and $(3,3,\sigma)$ are $\sigma$-type. The notation $\pm_1$ and $\pm_2$ indicate two independent choices of signs in the action of time reversal. All four such permutations complex conjugate the modular data.}
            \label{tab:doubledSU26}
        \end{table}
        
        The minimal modular extension we consider is $\C = \SU(2)_6 \boxtimes \SU(2)_6 \boxtimes \mathcal{I}(-9)/\{\psi \psi \sim 1\}$, where $\mathcal{I}(-9)$ is the $c_-=-9/2$ minimal modular extension of $\{1,\psi\}$ and the quotient means we condense all pairs of Abelian fermions from the different theories. The particle content, topological twists, quantum dimensions, and possible actions of ${\bf T}$ are listed in Table~\ref{tab:doubledSU26}, while the $S$-matrix is given in Appendix~\ref{app:doubledSU26}. The data were derived from a slight generalization of the results of~\cite{delmastro2021}; the un-generalized results are reviewed in Appendix~\ref{app:SmatrixCondensed}.

        Examining the set of permutations which preserve the modular data, one would naively conclude that $\ker r = \Z_2^2$ because there are $k=2$ distinct blocks of the $\C_v$ part of the $S$-matrix, namely one block formed by $(1,3,\sigma)_\pm$ and one block formed by $(3,1,\sigma)_\pm$. There are thus four possible lifts of the permutation action of $\rho_{\bf T}$ to $\C$ which preserve the modular data of $\C$. Two form a group homomorphism $\Z_2^{\bf T} \rightarrow \Aut_{LR}(\mathcal{C})$; their actions are given by a choice of signs:
        \begin{align}
            (1,3,\sigma)_+ &\leftrightarrow (3,1,\sigma)_\pm \nonumber\\
            (1,3,\sigma)_- &\leftrightarrow (3,1,\sigma)_\mp \nonumber \\
            (1,1,\sigma) &\leftrightarrow (3,3,\sigma),
            \label{eqn:doubledSU26homomorphism}
        \end{align}
        where the particle content is labeled as in Appendix~\ref{app:doubledSU26}. The other two possible lifts of $\rho_{\bf T}$ form a group homomorphism $\Z_4^{\bf T} \rightarrow \Aut_{LR}(\C)$:
        \begin{align}
            (1,3,\sigma)_+ \rightarrow (3,1,\sigma)_\pm \rightarrow (1,3,\sigma)_- &\rightarrow (3,1,\sigma)_\mp \rightarrow (1,3,\sigma)_+ \nonumber \\
            (1,1,\sigma) &\leftrightarrow (3,3,\sigma)
            \label{eqn:nonHomomorphism}
        \end{align}
        Within each pair, the two possible lifts differ by the permutation action of $\alpha_\psi$.
        
        Choosing a lift in Eq.~\ref{eqn:doubledSU26homomorphism}, the permutation action of $o_2({\bf g}, {\bf h}) \in \ker r$ is trivial. One can check that $\Upsilon_\psi$ does not respect locality in this theory, so $\widecheck{\Upsilon}_\Psi$ is not in $\ker r$; accordingly, if we assume that permutations which preserve the modular data uniquely determine an element of $\Aut(\C)$, $[o_2]$ is trivial. However, as argued above, $[n_2]$ is non-trivial. We therefore expect that permutations do not uniquely determine an element of $\Aut(\C)$ in this theory.
        
        In the cases in Eq.~\ref{eqn:nonHomomorphism} where the lift is not a group homomorphism $\Z_2^{\bf T}\rightarrow \Aut_{LR}(\C)$, we see explicitly that, as permutations,
            \begin{equation}
                \widecheck{\rho}_{\bf T}^2 = \alpha_\psi.
            \end{equation}
        Therefore, if only the permutations in Eq.~\ref{eqn:nonHomomorphism} define valid autoequivalences of $\C$ and $\ker r = \Z_2$, we would indeed obtain $[o_2]=[n_2]$ for $\mathcal{C}=\SO(3)_3^2$ (recall that since $G_b$ contains anti-unitary symmetries, $\Gamma^2$ is trivial so $q_{\Gamma^2}$ is the identity map). We speculate that this is the case; in order to check this speculation, one would need to solve for the $F$- and $R$-symbols of $\C$, which is a non-trivial task, and then directly attempt to solve for the $U$-symbols for each permutation action. We expect that a solution for the $U$-symbols exists only for the permutations in Eq.~\ref{eqn:nonHomomorphism}.
        
        \section{\texorpdfstring{$\H^3(G_b,\Z_2)$}{H3(Gb, Z2)} obstruction}
        \label{sec:H3}
        
        Suppose that the $\H^2(G_b,\ker r)$ obstruction vanishes, so that we may lift a group homomorphism $[\rho_{\bf g}]:G_b \rightarrow \Aut_{LR}(\mathcal{C})$ to a group homomorphism $[\widecheck{\rho}_{\bf g}] : G_b \rightarrow \Aut_{LR}(\C)$. Given a symmetry fractionalization pattern on $\mathcal{C}$, we ask whether or not that fractionalization pattern can be lifted to $\C$. We will show that there is an obstruction to this process valued in $\H^3(G_b,\Z_2)$. Our discussion will very similar to that of Ref.~\cite{fidkowski2018}, but~\cite{fidkowski2018} assumed $G_f = G_b \times \Z_2$ and made some technical assumptions which are known to fail in certain cases. We will use our general understanding of fermionic symmetry fractionalization to remove those technical assumptions.
        
        Note also that $[\widecheck{\rho}_{\bf g}]$ defines a $G_b$ symmetry action on a UMTC $\C$, so there may be an obstruction to localizing $[\widecheck{\rho}_{\bf g}]$, that is, to finding any symmetry fractionalization on $\C$ irrespective of whether it matches the symmetry fractionalization on $\mathcal{C}$. This obstruction is valued in $\H^3(G_b,\widecheck{\A})$ and can be computed in the standard way for bosonic SETs, see~\cite{barkeshli2019}. 
        We will show that this obstruction is in fact determined by the $\H^3(G_b,\Z_2)$ obstruction. 
        
        \subsection{\texorpdfstring{$\H^3(G_b,\Z_2)$}{H3(Gb, Z2)} anomaly}
        
        We start with symmetry fractionalization data, which we choose to characterize by $\omega_a({\bf g,h})\in C^2(G_b,K(\mathcal{C}))$ satisfying Eq.~\ref{eqn:OmegaEqualsDomega}. 
        In the present language, $\Omega_a \in Z^3(G_b,K(\mathcal{C}))$. The gauge freedom $\nu_a$ appearing in Eq.~\ref{eqn:OmegaomegaGaugeFreedom} is an element of $K(\mathcal{C})$, so only $[\Omega_a]\in \H^3(G_b,K(\mathcal{C}))$ is gauge-invariant.
        
        We are also given a lift $[\widecheck{\rho}]$ of the symmetry action $\rho$ to $\C$. A representative $\widecheck{\rho}_{\bf g}$ determines $\widecheck{\kappa}_{\bf g,h}$ via
        \begin{equation}
            \widecheck{\kappa}_{\bf g,h}\widecheck{\rho}_{\bf g}\widecheck{\rho}_{\bf h} = \widecheck{\rho}_{\bf gh}.
        \end{equation}
        Restricting this equation to $\mathcal{C}$, we find that
        \begin{equation}
            r(\widecheck{\kappa}_{\bf g,h})=\kappa_{\bf g,h}.
        \end{equation}
        Hence, given a decomposition $\widecheck{\beta}_a({\bf g,h})$ of $\widecheck{\kappa}_{\bf g,h}$ as a natural isomorphism, we can simply restrict these $\widecheck{\beta}_a$ to $\mathcal{C}$ to obtain a valid gauge choice for the decomposition $\beta_a({\bf g,h})$ of $\kappa_{\bf g,h}$. We will always work in this gauge where $\widecheck{\beta}_a$ lifts $\beta_a$. In this gauge the function $\widecheck{\Omega}_a({\bf g,h,k})\in Z^3(G_b,K(\C))$ lifts $\Omega_a({\bf g,h,k}) \in Z^3(G_b,K(\mathcal{C}))$.
        
        Note the logic here - we are using the a gauge-fixing (of the $\nu$ type) of $\widecheck{\beta}_a$ to determine a gauge-fixing (again of the $\nu$ type) of $\beta_a$. In general, not every gauge choice on $\mathcal{C}$ allows $\beta_a$ and $\Omega_a$ to be lifted to $\C$, specifically when $\widecheck{\Upsilon}_\psi$ violates locality but $\Upsilon_\psi$ respects locality.
        
        We now ask whether or not there exists a lift $\widecheck{\omega}_a \in C^2(G_b,K(\C))$ such that
        \begin{equation}
            \widecheck{\Omega}_a = d\widecheck{\omega}_a
            \label{eqn:CheckOmegaEqualsdCheckomega}
        \end{equation}
        for all $a \in \C$ and such that $r(\widecheck{\omega}_a)=\omega_a$.
        
       For the moment, we assume that at least one lift $\widecheck{\omega}_a \in C^2(G_b,K(\C))$ of $\omega_a$ exists. We will prove that such a lift exists later; the reason depends on whether $\Upsilon_\psi$ and $\widecheck{\Upsilon}_\psi$ respect or violate locality.
        
        We next claim that there always exists exactly two lifts $\widecheck{\omega}_a$ of $\omega_a$ (which need not, a priori, satisfy Eq.~\ref{eqn:CheckOmegaEqualsdCheckomega}). There always exists an element $p_a \in K(\C)$ defined by
        \begin{equation}
            p_a = M_{a,\psi} = \begin{cases}
              1 & \text{ if }a \in \C_0\\
              -1 & \text{ if }a \in C_1
            \end{cases}.
        \end{equation}
        Hence, given a lift $\widecheck{\omega}_a({\bf g,h})$, there is always another one $\widecheck{\omega}_a p_a^{\alpha({\bf g,h})}$ with $\alpha \in C^2(G_b,\Z_2)$ (here $\Z_2 = \{0,1\})$. Also, we may always write
        \begin{align}
            \widecheck{\Omega}_a &= M_{a,\widecheck{\cohosub{O}}}\\
            \widecheck{\omega}_a &= M_{a,\widecheck{\cohosub{w}}}
        \end{align}
        for some $\widecheck{\coho{w}},\widecheck{\coho{O}} \in \widecheck{\A}$ because $\C$ is modular; the $p_a$ freedom mentioned above amounts to changing $\widecheck{\coho{w}} \rightarrow \widecheck{\coho{w}} \times \psi$. By super-modularity of $\mathcal{C}$, the only element of $\C$ which braids trivially with all of $\mathcal{C}_0$ is $\psi$, so the two lifts defined by $\widecheck{\coho{w}}$ and $\widecheck{\coho{w}}\times \psi$ are the only ones which restrict to $\omega_a$ on $\mathcal{C}$. Hence, if a lift exists, there are exactly two such lifts.
        
        Choose one of these lifts, which we call $\widecheck{\omega}_a$. We must ask if Eq.~\ref{eqn:CheckOmegaEqualsdCheckomega} is satisfied. Define
        \begin{equation}
            \tilde{\Omega}_a({\bf g,h,k}) = \widecheck{\Omega}_a({\bf g,h,k}) \left(d\widecheck{\omega}_a({\bf g,h,k})\right)^{-1}.
        \end{equation}
        Certainly $\tilde{\Omega}_a \in Z^3(G_b,K(\C))$ since both $\widecheck{\Omega}_a$ and $d\widecheck{\omega}_a$ are in $Z^3(G_b,K(\C))$. Also, Eq.~\ref{eqn:CheckOmegaEqualsdCheckomega} is satisfied for all $a \in \C$ if and only if $\tilde{\Omega}_a = 1 \in Z^3(G_b,K(\C))$. This will not generally be the case; however, by definition Eq.~\ref{eqn:CheckOmegaEqualsdCheckomega} is satisfied for all $a \in \mathcal{C}$. Hence $\tilde{\Omega}_a = 1$ for all $a \in \mathcal{C}$ and therefore 
        \begin{equation}
            \tilde{\Omega}_a({\bf g,h,k})=M_{a,o_3({\bf g,h,k})}
        \end{equation}
        where one can check that 
        \begin{equation}
            o_3({\bf g,h,k}) = \widecheck{\coho{O}}({\bf g,h,k}) \times \overline{d\widecheck{\coho{w}}({\bf g,h,k})} \in \{1,\psi\} \simeq \Z_2.
        \end{equation}
        By straightforward computation, $d\tilde{\Omega}_a=1$, which implies $do_3 = 1$. We have the freedom to choose a different lift $\widecheck{\omega}_a({\bf g,h}) p_a^{\alpha({\bf g,h})}$, which modifies $\tilde{\Omega}_a \rightarrow \tilde{\Omega}_a p_a^{d\alpha}$, that is, it changes $o_3$ by a $\Z_2$-coboundary. Therefore, $[o_3] \in \H^3(G_b,\Z_2)$ is a well-defined cohomology class independent of the choice of lift $\widecheck{\omega}_a$. Also, if $[o_3]=1$, then there exists a representative $\tilde{\Omega}_a=1$, i.e., some lift satisfies Eq.~\ref{eqn:CheckOmegaEqualsdCheckomega}. We therefore see that $[o_3] \in \H^3(G_b,\Z_2)$ is the obstruction to lifting symmetry fractionalization from $\mathcal{C}$ to all of $\C$. 
        
        Note that we do not need to make any explicit reference to $G_f$, which is encoded via the constraint $\eta_\psi=\omega_2$. The lifted symmetry fractionalization data automatically agrees with the symmetry fractionalization data on $\mathcal{C}$ and therefore also obeys the constraint.
        
        The above argument is essentially a reformulation of that of Ref.~\cite{fidkowski2018}. Both our argument and that of Ref.~\cite{fidkowski2018} rely on an assumption that a lift $\widecheck{\omega}_a$ always exists. At the cochain level (i.e. in an arbitrary gauge), such a lift need not exist; for example, with $\mathcal{C}=\{1,\psi\}$ and $\C=\mathrm{Ising}=\{1,\psi\}$, there is no lift of the function $\omega_a \in K(\mathcal{C})$ with $\omega_1=1, \omega_\psi = -1$ to $K(\C)$; Ref.~\cite{fidkowski2018} observes this failure but does not go further. We now carefully prove our assumptions, specifically that there always exists some gauge in which both $\Omega_a$ and $\omega_a$ lift, proving that the above definition of $[o_3]$ is always valid.
     
        We consider different cases depending on whether $\widecheck{\Upsilon}_\psi$ and $\Upsilon_\psi$ respect locality. We will heavily rely on results from~\cite{bulmash2021} summarized in Sec.~\ref{subsec:fermionicSymmFrac} which relate whether or not $\Upsilon_\psi$ respects locality to various properties of $\mathcal{C}$ and $\C$.
        
        One might ask separately if there could be an $\H^3(G_b,\widecheck{\A})$ obstruction to defining \textit{any} $G_b$ symmetry fractionalization pattern on $\C$, whether it agrees with the symmetry fractionalization on $\mathcal{C}$ or not. Obviously if the $\H^3(G_b,\Z_2)$ obstruction $[\tilde{\Omega}]$ vanishes, there cannot be any such obstruction. We claim more generally that $[\tilde{\Omega}]$ actually determines the $\H^3(G_b,\widecheck{\A})$ obstruction, and in the following we explain the relationship in each case. 
        
        \subsubsection{Case: \texorpdfstring{$\Upsilon_\psi$}{Ypsi} violates locality}
        
        If $\Upsilon_\psi$ violates locality, then $K(\mathcal{C})=K_+(\mathcal{C})$, that is, every set of phases which obey the fusion rules on $\mathcal{C}$ are $+1$ on the fermion, and also $\widecheck{\A}=\A$. Hence $\Omega_a$ and $\omega_a$ are in $K_+(\mathcal{C})$, so
        \begin{align}
            \Omega_a &= M_{a,\cohosub{O}} \nonumber \\
            \omega_a &= M_{a,\cohosub{w}} \label{eqn:omegaMutualStats}
        \end{align}
        for $\coho{O} \in Z^3(G_b,\A/\{1,\psi\})$ and $\coho{w} \in C^2(G_b,\A/\{1,\psi\})$. We can then lift $\Omega_a$ and $\omega_a$ to $\C$ straightforwardly by choosing one of two lifts $\widecheck{\coho{O}}$ of $\coho{O}$ to $Z^3(G_b,\mathcal{A})$ and one of two lifts $\widecheck{\coho{w}}$ of $\coho{w}$ to $C^2(G_b,\mathcal{A})$, which allows us to extend Eq.~\ref{eqn:omegaMutualStats} to all $a \in \C$. The lifted symmetry action $[\widecheck{\rho}]$ on $\C$, which we have already taken as a given, determines which of the two lifts $\widecheck{\Omega}_a$ we must use, while there is freedom in which lift of $\omega_a$ we choose using $p_a$, as discussed previously.
        
        Having proven that a lift exists, we now discuss the $\H^3(G_b,K(\C)) \simeq \H^3(G_b,\widecheck{\A})$ obstruction. The short exact sequence
        \begin{equation}
            1 \rightarrow \Z_2 = \{1,\psi\} \rightarrow \A \rightarrow \A/\{1,\psi\} \rightarrow 1
        \end{equation}
        induces a map
        \begin{equation}
            i: \H^3(G_b,\Z_2) \rightarrow \H^3(G_b,\A)
        \end{equation}
        that is part of the long exact sequence
        \begin{equation}
            \cdots \H^2(G_b,\A/\{1,\psi\}) \stackrel{\delta}{\rightarrow} \H^3(G_b,\Z_2) \stackrel{i}{\rightarrow} \H^3(G_b,\A) \stackrel{q}{\rightarrow} \H^3(G_b,\A/\{1,\psi\})\rightarrow\cdots 
        \end{equation}
        The fact that $\Upsilon_\psi$ violates locality means that no minimal modular extension of $\mathcal{C}$ contains an Abelian parity vortex; therefore, $\widecheck{\A} = \A$. By construction, then, $i([o_3])=[\widecheck{\coho{O}}]$, which is the actual $\H^3(G_b,\widecheck{\A})$ obstruction to symmetry localization. 
        
        Therefore, the $\H^3(G_b,\Z_2)$ obstruction actually determines the $\H^3(G_b,\widecheck{A})$ obstruction of the lifted theory in the following sense. One possibility is that $[o_3] \in \ker i$, in which case symmetry localization is not obstructed on $\C$ even if the lift of the particular symmetry fractionalization pattern in question is obstructed. The other possibility is that $i([o_3)$ is nontrivial, in which case there is an obstruction to lifting the given symmetry fractionalization pattern on $\mathcal{C}$ to $\C$ simply because there is no consistent symmetry fractionalization pattern on $\C$ at all.
        
        We note that the fact that symmetry fractionalization is unobstructed on $\mathcal{C}$, i.e. that $[\coho{O}]= q([\widecheck{\coho{O}}])$ is trivial, means that $[\widecheck{\coho{O}}]\in \ker q = \mathrm{im}\phantom{i} i$, which is why we could always find $[o_3] \in \H^3(G_b,\Z_2)$ with $i(o_3])=[\widecheck{\coho{O}}]$.
        
        \subsubsection{Case: \texorpdfstring{$\Upsilon_\psi$}{Ypsi} and \texorpdfstring{$\widecheck{\Upsilon}_\psi$}{checkYpsi} respect locality}
        
        We show that every element of $K(\mathcal{C})$ has a lift to $K(\C)$, which implies that $\Omega_a$ and $\omega_a$ lift to $K(\C)$. First suppose that $\zeta_a \in K_+(\mathcal{C})$; then
        \begin{equation}
            \zeta_a = M_{a,x}
            \label{eqn:zetaAMutualStats}
        \end{equation}
        for some $x \in \A/\{1,\psi\}$. By simply choosing a representative of $x \in \A$, we obtain a lift of $\zeta_a$ by extending Eq.~\ref{eqn:zetaAMutualStats} to all $a \in \C$. As usual, there is another lift related by sending $x \rightarrow x \times \psi$, or equivalently by modifying the lift of $\zeta_a$ by $p_a$.
        
        Suppose instead that $\zeta_a \in K_-(\mathcal{C})$. Since $\widecheck{\Upsilon}_\psi$ respects locality, $\C$ contains an Abelian fermion parity vortex; call such a parity vortex $v$, and define $\lambda_a = M_{a,v} \in K_-(\mathcal{C})$ and $\widecheck{\lambda}_a = M_{a,v} \in K_-(\C)$. Then $\lambda^{-1}_a \zeta_a \in K_+(\mathcal{C})$ and, as we have already shown, has exactly two possible lifts $\widecheck{(\lambda \zeta)}_a=M_{a,\widecheck{x}}$, where the two possible choices of $\widecheck{x} \in \A$ differ by a fermion, to $K_+(\C)$. Hence there are exactly two (distinct) possible lifts $\widecheck{\lambda}_a\widecheck{(\lambda^{-1} \zeta)}_a=M_{a,\widecheck{x}\times v}$ of $\zeta_a$ to $K_-(\C)$. These two lifts differ by changing $\widecheck{x} \times v$ by a fermion, or equivalently by modifying the lift by $p_a$.
        
        This proves that every element of $K(\mathcal{C})$ lifts to $K(\C)$, as desired.
        
        Our argument that the $\H^3(G_b,\Z_2)$ obstruction determines the $\H^3(G_b,\widecheck{\A})$ obstruction for general symmetry fractionalization on $\C$ carries through from the case where $\Upsilon_\psi$ violates locality. The only difference is that we start from a short exact sequence
        \begin{equation}
            1 \rightarrow \Z_2 \rightarrow K(\C)=\widecheck{\A} \rightarrow K(\mathcal{C}) \rightarrow 1.
        \end{equation}
        Using the fact that the (bosonic) obstruction to symmetry fractionalization on $\mathcal{C}$ is valued in $\H^3(G_b,K(\mathcal{C}))$, the rest of the argument carries through \textit{mutatis mutandis}.
        
        \subsubsection{Case: \texorpdfstring{$\Upsilon_\psi$}{Ypsi} respects locality but \texorpdfstring{$\widecheck{\Upsilon}_\psi$}{checkYpsi} does not}
        
        This case is a bit more subtle than the others because $\C$ does not contain an Abelian fermion parity vortex; hence $\widecheck{\A}=\A$, and every element of $K(\C)$ restricts to an element of $K_+(\mathcal{C})$. As such, $K_+(\C)/\Z_2 = K_+(\mathcal{C})$, where the $\Z_2$ subgroup is generated by $p_a=M_{a,\psi}$.
        
        In general, $\Omega_a$ and $\omega_a$ take values in $K(\mathcal{C})$, not necessarily in $K_+(\mathcal{C})$, so they need not have lifts to $K(\C)$. However, our gauge fixing on $\mathcal{C}$ guaranteed that $\widecheck{\Omega}_a$ lifts $\Omega_a$. We just need to check that in this gauge, $\omega_a$ also has a lift. 
        
        By Eq.~\ref{eqn:modifiedHomomorphismCheck}, $[\widecheck{\kappa}_{\bf g,h}] = [\widecheck{\Upsilon}_\psi]^{\tilde{w}_2({\bf g,h})}$. Hence
        \begin{equation}
            \widecheck{\beta}_\psi({\bf g,h})=\beta_\psi({\bf g,h})=\omega_2({\bf g,h})
        \end{equation}
        in this gauge. Since $\eta_\psi({\bf g,h})=\omega_2$ as well, we conclude using Eq.~\ref{eqn:etaDef} that 
        \begin{equation}
            \omega_\psi({\bf g,h})=+1,
        \end{equation}
        that is, $\omega_a \in K_+(\mathcal{C})$. Hence Eq.~\ref{eqn:omegaMutualStats} applies with $\coho{w}\in C^3(G_b,\A/\{1,\psi\})$. As in the previous cases, then, we may lift $\omega_a$ to $K(\C)$ by extending Eq.~\ref{eqn:omegaMutualStats} to all of $\C$, with a $C^2(G_b,\Z_2)$ choice for $\widecheck{\omega}_a$.
        
        The argument that the $\H^3(G_b,\Z_2)$ obstruction determines the $\H^3(G_b,\widecheck{\A})$ obstruction is identical to the case in which $\Upsilon_\psi$ violates locality; the argument in the latter case really only used the fact that $\widecheck{\A}=\A$ for the minimal modular extension in question, which is true whenever $\widecheck{\Upsilon}_\psi$ violates locality.
        
        \subsection{Dependence on symmetry fractionalization class}
        
        We can see explicitly how the $\mathcal{H}^3(G_b,\Z_2)$ anomaly depends on the symmetry fractionalization class on $\mathcal{C}$ as follows. As discussed in~\cite{bulmash2021}, symmetry fractionalization on $\mathcal{C}$ is an $\mathcal{H}^2(G_b,\A/\{1,\psi\})$ torsor, that is, changing symmetry fractionalization classes on $\mathcal{C}$ amounts to shifting $\w({\bf g,h})\rightarrow \w'({\bf g,h}) = \w({\bf g,h})\times \coho{t}({\bf g,h})$ 
        with $\coho{t} \in Z^2(G_b,\A/\{1,\psi\})$. The new symmetry fractionalization pattern depends only on the cohomology class $[\coho{t}]$. Given a particular lift $\widecheck{\w}$ of $\w$, then, we can pick a particular lift $\widecheck{\coho{t}}$ of $\coho{t}$ to obtain a lift 
        $\widecheck{\w}'$ of $\w'$.
        A representative of the obstruction class
        $[o_3']$ corresponding to $\widecheck{\w}'$
        is thus given by
        \begin{equation}
            o_3' = \widecheck{\O} \times \overline{d\widecheck{\w}'} = o_3 \times \overline{d\widecheck{\coho{t}}}
        \end{equation}
        Changing $\coho{t}$ by a coboundary in $\A/\{1,\psi\}$ leaves $d\coho{t}$ invariant and thus can affect $o_3$
        only through the choice of lift to $d\widecheck{\coho{t}}$. Changing the lift $\widecheck{\coho{t}}$ changes $o_3$ 
        by a $\Z_2$-coboundary, so 
        $[o_3]$ is independent of the lift.  Certainly $d\widecheck{\coho{t}} \in B^3(G_b,\A)$, but generically $d\widecheck{\coho{t}} \not\in B^3(G_b,\Z_2)$. As such, $[d\widecheck{\coho{t}}]$ is not generally trivial in $\mathcal{H}^3(G_b,\Z_2)$ and thus can change the obstruction class.
        
        \subsection{Example: two layers of semion-fermion}
        
        We consider two layers of semion-fermion topological order with $G_f=\Z_4^{{\bf T},f}=\Z_2^{\bf T} \rtimes \Z_2^f$ symmetry. The particles in $\mathcal{C}$ are generated by two Abelian semions $s_1,s_2$ with $\theta_{s_i}=+i$ and trivial mutual braiding and a transparent fermion $\psi$. Time-reversal acts as $\,^{\bf T}s_i=\psi s_i$.
        
        Up to gauge transformations, one can check that, with the anyons ordered $1,s_1,s_2,s_1s_2,\psi,\psi s_1, \psi s_2, \psi s_1 s_2$, then
        \begin{equation}
            U_{\bf T}(a,b;a\times b) = \begin{pmatrix}
            1 & 1 & 1 & 1 & 1 & 1 & 1 & 1\\
            1 & 1 & 1 & 1 & 1 & 1 & 1 & 1\\
            1 & -1 & 1 & -1 & 1 & -1 & 1 & -1\\
            1 & -1 & 1 & -1 & 1 & -1 & 1 & -1\\
            1 & -1 & -1 & 1 & 1 & -1 & -1 & 1\\
            1 & -1 & -1 & 1 & 1 & -1 & -1 & 1\\
            1 & 1 & -1 & -1 & 1 & 1 & -1 & -1\\
            1 & 1 & -1 & -1 & 1 & 1 & -1 & -1
            \end{pmatrix}
        \end{equation}
        and $\eta_a({\bf T,T}) = \{ 1,i,i,1,-1,-i,-i,-1 \}$ in the same order. There is also a solution with $\eta_a \rightarrow \eta_a^{\ast}$, but it behaves similarly. We can calculate $\kappa$ from $U$ by
        \begin{equation}
            \kappa_{\bf T,T}(a,b;a \times b) = U_{\bf T}^{\ast}(a,b;a \times b)U_{\bf T}^{\ast}(\,^{\bf T}a,\,^{\bf T}b;\,^{\bf T}(a \times b))
        \end{equation}
        from which we obtain $\beta_a({\bf T,T}) = \{1,i,i,1,-1,-i,-i,-1\}$ in our preferred gauge. From this we find 
        \begin{equation}
            \omega_a({\bf T,T}) = \frac{\beta_a({\bf T,T})}{\eta_a({\bf T,T})} = +1 \text{ for all }a
        \end{equation}
        
        We now need to obtain $\O$ of the modular extension. The modular extension we consider is the tensor product $\U_2 \times \U_2 \times \U_{-2} \times \U_{-2}$, whose quasiparticles are generated by $s_1,s_2,v$, and $\psi v$, where $v$ is a semionic fermion-parity vortex with $\theta_v=-i$. Note that $\theta_{\psi v}=-i$ as well due to the nontrivial braiding. Coincidentally this theory is two copies of the double-semion theory, but that fact plays no role here. Note that $\widecheck{\Upsilon}_\psi$ respects locality since there is an Abelian fermion parity vortex.
        
        Our earlier expression for $\omega$ can now be written
        \begin{equation}
            \omega_a = M_{a,1},
        \end{equation}
        for $a \in \mathcal{C}$, up to a fermion. That is, $\w({\bf T,T}) = 1$.
        
        There are two options for the extension of the action of time-reversal. Note that $s_1s_2$ is a fermion, so $vs_1s_2$ and $v\psi s_1 s_2$ are both semions with $\theta=+i$. Either we can set $v \rightarrow v s_1s_2$ or $v \rightarrow v\psi s_1 s_2$. We choose the latter permutation action first.
        
        Solving the $U-F$ and $U-R$ consistency equations by computer, with the anyons ordered
        \begin{equation}
        1,s_1,s_2,s_1s_2,v,vs_1,vs_2,vs_1s_2,v\psi,v\psi s_1, v\psi s_2, v\psi s_1s_2, \psi, \psi s_1, \psi s_2, \psi s_1 s_2, \nonumber
        \end{equation}
        we obtain
        \setcounter{MaxMatrixCols}{20}
        \begin{equation}
            U_{\bf T}(a,b;a\times b) = \begin{pmatrix}
            1 & 1 & 1 & 1 & 1 & 1 & 1 & 1 & 1 & 1 & 1 & 1 & 1 & 1 & 1 & 1\\
            1 & 1 & 1 & 1 & 1 & 1 & 1 & 1 & 1 & 1 & 1 & 1 & 1 & 1 & 1 & 1\\
            1 & -1 & 1 & -1 & -1 & 1 & -1 & 1 & -1 & 1 & -1 & 1 & 1 & -1 & 1 & -1\\
            1 & -1 & 1 & -1 & 1 & -1 & 1 & -1 & 1 & -1 & 1 & -1 & 1 & -1 & 1 & -1\\
            1 & 1 & -1 & -1 & -1 & -1 & 1 & 1 & -1 & -1 & 1 & 1 & 1 & 1 & -1 & -1\\
            1 & 1 & 1 & 1 & -1 & -1 & -1 & -1 & 1 & 1 & 1 & 1 & -1 & -1 & -1 & -1\\
            1 & 1 & -1 & -1 & 1 & 1 & -1 & -1 & -1 & -1 & 1 & 1 & -1 & -1 & 1 & 1\\
            1 & 1 & 1 & 1 & -1 & -1 & -1 & -1 & -1 & -1 & -1 & -1 & 1 & 1 & 1 & 1\\
            1 & 1 & -1 & -1 & 1 & 1 & -1 & -1 & -1 & -1 & 1 & 1 & -1 & -1 & 1 & 1\\
            1 & 1 & 1 & 1 & -1 & -1 & -1 & -1 & 1 & 1 & 1 & 1 & -1 & -1 & -1 & -1\\
            1 & 1 & -1 & -1 & 1 & 1 & -1 & -1 & -1 & -1 & 1 & 1 & -1 & -1 & 1 & 1\\
            1 & 1 & 1 & 1 & 1 & 1 & 1 & 1 & 1 & 1 & 1 & 1 & 1 & 1 & 1 & 1\\
            1 & -1 & -1 & 1 & -1 & 1 & 1 & -1 & -1 & 1 & 1 & -1 & 1 & -1 & -1 & 1\\
            1 & -1 & -1 & 1 & 1 & -1 & -1 & 1 & 1 & -1 & -1 & 1 & 1 & -1 & -1 & 1\\
            1 & 1 & -1 & -1 & -1 & -1 & 1 & 1 & -1 & -1 & 1 & 1 & 1 & 1 & -1 & -1\\
            1 & 1 & -1 & -1 & -1 & -1 & 1 & 1 & -1 & -1 & 1 & 1 & 1 & 1 & -1 & -1
            \end{pmatrix}
        \end{equation}
        from which we obtain $\kappa$ and 
        \begin{equation}
            \beta_a = \{1,i,i,1,i,-1,1,-i,-i,1,-1,i,1,i,i,1\}
        \end{equation}
        up to a gauge transformation. Then
        \begin{equation}
            \Omega_a = \{1,1,1,1,-1,-1,-1,-1,-1,-1,-1,-1,1,1,1,1\} = M_{a,\psi}
        \end{equation}
        so $\widecheck{\O} = \psi$. Note that $\widecheck{\O}$ is not ambiguous by a fermion since it is defined for all of $\C$. We can now compute, whether we choose $\w({\bf T,T})=1$ or $\psi$,
        \begin{equation}
            \tilde{\O}({\bf T,T,T}) = \widecheck{\O}({\bf T,T}) \overline{d\w({\bf T,T})} = \psi \times 1 = \psi
        \end{equation}
        Hence $[\tilde{\O}] \in \H^3(\Z_2^{\bf T},\Z_2)$ is nontrivial and the symmetry fractionalization is obstructed. 
        
        One can check that choosing the other action $v \rightarrow v s_1 s_2$ under time-reversal leads to the same result.
        
        There is a shortcut to see that there is an inconsistency here. We stated above that 
        \begin{equation}
            \eta_{s_1s_2}=+1=-\theta_{s_1 s_2} \text{ and } \eta_{\psi s_1 s_2} = -1 = -\theta_{\psi s_1 s_2}
        \end{equation}
        One can show on general grounds that \cite{barkeshli2019},
        \begin{equation}
            \eta_a = \theta_a \text{ if } a = b \times {\bf{T}}(b)
            \label{eqn:etaThetaConstraint}
        \end{equation}
        for some $b$ in the category in question. For $a=s_1s_2$ or $\psi s_1s_2$, we have $\eta_a \neq \theta_a$ for this fractionalization pattern. Eq.~\ref{eqn:etaThetaConstraint} is satisfied in $\mathcal{C}$ because, for the $a$ in question, there is no $b \in \mathcal{C}$ for which $a=b \times {\bf T}(b)$. However, in the modular extension $\C$, such a $b$ does exist; $s_1s_2 = vs_1s_2 \times {\bf T}(vs_1s_2)$ if we use the permutation action $v \rightarrow v s_1 s_2$, and $\psi s_1 s_2 = v s_1 s_2 \times {\bf T}(vs_1s_2)$ under the permutation action $v \rightarrow v\psi s_1 s_2$. Therefore, there is an inconsistency between the consistency of fractionalization in the modular extension and the fractionalization pattern in $\mathcal{C}$.

        \subsection{'t Hooft anomaly and dependence of $[o_3]$ on choices}
        \label{subsec:o3ChoiceDependence}
        
        In defining $[o_3]$, we have made two choices. We made a choice of modular extension $\C_\nu$ that is free of $\mathcal{H}^2$ obstruction and a choice of lift $\widecheck{\rho}$. Given these choices, it then makes sense to ask whether the symmetry fractionalization data $\{ \eta_a \}$ on $\mathcal{C}$ can be lifted to $\C_\nu$, and $[o_3] \in \mathcal{H}^3(G_b, \Z_2)$ is the obstruction to such a lift. In principle, $[o_3]$ could depend on the choices $\nu$ and $\widecheck{\rho}$; to highlight this dependence, we can write $[o_3^{(\nu, \widecheck{\rho})}]$. 
        In Sec. \ref{subsec:fermionSPT}, we saw that (3+1)D FSPTs define an element $[n_3] \in \mathcal{H}^3(G_b, \Z_2)/\Gamma^3$, if the lower layer data, $n_1, n_2$, vanish. Since $q_{\Gamma^3}( [o_3]) \in \mathcal{H}^3(G_b, \Z_2)/\Gamma^3$,
        it is therefore natural to assume that (2+1)D FSETs with vanishing $[o_1]$ and $[o_2]$ obstructions must exist at the surface of (3+1)D FSPTs characterized by $[n_3] = q_{\Gamma^3}( [o_3])$ (and vanishing $n_1, n_2$).

        It is natural to expect that the super-modular category $\mathcal{C}$ and its symmetry fractionalization data fully determine the 't Hooft anomaly, equivalently the (3+1)D FSPT that hosts the given theory at its surface. It would thus follow that $q_{\Gamma^3}( [o_3^{(\nu, \widecheck{\rho})}])$ should be independent of valid changes of $\nu$ and $\widecheck{\rho}$. Below we will examine this expectation in detail.
        
        Given a lift $\widecheck{\rho}$ for $\C_{\nu}$ and a group homomorphism $\pi \in Z^1(G_b,\Z_2)$ from $G_b$ to $\Z_2$, one can obtain a specific topological autoequivalence $\widecheck{\rho}_{\bf g}'$ on $\C_{\nu'}$~\cite{aasen21ferm}. In the case $\nu'-\nu = 0$, this amounts to modifying the lift $\widecheck{\rho}_{\bf g}$ for a fixed modular extension as follows:
        \begin{equation}
            \widecheck{\rho}_{\bf g} \rightarrow \widecheck{\rho}_{\bf g}' = \alpha_\psi^{\pi({\bf g})}\widecheck{\rho}_{\bf g}.
        \end{equation}
        
        Then, as we will explain, we expect the following result to hold in all cases:
         \begin{equation}
             [o_3^{(\nu', \widecheck{\rho}')}] = [o_3^{(\nu, \widecheck{\rho})}] + [s_1] \stdcup [\pi] \stdcup [\pi] +  [\pi] \stdcup [\tilde{\omega}_2] +  \frac{\nu' - \nu}{2} [\tilde{\omega}_2] \stdcup_1 [\tilde{\omega}_2]
            \label{eqn:o3ChangeModExt}
        \end{equation}
        Note that when $\nu' -\nu$ is odd, we can only consider the change of $[o_3]$ when $[\tilde{\omega}_2]$ and $[s]$ are trivial, in which case Eq. \ref{eqn:o3ChangeModExt} gives $[o_3^{(\nu', \widecheck{\rho}')}] = [o_3^{(\nu, \widecheck{\rho})}]$. This is because when $[\tilde{\omega}_2]$ is non-trivial and $\nu' -\nu$ is odd, Eq.~\ref{eqn:o2ABKResult} implies that either $[o_2^{(\nu)}]$ or $[o_2^{(\nu')}]$ is non-trivial, in which case $[o_3]$ would be ill-defined. Also, this equation is only meaningful with $s_1 \neq 0$ when $\nu'-\nu = 0 \mod 8$.
        
        Motivated by a conjecture of an earlier version of this paper,\footnote{An earlier version of this paper conjectured that $[o_3]$ itself, and not just its image under $q_{\Gamma^3}$ would be independent of $\nu$.} version 4 of Ref.~\cite{aasen21ferm} proved the above formula in general except for the case when $\Upsilon_\psi$ violates locality with $\nu'-\nu$ odd (which forces $\tilde{\omega}_2 = 0$). We reproduce~\footnote{v3 of Ref.~\cite{aasen21ferm} contained Eq.~\ref{eqn:o3ChangeModExt} restricted to the case $s_1 = 0$, $\nu' - \nu$ even, and $\widecheck{\Upsilon}_\psi$ respects locality. After discussions between the present authors and the authors of Ref.~\cite{aasen21ferm}, v4 corrected their formula for the change of $\widecheck{\Omega}_a$ under change of lift when $\widecheck{\Upsilon}_\psi$ violates locality and used this result to generalize Eq.~\ref{eqn:o3ChangeModExt} to all cases except when $\Upsilon_\psi$ violates locality with $\nu'-\nu$ odd.} the proof for $\nu'-\nu = 0$ (with $G_b$ allowed to be anti-unitary) in Appendix~\ref{app:o3Changes}, and also discuss invariance of $[o_3]$ under various gauge transformations. The remaining case where $\Upsilon_\psi$ violates locality with $\nu'-\nu$ odd is technically challenging, and therefore still open, but we see no conceptual reason to expect the formula to fail.
        
        The last unproven case notwithstanding, this result shows that $q_{\Gamma^3}([o_3^{(\nu,\widecheck{\rho})}])$ is the same for all $\nu$ with vanishing $[o_2^{(\nu)}]$ and all valid lifts $\widecheck{\rho}$. It also shows that if $q_{\Gamma^3}([o_3^{(\nu,\widecheck{\rho})})$ is trivial in $\H^3(G_b,\Z_2)/\Gamma^3$, then there exists a choice of $\nu$ and $\widecheck{\rho}$ such that $[o_3] \in \H^3(G_b,\Z_2)$ is trivial. As such, $q_{\Gamma^3}([o_3])$ should indeed be viewed as a piece of the 't Hooft anomaly of the FSET.
        
        \section{$\H^4(G_b,\U)$ obstruction and 't Hooft anomaly}
        \label{H4sec}
        
        Suppose that the obstructions $o_1,o_2,o_3$ all vanish. Then we can lift the full set of symmetry fractionalization data on $\mathcal{C}$ to $\C$. That is, we now have a bosonic SET, i.e. a UMTC $\C$ equipped with $G_b$ symmetry fractionalization data. One can then attempt to gauge $G_b$. As a first step, one must apply the standard formalism for bosonic topological phases~\cite{barkeshli2019} to construct a $G_b$-crossed extension of $\C$, that is, a theory which describes both the excitations in $\C$ and $G_b$ symmetry defects. There is a known obstruction~\cite{barkeshli2019,ENO2010}, which we shall call $[o_4]$, to doing so, valued in $\H^4(G_b,\U)$. $[o_4]$ quantifies the failure of the consistency of fusion of the symmetry defects, that is, it quantifies an inability to define $F$-symbols for the defects which obey the pentagon equation. This obstruction is the last obstruction in the anomaly cascade.
        
        The 4-cocycle $o_4$ depends explicitly on all of the choices required to gauge fermion parity and define symmetry fractionalization on $\C$. We can denote this dependence explicitly by writing $o_4^{(\nu, \widecheck{\rho}, \widecheck{\eta})}$. 

        Therefore we have a 4th cohomology class
        \begin{align}
          [o_4^{(\nu, \widecheck{\rho}, \widecheck{\eta})}] \in \mathcal{H}^4(G_b, \U) . 
        \end{align}
        $o_4^{(\nu, \widecheck{\rho}, \widecheck{\eta})}$ is in general a non-trivial function of its arguments. 
        
        It is straightforward to see how $[o_4]$ changes under changing the symmetry fractionalization class $\widecheck{\eta}$. Changes of symmetry fractionalization which are consistent with the symmetry fractionalization on $\mathcal{C}$ are given by:
        \begin{align}
        \widecheck{\eta}_a({\bf g}, {\bf h}) \rightarrow \widecheck{\eta}_a({\bf g}, {\bf h}) M_{a, \cohosub{t}({\bf g}, {\bf h})},   
        \end{align}
        with $\coho{t}({\bf g}, {\bf h}) \in \{1, \psi\} \simeq \Z_2$. One can then use the relative anomaly formula~\cite{barkeshli2019rel} to compute
        \begin{equation}
            o_4({\bf g,h,k,l}) \rightarrow o_4({\bf g,h,k,l}) R^{\cohosub{t}({\bf k,l}),\cohosub{t}({\bf g,h})}\eta_{\cohosub{t}({\bf k,l})}({\bf g,h}) = o_4({\bf g,h,k,l})(-1)^{\cohosub{t} \stdcup \cohosub{t} + \tilde{\omega}_2 \stdcup \cohosub{t}}
        \end{equation}
        where $\tilde{\omega}_2$ is the additive $\Z_2$ representation of $\omega_2$ appearing in Eq.~\ref{eqn:w2Def}. Comparing to Eq.~\ref{eqn:Gamma4Subgroup}, we see that $o_4$ changes by an element of $\Gamma^4$.
        
        We can also consider how $[o_4]$ depends on the choice of $\nu$ and $\widecheck{\rho}$. In general, changing $\nu$ and $\widecheck{\rho}$ may change the lower-level obstructions $[o_2]$, and $[o_3]$, in which case the change of $[o_4]$ is not well-defined. Nevertheless, we may consider changes in $\nu$ and $\widecheck{\rho}$ such that $o_2$ and $o_3$ remain trivial. It would be interesting to derive a general formula for how $[o_4]$ changes under changing $\nu$ and $\widecheck{\rho}$. It is not clear if a simple general formula exists. 
        
        While $[o_4^{(\nu, \widecheck{\rho}, \widecheck{\eta})}]$ is in general a non-trivial function of its arguments, we can consider the image of $[o_4]$ under the map $q_{\Gamma^4}$ defined in Eq.~\ref{eqn:qGamma4}. Below we conjecture that $q_{\Gamma^4}( [o_4] )$ is independent of the choices $\nu, \widecheck{\rho}, \widecheck{\eta}$. 

        Recall that in Sec.~\ref{subsec:fermionSPT}, we explained that (3+1)D FSPTs, and therefore 't Hooft anomalies for (2+1)D FSETs, are classified by a set of data $(n_1, n_2, n_3, \nu_4)$. Furthermore, if $n_1, n_2, n_3 = 0$, then (3+1)D FSPTs, and therefore 't Hooft anomalies of (2+1)D FSETs, are characterized by an element
        \begin{align}
            [\nu_4] \in \mathcal{H}^4(G_b, \U)/\Gamma^4. 
        \end{align}

        We expect that the 't Hooft anomaly of the (2+1)D system is entirely a property of the quasiparticles of the fermionic theory, described by the super-modular category $\mathcal{C}$, together with the symmetry action $\rho$ and symmetry fractionalization data $\eta$. If this expectation is correct, then the anomaly should be independent of all the choices involved in lifting the symmetry fractionalization data to a given modular extension. Anomaly matching therefore leads to the expectation that a (2+1)D FSET with $o_1,o_2,o_3$ all vanishing and some $[o_4]$ can only exist at the surface of a (3+1)D FSPT with $n_1,n_2,n_3$ vanishing and some $[\nu_4]$ such that 
        \begin{align}
            [\nu_4] = q_{\Gamma^4}( [o_4] ) . 
        \end{align}

        The above discussion then leads us to the following formal conjecture:
        \begin{conjecture}
        \label{H4conj}
        \begin{enumerate}
            \item[\phantom{a}]
            \item[(a)] $q_{\Gamma^4}( [o_4^{(\nu, \widecheck{\rho}, \widecheck{\eta})}]$ is independent of changing $\nu, \widecheck{\rho}, \widecheck{\eta}$, as long as $[o_1], [o_2], [o_3]$ all vanish. 
            \item[(b)] If $q_{\Gamma^4}( [o_4^{(\nu, \widecheck{\rho}, \widecheck{\eta})}])$ is trivial in $\mathcal{H}^4(G_b, \U)/\Gamma^4$, then there exists a choice of $\nu, \widecheck{\rho}, \widecheck{\eta}$ such that $[o_4^{(\nu, \widecheck{\rho}, \widecheck{\eta})}]$ is trivial in $\mathcal{H}^4(G_b, \U)$. 
        \end{enumerate}
        \end{conjecture}
        
        \section{Additional examples}
        \label{sec:examples}
        
         \subsection{T-Pfaffian}
            
            T-Pfaffian is the surface theory for the $G_b = U(1) \rtimes \Z_2^{\bf T}$ topological insulator, where $G_f = [U(1) \rtimes \Z_4^{{\bf T},f}]/\Z_2$. It consists of a subcategory $\mathcal{C}$ of the Ising $\times \U_{-8}$ state as follows. Labeling elements of Ising $\times \U_{-8}$ by $a_j$ with $a \in \{I,\psi,\sigma\}$ and $j=0,1,\ldots,7$, the quasiparticle content of $\mathcal{C}$ consists of the twelve quasiparticles $\{I_{2k},\psi_{2k},\sigma_{2k+1}\}$ for $k=0,1,2,3$. The transparent fermion of the category is $\psi=\psi_4$ (caution with the notation; $\psi$ alone means the physical fermion of $\mathcal{C}$, and $\psi_k$ is a label in Ising $\times \U_{-8}$). Time reversal $[\rho_{\bf T}]$ interchanges $I_2 \leftrightarrow \psi_2$ and $I_6 \leftrightarrow \psi_6$.
            
            It is not hard to check that $\Upsilon_\psi$ respects locality in this theory; therefore, $\Aut(\mathcal{C})=\Aut_{LR}(\mathcal{C})$, and one can check that $[\rho_{\bf g}]$ is determined entirely by its permutation action. In particular, writing ${\bf g} = (e^{i\theta_{\bf g}},{\bf T}^{a_{\bf g}})$ for $a\in\{0,1\}$, then $[\rho_{\bf g}]$ is the identity if $a_{\bf g}=0$ and is nontrivial and equal to $[\rho_{\bf T}]$ if $a_{\bf g}=1$, independent of $\theta_{\bf g}$.
            
            The minimal modular extensions of this theory have $c_- \in \frac{1}{2}\Z$. The minimal modular extension $\C$ which is compatible with time reversal, i.e., with $c_-=0$ can be written
            \begin{equation}
                \C = \left(\mathcal{C} \boxtimes \text{Ising}\right)/\{\psi \psi'\sim 1\},
                \label{eqn:gaugedTPfaffian}
            \end{equation}
            where the denominator means condensing the bound state of the physical fermion $\psi$ in $\mathcal{C}$ with the fermion $\psi'$ in the additional copy of Ising. One can check that $\C$ has only $\sigma$-type fermion parity vortices and admits a permutation action of ${\bf T}$ which lifts the permutation action of $\rho_{\bf T}$ as shown in Refs.~\cite{Bonderson13d,chen2014b,barkeshli2019tr}. Since $\mathcal{C}$ has no modular isotopes\cite{Bonderson07b}, the aforementioned permutation must define a full autoequivalence $[\widecheck{\rho}_{\bf T}]$ of $\mathcal{C}$, so the $\H^1(G_b,\Z_{\bf T})$ anomaly vanishes:
            \begin{equation}
                [o_1] = 0 \in \H^1(G_b,\Z_{\bf T}).
            \end{equation}
            
            Since all the vortices are $\sigma$-type, $\widecheck{\Upsilon}_\psi$ does not respect locality and is therefore the nontrivial element of $\ker r = \Z_2$. This permutation action squares to the identity permutation, although this does not guarantee that the $\H^2(G_b,\Z_2)$ anomaly vanishes because $\widecheck{\Upsilon}_\psi$ is non-permuting.
            
            Physically, we expect that the anomaly arises because the physical fermion $\psi$ carries nontrivial $\U^f$ quantum numbers, but the fermion parity vortices are all $\sigma$-type, that is, they absorb $\psi$. The fact that $\psi$ carries nontrivial $\U^f$ quantum numbers enters through $\eta_\psi$, which appears at the level of the $\H^3(G_b,\Z_2)$ obstruction. We therefore conjecture the following:
            \begin{conjecture}
            The T-Pfaffian state has vanishing $[o_2]$, but non-vanishing $[o_3]$.
            \end{conjecture}
            Checking this conjecture explicitly would require knowledge of the full $F$- and $R$-symbols of the gauged T-Pfaffian state, which we do not know in general how to compute from the decomposition Eq.~\ref{eqn:gaugedTPfaffian}.
            
            We contrast the present picture with a more typical argument \cite{cheng}. One can use a decorated domain wall construction wherein the bulk (3+1)D topological insulator can be understood in terms of decorating ${\bf T}$ domain walls with a (2+1)D integer quantum Hall state with Chern number $1$. This construction assumes $\U^f$ symmetry from the outset and requires a $c_-=1/2 \mod 1$ minimal modular extension on the boundary so that ${\bf T}$ domain walls on the boundary carry a chiral mode arising from the integer quantum Hall state in the bulk. However, $c_-=1/2 \mod 1$ is manifestly incompatible with time reversal symmetry in isolation. In this context, then, the anomaly is characterized as a nontrivial element of $\H^1(\Z_2^{\bf T},\Z_{\bf T})$. Another way to state this latter construction is that we gauge $\U^f$ symmetry first, and then attempt to lift the $G_f/\U^f = \Z_2^{\bf T}$ symmetry to the gauged theory. In this context, the T-Pfaffian anomaly appears as an $\H^1(G_f/\U^f,\Z_{\bf T})$ obstruction. By comparison, in our framework, we gauge fermion parity first and then attempt to lift the $G_f/\Z_2^f=G_b$ symmetry to the gauged theory. In the latter case, the $\H^1(G_b,\Z_{\bf T})$ anomaly vanishes, and T-Pfaffian fits into a higher level of the anomaly cascade. Both perspectives correspond to valid ways of calculating the same anomaly; they simply decompose the classification differently. Our perspective is more general, since fermionic systems always have $\Z_2^f$ symmetry but need not have $\U^f$ symmetry.
        
        \subsection{\texorpdfstring{$\Sp(3)_3 \times \{1,\psi\}$}{Sp33 x {1, psi}}}
        
        $\mathcal{C} = \Sp(3)_3 \boxtimes \{1,\psi\}$ gives an example of a $\nu = 6$ phase of the $\Z_{16}$ classification for $G_f = \Z_4^{{\bf T},f}$. $\Sp(3)_3$ has central charge $c = 1 \mod 8$, so we should consider $\Sp(3)_3 \times U(1)_{-1}$ to get a time-reversal invariant theory. $\Sp(3)_3$ itself has $20$ particles, so $\mathcal{C} = \Sp(3)_3 \boxtimes \{1,\psi\}$ has $40$ particles. The modular data of $\Sp(3)_3$ were obtained by computer using SageMath. We tabulate the quantum dimensions and topological twists of $\Sp(3)_3$ in Table~\ref{tab:Sp33Data}; the $S$-matrix is large and unenlightening, so we do not write it explicitly here. The particles will be labeled from 1 to 20, where particle 1 is the identity and the others are ordered by increasing quantum dimension but otherwise arbitrarily. We can label a particle in $\mathcal{C}$ by a pair $(n,x)$ with $n$ from 1 to 20 and $x \in \{1,\psi\}$. There is a unique action of time reversal which preserves the modular data up to complex conjugation, with the corresponding permutation given in Table~\ref{tab:Sp33Data}.
        
        \begin{table}
            \centering
            \renewcommand{\arraystretch}{1.3}
            \begin{tabular}{@{}clclclclclclclclclclc}
            \toprule[2pt]
                $n$ && 1 && 2 && 3 && 4 && 5 && 6 && 7 && 8 && 9 && 10  \\\hline 
                 $d_{(n,x)}$ && 1 && 1 && 3.49 && 3.49 && 4.49 && 4.49 && 4.49 && 4.49 && 5.60 && 5.60 \\
                 $\theta_{(n,1)}$ && 1 && $i$ && $i$ && 1 && $e^{-13 \pi i/14}$ && $e^{4 \pi i/7}$ && $e^{-\pi i/14}$ && $e^{-4 \pi i /7}$ && $e^{6 \pi i/7}$ && $e^{-6 \pi i/7}$ \\
                 ${\bf T}((n,1))$ && $(1, 1)$ && $(2,\psi)$ && $(3, \psi)$ && $(4,1)$ && $(7, \psi)$ && $(8,1)$ && $(5, \psi)$ && $(6,1)$ && $(10, 1)$ && $(9,1)$  \\
                 $\theta_{(n,a)}$ && $e^{-\pi i/4}$ && $e^{\pi i/4}$ && $e^{\pi i/4}$ && $e^{-\pi i/4}$ && $e^{23 \pi i/28}$ && $e^{9 \pi i/28}$ && $e^{-9 \pi i/28}$ && $e^{-23 \pi i/28}$ && $e^{17 \pi i/28}$ && $e^{25 \pi i/28}$\\
                 ${\bf T}((n,a))$ && $(2,\overline{a})$ && $(1,a)$ && $(4,a)$ && $(3, \overline{a})$ && $(8,a)$ && $(7, \overline{a})$ && $(6,a)$ && $(5,\overline{a})$ && $(12,\overline{a})$ && $(11, \overline{a})$  \\ \hline
                 \\ \hline
                 $n$ && 11 && 12 && 13 && 14 && 15 && 16 && 17 && 18 && 19 && 20\\ \hline
                 $d_{(n,x)}$ && 5.60 && 5.60 && 9.10 && 9.10 && 10.10 && 10.10 && 10.10 && 10.10 && 11.59 && 11.59\\
                 $\theta_{(n,1)}$ && $e^{-9 \pi i/14}$ && $e^{-5 \pi i/14}$ && 1 && $i$ && $e^{-2\pi i/7}$ && $e^{3 \pi i/14}$ && $e^{2\pi i/7}$ && $e^{11 \pi i/14}$ && $-i$ && $-1$ \\
                 ${\bf T}((n,1))$ && $(12, \psi)$ && $(11, \psi)$ && $(13, 1)$ && $(14, \psi)$ && $(17,1)$ && $(18, \psi)$ && $(15,1)$ && $(16, \psi)$ && $(19, \psi)$ && $(20,1)$ \\
                 $\theta_{(n,a)}$  && $e^{-25 \pi i/28}$ && $e^{-17 \pi i/28}$ && $e^{-\pi i/4}$ && $e^{\pi i/4}$ && $e^{-15 \pi i/28}$ && $e^{-\pi i/28}$ && $e^{\pi i/28}$ && $e^{15 \pi i/28}$ && $e^{-3\pi i/4}$ && $e^{3 \pi i/4}$ \\
                 ${\bf T}((n,a))$ && $(10, a)$ && $(9,a)$ && $(14, \overline{a})$ && $(13, a)$ && $(18, \overline{a})$ && $(17, a)$ && $(16, \overline{a})$ && $(15, a)$ && $(20, a)$ && $(19, \overline{a})$ \\ \bottomrule[2pt]
            \end{tabular}
            \caption{Quantum dimensions, topological twists, and time-reversal actions for $\C = \Sp(3)_3 \boxtimes \mathcal{I}(2)$. All other data can be determined from the tabulated data; e.g. $\theta_{(n,\psi)}=-\theta_{(n,1)}$, and ${\bf T}((n,\overline{a}))={\bf T}((n,a))\times (1,\psi)$. The action of ${\bf T}$ on $\mathcal{C} = \Sp(3)_3 \boxtimes \{1, \psi\}$ is uniquely determined and squares to the identity. Its lift to $\C$ is ambiguous by $\alpha_\psi$ and squares to $\alpha_\psi$. Quantum dimensions can be written exactly in terms of 56th roots of unity, but we do not do so here.}
            \label{tab:Sp33Data}
        \end{table}
        
        The time-reversal invariant modular extension is $\C = \Sp(3)_3 \boxtimes \mathcal{I}(-2)$. We label the parity vortices of $\mathcal{I}(-2)$ by $a$ and $\overline{a}$, with $a \times \psi = \overline{a}$ and $a \times \overline{a} = 1$. It is clear that $\C$ has $v$-type vortices, and furthermore that $\widecheck{\Upsilon}_\psi$ respects locality, so $\Aut(\C)=\Aut_{LR}(\mathcal{C})$ and $\ker r = \{1,\alpha_\psi\} = \Z_2$ where $\alpha_\psi$ fuses a fermion into each parity vortex. The modular data of this product theory can be computed straightforwardly from the modular data of the constituents. Labeling the particles in $\C$ as $(n,x)$ with, again, $n=1,2,\ldots,20$, and this time $x \in \{1,a,\overline{a},\psi\}$, one can check directly that there are exactly two lifts of the action of time reversal on $\mathcal{C}$ to $\C$. One is tabulated in Table~\ref{tab:Sp33Data}, and the other is given by composing with $\alpha_\psi$ (which switches $a \leftrightarrow \overline{a})$. It is clear by inspection that these actions do not square to the identity; for example, under time reversal,
        \begin{equation}
            (1,a) \rightarrow (2,\overline{a}) \rightarrow (1,\overline{a}) \rightarrow (2,a) \rightarrow (1,a)
        \end{equation}
        which yields a representation of $\Z_4^{\bf T}$, not of $\Z_2^{\bf T}$. Composition with $\alpha_\psi$ does not change this fact, so we conclude that, as expected, this theory has an $\H^2(\Z_2^{\bf T},\Z_2)$ obstruction.

        \section{Discussion}
        \label{sec:discussion}
        
        We have systematically characterized the set of obstructions which appear in lifting symmetry fractionalization data from a super-modular category $\mathcal{C}$ to a minimal modular extension $\C$. We found that this data is in good correspondence with the known classification of fermionic SPTs and provided an understanding for each obstruction:
        \begin{enumerate}
            \item The $\H^1(G_b,\Z_{\bf T})$ piece is the obstruction to defining a lift of the autoequivalence $[\rho_{\bf g}]$ to $[\widecheck{\rho}_{\bf g}] \in \Aut_{LR}(\C_{\nu})$ for some minimal modular extension $\C_{\nu}$.
            
            \item When the $\H^1$ obstruction vanishes and one can define lifts of the maps $[\rho_{\bf g}]$ for some choice of minimal modular extension $\nu$, the $\H^2(G_b,\ker r)$ piece is the obstruction to choosing lifts  $[\widecheck{\rho}_{\bf g}]:G_b \rightarrow \Aut_{LR}(\C_\nu)$ with the appropriate group structure.
            
            \item When the $\H^1$ and $\H^2$ obstructions vanish and we pick a set of lifts $[\widecheck{\rho}]$ with the appropriate group structure for a given modular extension $\nu$, the $\H^3(G_b,\Z_2)$ piece is the obstruction to lifting the symmetry fractionalization data $\{\eta_a\}$ on $\mathcal{C}$ to $\{\widecheck{\eta}_a\}$ on $\C_\nu$. 
            
            \item When the $\H^1$, $\H^2$, $\H^3$ obstructions all vanish and we pick the lifts $\widecheck{\rho}$ and symmetry fractionalization data $\widecheck{\eta}$ on $\C_\nu$ for some $\nu$, the $\H^4(G_b,\U)$ piece is the obstruction to lifting $\C_\nu$, together with its symmetry fractionalization data, to a $G_b$-crossed modular tensor category $\C_{G_b}^\times$.
        \end{enumerate}
        
        Our work raises a number of open questions. First, we have defined a series of obstructions $[o_1]$, $[o_2^{(\nu)}]$, $[o_3^{(\nu, \widecheck{\rho})}]$, and $[o_4^{(\nu, \widecheck{\rho}, \widecheck{\eta})}]$. We expect, but have not shown explicitly, that the bulk (3+1)D FSPT that hosts our given (2+1)D theory on its surface is characterized by the image of the maps $q_{\Gamma^i} : \mathcal{H}^i \rightarrow \mathcal{H}^i/ \Gamma^i$, for the finite groups $\Gamma^i$ reviewed in Section \ref{subsec:fermionSPT}. This has led us to conjecture that the $q_{\Gamma^i}[o_i]$ are independent of the various valid choices involving $\nu$, $\widecheck{\rho}$, and $\widecheck{\eta}$. For the case $i=2$, this conjecture is proven so long as $\ker r = \Z_2$. For the case $i=3$, there is a loose end in proving Eq.~\ref{eqn:o3ChangeModExt} for the case where $\Upsilon_\psi$ violates locality with $\nu'-\nu$ is odd. For $i = 4$, we have derived how $[o_4^{(\nu, \widecheck{\rho}, \widecheck{\eta})}]$ changes under changes of $\widecheck{\eta}$, but not under valid changes of $\nu, \widecheck{\rho}$.
        
        We have also conjectured, based on the (3+1)D FSPT classification, that in general $\ker r = \Z_2$. However we have also found an interesting example involving doubled $\SU(2)_6$ where there are multiple independent permutations of fermion parity vortices that keep the modular data invariant and act trivially on the super-modular category; these point to a possibility that $\ker r$ is larger than  $\Z_2$, or that there are permutations of anyons that preserve the modular data which do not correspond to auto-equivalences of the category. 
        
        While we have examples of a non-trivial $\H^2/\Gamma^2$ anomaly for anti-unitary symmetries (namely the semion-fermion theory with ${\bf T}^2 = (-1)^F$), we do not have an example for unitary $G_b$. 
        
        We have given an anomaly inflow argument for the $\H^2$ contribution to the anomaly using the decorated domain wall construction; it would be useful to gain a similar understanding for the other contributions.
        
        The Wang-Gu results for (3+1)D FSPTs reviewed in Sec. \ref{subsec:fermionSPT} have significantly more structure than we have derived so far in our obstruction theory. There is a set of data $(n_1, n_2, n_3, \nu_4)$ obeying complicated consistency equations, group multiplication laws under stacking invertible phases, and equivalences. However we have only seen how one can extract part of this data from the (2+1)D fermion SET. It would be interesting to understand to what extent more aspects of the general (3+1)D FSPT characterization can be extracted. For example, it may be possible to extract, from the (2+1)D data, the full group structure of the anomalies, which corresponds to a group extension of the groups $\H^1$, $\H^2/\Gamma^2$, $\H^3/\Gamma^3$, and $\H^4/\Gamma^4$. Alternatively, can we extract the specific consistency equations for $(n_1,n_2,n_3,\nu_4)$ derived in \cite{WangGu}?
        
        Our formalism suggests a very general way of understanding mixed anomalies. We have shown that the data characterizing the anomaly, namely, a spectral sequence decomposition into the Wang-Gu data of the bordism group $\Omega_4(BG_b, \xi)$, where $\xi$ is a kind of twisted spin structure, can be interpreted by first gauging fermion parity and then considering a sequence of obstructions to lifting the $G_b$ symmetry to the fermion parity gauged theory. Consider instead a bosonic SET with symmetry group $G = G_1 \times G_2$. Then there is a spectral sequence decomposition of the $\H^4(G,\U)$ bosonic SET anomaly via the K{\"u}nneth decomposition. Our procedure might be applied to understand the mixed anomaly between $G_1$ and $G_2$ in terms of a ``two-step" gauging process, where one first gauges $G_1$ and then determines a cascade of obstructions to lifting the $G_2$ symmetry to the $G_1$-gauged theory. More generally, if $G$ is given by a short exact sequence
        \begin{equation}
            1 \rightarrow G_1 \rightarrow G \rightarrow G_2 \rightarrow 1,
        \end{equation}
        there is again a spectral sequence decomposition of $\H^4(G,\U)$ and a similar construction may apply. 
        
        \section{Acknowledgements}
        
        We thank Parsa Bonderson, Meng Cheng, and Zhenghan Wang for discussions. In particular, MB thanks Parsa Bonderson for detailed comments on the first version of this paper posted to the arXiv. We also thank Ryohei Kobayashi and Srivatsa Tata for collaboration on related projects. This work is supported by NSF CAREER (DMR- 1753240) and JQI-PFC-UMD. 
        
        \it Note added in revision: \rm When the first version of this paper was posted to the arXiv, a number of closely related papers were also posted simultaneously \cite{bulmash2021,manjunath21inv,aasen21ferm}. \cite{bulmash2021} developed a theory of fermionic symmetry fractionalization, which is used as the starting point for this paper. \cite{manjunath21inv} developed a theory of (2+1)D invertible fermionic topological phases. \cite{aasen21ferm} develops a general characterization and classification of fermionic symmetry-enriched topological phases in (2+1)D, containing many results from \cite{bulmash2021,manjunath21inv}, and the current paper. 
        
        Ref. \cite{aasen21ferm} in particular included a result for the map $\alpha_\psi$, reproduced here in Eq.~\ref{eqn:alphaPsiUSymbols}, which is used here to prove a conjecture from the first version of this paper, which is now Eq.~\ref{eqn:UsigmaPsiSigmaMinus1}. Ref.~\cite{aasen21ferm} also used Eq.~\ref{eqn:alphaPsiUSymbols} to give the most general proof that $[\alpha_\psi]$ always generates a $\Z_2$ subgroup of $\ker r$, a fact which we use to remove the dependence of Theorem \ref{thm:kerR} on a technical conjecture stated in the first version of this paper. Moreover, in this revision we have added the discussion in Sec. \ref{UpsiCheckViolates}, which has overlap with results of \cite{aasen21ferm}, as noted in the main text. 
     
        Earlier versions of this paper conjectured that the obstructions $[o_2]$, $[o_3]$, and $[o_4]$ are independent of modular extension, based on general expectations of 't Hooft anomalies. Motivated by these conjectures, v2 of Ref. ~\cite{aasen21ferm} examined this question and showed how $[o_2]$ and $[o_3]$ can change under modular extension. In this revision we have revised these conjectures to state that $q_{\Gamma^i}( [o_i])$, for $i = 2,3,4$ is independent of the choices involved in the definition of $[o_i]$, matching results in the literature for classification of (3+1)D FSPTs. 
 
        We thank Parsa Bonderson for extensive discussions regarding revisions of both of our papers and for sharing early versions of revised drafts of \cite{aasen21ferm}.
        
        \appendix 
        
        \section{\texorpdfstring{$S$}{S}-matrix of a condensed theory}
        \label{app:SmatrixCondensed}
        
        We presently reproduce, with additional detail, the derivation from~\cite{delmastro2021} of the $S$-matrix $\hat{S}$ of the theory obtained by condensing an Abelian boson $\phi$ in a theory with $S$-matrix $S$. Our main addition to~\cite{delmastro2021} is a careful discussion of gauge invariance. Specifically, $\hat{S}$ should be gauge-invariant under vertex basis transformations, a fact which is in tension with the expression of $\hat{S}$ in terms of the gauge-dependent punctured $S$-matrix of the uncondensed theory.
        
        Assume that $\phi^n=1$ for some $n>1$; we will specialize to $n=2$ later. Then we can view the Wilson lines for $\phi$ as generating a $\Z_n$ 1-form symmetry of the theory; condensing $\phi$ corresponds to gauging that symmetry. We perform this gauging to understand the Hilbert space of the condensed theory on the torus.
        
        In the path integral picture, gauging the symmetry means that we sum over all insertions of the symmetry generator on nontrivial cycles of the 3-torus. We begin by considering insertions of a $\phi$ Wilson line through the time circle. Then at every spatial slice, we are summing over states with no insertion and with all possible insertions of $\phi^k$, that is, we consider an expanded Hilbert space of the original theory consisting of the Hilbert space on the torus and the Hilbert spaces on the punctured torus with punctures labeled $\phi^k$ for all $k<n$. (Notationally, the use of $k$ in this appendix is completely unrelated to the $k$ we defined in Sec.~\ref{sec:H2}.)
        
        States on the punctured Hilbert space with puncture $\phi^k$ are labeled $\ket{a;\phi^k}$ where $N_{a,\bar{a}}^{\phi^k}>0$. Physically, these states correspond to ones where an $a,\bar{a}$ pair is created from vacuum, wrapped around a given spatial cycle of the torus (which we will call cycle $\beta$ for concreteness), and then fused into $\phi^k$ at the puncture. Under a basis transformation $\Gamma^{ab}_c$ of the fusion spaces $V^{ab}_c$, these states transform as
        \begin{equation}
            \ket{a;\phi^k} \rightarrow \Gamma^{a\bar{a}}_{\phi^k}\left(\Gamma^{a\bar{a}}_1\right)^{\ast}\ket{a;\phi^k}
            \label{eqn:fusionSpaceTransformOnTorusHilbertSpace}
        \end{equation}
        and are thus gauge-invariant (in this sense) only when $k=0$.
        
        Also, observe that each anyon $a$ must have a $\Z_n$ orbit of some length $\ell_a$ which divides $n$. By definition, if $\ket{a;\phi^k}$ is a nonzero state, then $N_{a,\bar{a}}^{\phi^k} = N_{a,\phi^{k}}^a>0$, that is, $a \times \phi^k = a$. Hence $k$ is a multiple of $\ell_a$. Hence the total number of distinct states in the expanded Hilbert space that are associated to the anyon $a$ is $n/\ell_a$.
        
        Next consider inserting a Wilson loop for $\phi^k$ around cycle $\alpha$ of the torus. Let $W_\alpha(\phi^k)$ be the operator which does this; then since $\phi$ is Abelian,
        \begin{equation}
            W_\alpha(\phi^k)\ket{a;\phi^m} = M_{a,\phi^k}\ket{a;\phi^m}
        \end{equation}
        with $M$ the mutual statistics. Since $\phi$ is Abelian,
        \begin{equation}
            M_{a,\phi^k} = e^{2\pi i k q/n}
        \end{equation}
        for some integer $q$. Summing over all of these insertions,
        \begin{align}
            \sum_{k=0}^{n-1}W_\alpha(\phi^k)\ket{a;\phi^m} &= \left(\sum_{k=0}^{n-1}e^{2\pi i k q/n}\right)\ket{a;\phi^m}\\
            &= \delta_{q,0}\ket{a;\phi^m}
        \end{align}
        That is, we require $a$ to braid trivially with $\phi$ to keep its corresponding states in the condensed theory.
        
        Finally, we consider inserting a Wilson loop of $\phi^k$ Wilson loops around cycle $\beta$ of the torus. Clearly this takes the state $\ket{a;\phi^m} \rightarrow \ket{a\times \phi^k;\phi^m}$, so the only states we should consider are
        \begin{equation}
            \ket{[a];\phi^m} = \frac{1}{\sqrt{\ell_a}}\sum_{k=0}^{\ell_a} \ket{a \times \phi^k;\phi^m}
            \label{eqn:symmetricState}
        \end{equation}
        where here $[a]$ labels the orbit of $a$ under fusion with $\phi$.
        
        The Hilbert space therefore consists of states $\ket{[a];\phi^m}$ such that $a$ braids trivially with $\phi$ and $m=0,1,\ldots,n/\ell_a-1$ (that is, there are $n/\ell_a$ states per anyon). In particular, there is a unique state associated to $[a]$ if $a$ has orbit length $\ell_a=n$ but the orbit $[a]$ splits to $n/\ell_a$ states after condensation if $a \times \phi^{\ell_a} = a$ for some $\ell_a < n$.
        
        Now consider the transformation of these states under modular transformations. By definition, in this basis, if ${\bf S}$ and ${\bf T}$ are the operators which implement $S$ and $T$ modular transformations,
        \begin{align}
            \bra{a;\phi^m}{\bf S}\ket{b;\phi^s} &= \delta_{m,s}S_{ab}^{(\phi^m)}\\
            \bra{a;\phi^m}{\bf T}\ket{b;\phi^s} &= \delta_{m,s}T_{ab}^{(\phi^m)}\\
        \end{align}
        where $S_{ab}^{(\phi^m)}$ and $T_{ab}^{(\phi^m)}$ are the punctured $S-$ and $T$-matrices. Note that all elements of $T^{(\phi^m)}$ and diagonal elements $S_{aa}^{(\phi^m)}$ of the punctured $S$-matrix are gauge-invariant under vertex basis transformations, while off-diagonal elements of the punctured $S$-matrix (for $m>0$) are not generally gauge-invariant.
        
        It is diagrammatically straightforward to check that
        \begin{align}
            S_{a\times \phi^j,b}^{(\phi^m)}= S_{a,b\times \phi^k}^{(\phi^m)} = S_{a,b}^{(\phi^m)}
        \end{align}
        provided $[a]$ and $[b]$ are both deconfined particles, i.e., $a$ and $b$ both braid trivially with $\phi$), and that
        \begin{equation}
            T_{ab}^{(\phi^m)} = \theta_a \delta_{a,b}
        \end{equation}
        
        Hence
        \begin{align}
            \hat{S}_{([a],m),([b],s)}=\bra{[a];\phi^m}{\bf S}\ket{[b];\phi^s} &= \frac{1}{\sqrt{\ell_a \ell_b}} \sum_{a \in [a], b \in [b]} \delta_{m,s}S_{ab}^{(\phi^m)}\\
            &= \sqrt{\ell_a \ell_b}\delta_{m,s}S_{ab}^{(\phi^m)}\\
            \hat{T}_{([a],m),([b],s)]} = \bra{[a];\phi^m}{\bf T}\ket{[b];\phi^s} &= \delta_{m,s}T_{ab}^{(\phi^m)}\\
        \end{align}
        are well-defined expressions independent of the representatives of $[a]$ and $[b]$ we choose and are the $S$- and $T$-matrices of the condensed theory in this basis for the torus Hilbert space.
        
        We now discuss gauge freedom very carefully. In the uncondensed theory, there is a preferred basis $\ket{a;\phi^m}$ of the Hilbert space, where $a$ is any anyon. For $m=0$, these states are invariant under gauge transformations of the fusion spaces. However, there is still some gauge freedom, in the sense that we could send $\ket{a;1} \rightarrow e^{i\alpha_a}\ket{a;1}$ for some phases; this gauge freedom modifies the $S$-matrix. We can canonically fix this gauge freedom up to a global phase rotation (physically, a gauge choice for the vacuum state) by demanding that $S_{1 a}$ be real and positive for all $a$ and also that $S_{ab}$ be symmetric. Another way to say this is that we could be handed some phase-rotated states $\ket{a;1}$ from the outset; we would observe that we have the ``wrong" basis because the resulting $S$-matrix would not have these nice properties, and this could be corrected with a (diagonal) basis transformation.
        
        For $m>0$, there is no such canonical basis for these states. The situation is in some sense worse, because changing the basis of the fusion spaces induces a particular change of basis on the punctured torus Hilbert space, of the form Eq.~\ref{eqn:fusionSpaceTransformOnTorusHilbertSpace}, and therefore changes the punctured $S$-matrix of the uncondensed theory. This naively seems disturbing because the $S$-matrix of the condensed theory appears to depend on a fusion space basis in the original theory. However, from the above perspective, there is no such problem; such a change of basis of fusion spaces in the uncondensed theory is just a particular special case of changing the basis for the torus Hilbert space of the condensed theory. We simply imagine that we are handed the states $\ket{a;\phi^n}$ in some fixed but non-canonical basis, and then inspect the $S$-matrix of the condensed theory; if it does not have the intended properties, we were handed the ``wrong" basis and should perform a torus Hilbert space basis transformation to fix it.
        
        There is one subtlety here if $1 < \ell_a < n$, which is that we must gauge-fix 
        \begin{equation}
            W_\beta(\phi)\ket{a;\phi^m}=\ket{a\times \phi;\phi^m}
        \end{equation}
        in the uncondensed theory. This is needed in order to ensure that the state $\ket{[a];\phi^m}$ in Eq.~\ref{eqn:symmetricState} is indeed symmetric under insertion of $\phi$ Wilson loops.
        
        Having discussed gauge freedom, we now need to check if our states in the condensed theory are in the canonical gauge. For $\ell_a =\ell_b = n$, then $\hat{S}_{[a],[b]} \propto S_{a,b}$, so $\hat{S}$ inherits its nice properties from $S$; these states are in the correct gauge. For $\ell_a < n$, however, we are certainly not in the correct gauge since in this basis
        \begin{equation}
            \hat{S}_{1,([a],k)} = \begin{cases}
              \sqrt{2} S_{1,a} & k=0\\
              0 & \text{else}
            \end{cases}
        \end{equation}
        
        On general grounds, quasiparticles must correspond to some superposition of states $\ket{[a];\phi^m}$ with fixed $[a]$, so we may restrict our attention to a fixed-$[a]$ sector.
        
        At this point we restrict ourselves to $n=2$. Then if $\ell_a = 1$, we have $a \times \phi = a$ and a general basis transformation is of the form
        \begin{equation}
            \begin{pmatrix}
                    \ket{[a];+}\\
                    \ket{[a];-}
            \end{pmatrix}= U^{(a)} \begin{pmatrix}
                    \ket{[a];1}\\
                    \ket{[a];\phi}
            \end{pmatrix}
            =
            e^{i \alpha_a}\begin{pmatrix}
                    e^{-i(\beta_a+\delta_a)/2}\cos(\gamma_a/2) & -e^{-i(\beta_a-\delta_a)/2}\sin(\gamma_a/2)\\
                    e^{i(\beta_a-\delta_a)/2}\sin(\gamma_a/2) & e^{i(\beta_a+\delta_a)/2}\cos(\gamma_a/2)
            \end{pmatrix}
            \begin{pmatrix}
                    \ket{[a];1}\\
                    \ket{[a];\phi}
            \end{pmatrix}
        \end{equation}
        
        Strictly speaking $U^{(a)}$ is a diagonal block in the transformation of the entire Hilbert space; the state $\ket{1;1}$ does not transform in this basis transformation. Accordingly,
        \begin{align}
            (\hat{S}_{1,(a,+)},\hat{S}_{1,(a,-)}) &=  (\hat{S}_{1,(a,1)},\hat{S}_{1,(a,\phi)})\left(U^{(a)}\right)^{\dagger}\\
            &= \sqrt{2}S_{1,a}e^{-i(\alpha-\delta/2)}\left(e^{i\beta/2}\cos(\gamma/2), e^{-i\beta/2} \sin \gamma/2\right)
        \end{align}
        We require that
        \begin{equation}
            d_{(a,+)}+d_{(a,-)}=\frac{\hat{S}_{1,(a,+)}}{\hat{S}_{1,1}} + \frac{\hat{S}_{1,(a,-)}}{\hat{S}_{1,1}} = \frac{S_{1,a}}{S_{1,1}} = d_a
        \end{equation}
        following the usual rules for preserving the quantum dimension of split particles. Using $\hat{S}_{1,1}=2S_{1,1}$, equating the magnitudes leads to
        \begin{equation}
            \cos \beta_a \sin \gamma_a = 1
        \end{equation}
        That is, $\gamma=\pi/2$ and $\beta = 0$ or $\gamma=3\pi/2$ and $\beta = \pi$. We now ensure that $S_{1,(a,\pm)}$ are real and positive. In the first solution, we find this requires $\alpha - \delta/2 = 0 \mod 2\pi$, while in the second we find $\alpha - \delta/2 = 3\pi/2 \mod 2\pi$. Substituting back, we find that these two solutions are related by switching the rows of $U^{(a)}$, so that up to a basis reordering the only solutions for $U^{(a)}$ are
        \begin{equation}
            U^{(a)} = \frac{1}{\sqrt{2}}\begin{pmatrix}
                    1 & e^{i \delta_a}\\
                    1 & -e^{i \delta_a}
            \end{pmatrix}
        \end{equation}
        Note that $\delta_a$ is exactly the gauge freedom in $\ket{a;\phi}$.
        
        Next, we demand that $\hat{S}$ is symmetric. In this basis,
        \begin{equation}
            \hat{S}_{(a,\pm),(b,\pm)} = U^{(a)}\begin{pmatrix}
                S_{ab} & 0 \\
                0 & S_{ab}^{(\phi)}
            \end{pmatrix}\left(U^{(b)}\right)^{\dagger}
            =\frac{1}{2}\begin{pmatrix}
                S_{ab}+e^{i(\delta_a-\delta_b)}S_{ab}^{(\phi)} & S_{ab}-e^{i(\delta_a-\delta_b)}S_{ab}^{(\phi)}\\
                S_{ab}-e^{i(\delta_a-\delta_b)}S_{ab}^{(\phi)} &
                S_{ab}+e^{i(\delta_a-\delta_b)}S_{ab}^{(\phi)}
            \end{pmatrix}
        \end{equation}
        For $a=b$ this is clearly symmetric and independent of our remaining gauge freedom. If $a \neq b$, then symmetry of $\hat{S}$ amounts to a choice
        \begin{equation}
            e^{2i (\delta_a - \delta_b)} = \frac{S_{ba}^{(\phi)}}{S_{ab}^{(\phi)}}
        \end{equation}
        Equivalently, we may think of this as fixing a gauge so that the \textit{uncondensed} punctured $S$-matrix $S_{ab}^{(\phi)}$ is symmetric in $a$ and $b$. If we are already in such a gauge, then we may choose $\delta_a = \delta_b =0$, which reduces to the results of Ref.~\cite{delmastro2021}.
        
        To prove that such a gauge exists, we perform the diagrammatic manipulation in Fig.~\ref{fig:puncturedSSymmetry} to show
        \begin{equation}
            S^{(\phi)}_{ab}=\frac{d_a F^{a,\overline{a},a}_{a,\phi,1}F^{\phi,\phi,a}_{a,1,a}F^{\phi,a,\overline{a}}_{\phi,a,1}}{d_b F^{b,\overline{b},b}_{b,\phi,1}F^{\phi,\phi,b}_{b,1,b}F^{\phi,b,\overline{b}}_{\phi,b,1}}S^{(\phi)}_{ba}
            \label{eqn:puncturedSSymmetry}
        \end{equation}
        
        \begin{figure}
            \centering
            \includegraphics[width=0.8\columnwidth]{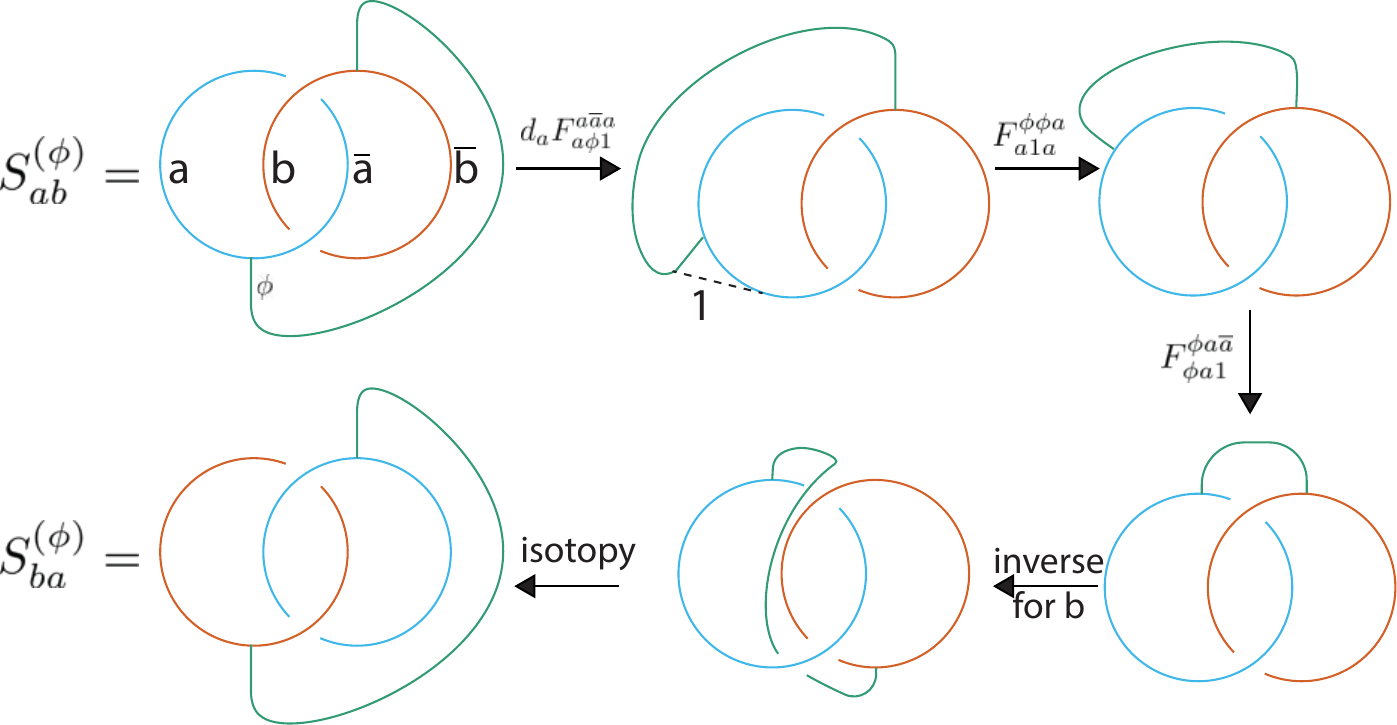}
            \caption{Diagrammatic manipulation in the standard BFC graphical calculus (see, e.g.,~\cite{Bonderson07b} or~\cite{barkeshli2019} for reviews and~\cite{BondersonBeyondModular,WenBeyondModularData} for graphical representations of the punctured $S$-matrix) leading to Eq.~\ref{eqn:puncturedSSymmetry}. Blue lines are $a$ or $\overline{a}$ orange are $b$ or $\overline{b}$, and green is $\phi$.  We are assuming $\phi\times \phi=1$ and that $a \times \phi = a$, $b \times \phi = b$. ``Inverse for $b$ means to repeat the same process as the first three steps, but in reverse and moving the $\phi$ line attached to the $\ket{b,\overline{b};\phi}$ vertex. In the step labeled ``isotopy", we use the fact that the double braid of $\phi$ is trivial with both $a$ and $b$.}
            \label{fig:puncturedSSymmetry}
        \end{figure}
        
        Letting
        \begin{equation}
            e^{2i \delta_a} = d_a F^{a,\overline{a},a}_{a,\phi,1}F^{\phi,\phi,a}_{a,1,a}F^{\phi,a,\overline{a}}_{\phi,a,1}
        \end{equation}
        we obtain the gauge transformation which symmetrizes the punctured $S$-matrix.
        
        Notice that there is some residual gauge freedom preserving the symmetry of the uncondensed punctured $S$-matrix; we may send $e^{i\delta_a} \rightarrow -e^{i\delta a}$. This gauge transformation simply flips the role of the rows in $U^{(a)}$ and thus corresponds to relabeling $(a,\pm)\rightarrow(a,\mp)$.
        
        \section{Proof of Proposition~\ref{prop:unitaryH2Trivial}}
        \label{app:H2ChangeModularExt}
        
        The proof is quite involved, so we begin with an outline of the strategy. First, for $\delta \nu \in \Z_{16}$, let $\mathcal{I}(\delta \nu)$ be the minimal modular extension of $\{1,\psi\}$ with chiral central charge $c_-=\delta \nu/2 \mod 8$. Then the 16-fold way tells us that, given $\C_{\nu_1}$, the minimal modular extension $\C_{\nu_2}$ is constructed as
        \begin{equation}
            \C_{\nu_2} = \C_{\nu_1}\boxtimes \mathcal{I}(\delta \nu)/\{\psi \psi \sim 1\}
        \end{equation}
        where the quotient means that we condense the bound state of the preferred fermions in $\C_{\nu_1}$ and $\mathcal{I}(\delta \nu)$. The proof goes in four steps:
        \begin{enumerate}
            \item[Step 1:] Use the results of~\cite{delmastro2021}, which we review in expanded detail in Appendix~\ref{app:SmatrixCondensed}, to derive the anyon content and the $S$- and $T$-matrices of the condensed theory $\C_{\nu_2}$.
            \item[Step 2:] Directly construct an anyon permutation $\widecheck{\rho}^{(2)}_{\bf g}$ on $\C_{\nu_2}$ that lifts the permutation action of $\rho_{\bf g}$.
            \item[Step 3:] Show that $\widecheck{\rho}_{\bf g}^{(2)}$ preserves the $S$- and $T$-matrices of $\C_{\nu_2}$.
            \item[Step 4:] Compute the permutation action of $o_2^{(2)}({\bf g,h})$.
        \end{enumerate}
        
        \underline{Step 1:} According to~\cite{delmastro2021}, the anyon content of $\C_{\nu_1+1}$ can be labeled by equivalence classes of anyons $(a,x) \in \C_{\nu_1}\boxtimes \mathcal{I}(\delta \nu)$ and, in some cases, a sign. In all cases, the anyon sector $(\C_{\nu_2})_0 \sim \mathcal{C}$ consists of the equivalence classes $(a_0,1) \sim (a_0 \times \psi,\psi)$ for $a_0 \in (\C_{\nu_1})_0=\mathcal{C}$. The behavior of the fermion parity vortex sector depends on the parity of $\delta \nu$.
        
        If $\delta \nu$ is even, then $\mathcal{I}(\delta \nu)$ is Abelian, and its particles can be labeled $\{1,\psi,v,v\times \psi\}$. Then the fermion parity vortices of $\C_{\nu_2}$ are as follows:
        \begin{enumerate}
            \item $(a_v,v)\sim (a_v \times \psi,v \times \psi)$ for $a_v \in (\C_{\nu_1})_v$. These form the sector $(\C_{\nu_2})_v$.
            \item $(a_\sigma, v) \sim (a_\sigma, v \times \psi)$ for $a_\sigma \in (\C_{\nu_1})_\sigma$. These form the sector $(C_{\nu_2})_\sigma$.
        \end{enumerate}
        Notice that anyon labels in $\C_{\nu_1}$ are in one-to-one correspondence with anyon labels in $\C_{\nu_2}$.
        
        If $\delta \nu$ is odd, the particles of $\mathcal{I}(\delta \nu)$ obey Ising fusion rules. Labeling the particles by $\{1,\psi,\sigma\}$, the fermion parity vortices of $\C_{\nu_2}$ are as follows:
        \begin{enumerate}
            \item $(a_v,\sigma) \sim (a_v \times \psi, \sigma)$ for $a_v \in (\C_{\nu_1})_v$. These form the sector $(\C_{\nu_1+1})_\sigma$.
            \item $(a_\sigma,\sigma)_\pm$ for $a_0 \in (\C_{\nu_1})_\sigma$. These form the sector $(C_{\nu_1+1})_v$ with $(a_\sigma,\sigma)_+ \times \psi = (a_\sigma,\sigma)_-$. We say in this case that $(a_\sigma,\sigma)$ splits after condensation.
        \end{enumerate}
        
        The $T$-matrix of the condensed theory is simple: the topological spin of $(a,x)$ (including the split case $(a_\sigma, \sigma)_\pm$) is $\theta_a \theta_x$.
        
        Next consider the $S$-matrix $S^{\C_{\nu_2}}$ of $\C_{\nu_2}$, and let $S^{(z)}$ be the punctured $S$-matrix for the product theory $\C_{\nu_1}\boxtimes \mathcal{I}(\delta \nu)$. As discussed in Appendix~\ref{app:SmatrixCondensed}, we can choose a gauge for the uncondensed theory such that $S^{((\psi,\psi))}_{ab}$ is symmetric. With that gauge-fixing, we can write down the following expression for the $S$-matrix of the condensed theory \cite{delmastro2021}. If $(a,x)$ and $(b,y)$ do not split under condensation, then
        \begin{equation}
            S_{(a,x),(b,y)}^{\C_{\nu_2}} = 2 S^{(1)}_{(a,x),(b,y)}
            \label{eqn:condensationSmatrixNoSplit}
        \end{equation}
        If only $(a,x)$ splits or only $(b,y)$ splits, then
        \begin{equation}
             S_{(a,x)_{\pm},(b,y)}^{\C_{\nu_2}} = S^{(1)}_{(a,x),(b,y)},
             \label{eqn:condensationSmatrixOneSplit}
        \end{equation}
        independent of the sign. Finally, if both $(a,x)$ and $(b,y)$ both split, then
        \begin{equation}
            S_{(a,x)_{\pm},(b,y)_{\pm}}^{\C_{\nu_2}} = \frac{1}{2}\begin{pmatrix}
            S^{(1)}_{(a,x),(b,y)} + S^{((\psi,\psi))}_{(a,x),(b,y)} & S^{(1)}_{(a,x),(b,y)} - S^{((\psi,\psi))}_{(a,x),(b,y)} \\
            S^{(1)}_{(a,x),(b,y)} - S^{((\psi,\psi))}_{(a,x),(b,y)} &
            S^{(1)}_{(a,x),(b,y)} + S^{((\psi,\psi))}_{(a,x),(b,y)}
            \end{pmatrix}
            \label{eqn:ScondensedSplit}
        \end{equation}
        where the matrix rows correspond to the sign for $(a,x)_{\pm}$ and the columns correspond to the sign for $(b,y)_{\pm}$. In this last case, $x=y=\sigma$ and we can calculate further using the fact that $x=y=\sigma$ and the punctured $S$-matrix for the product theory is the product of punctured $S$-matrices. We know that $S^{(1)}_{\sigma \sigma}=0$ in the Ising theory (and $S^{(1)}_{a_\sigma,b_\sigma}=0$ in $\C_{\nu_1}$ as well), so those terms all drop. By direct computation in the gauge of, e.g.,~\cite{Bonderson07b}, we find the only nonzero element of the punctured $S$-matrix in the Ising theory
        \begin{equation}
            S^{\mathcal{I}(\delta \nu),(\psi)}_{\sigma \sigma} = \sqrt{2} e^{-2\pi i \delta\nu/8}.
        \end{equation}
        Hence when both $(a,x)$ and $(b,y)$ split,
        \begin{equation}
            S_{(a,x)_{\pm},(b,y)_{\pm}}^{\C_{\nu_2}} = e^{-2\pi i \delta\nu/8}\frac{S^{\C_{\nu_1},(\psi)}_{a,b}}{\sqrt{2}}\begin{pmatrix}
            1 & -1\\
            -1 & 1
            \end{pmatrix}
            \label{eqn:SmatrixSplitSplitBlock}
        \end{equation}
        where $S^{\C_{\nu_1},(\psi)}_{a,b}$ is the punctured $S$-matrix of $\C_{\nu_1}$. Thanks to our gauge fixing for the punctured $S$-matrix in the uncondensed theory, $S^{\C_{\nu_2}}$ is symmetric.
        
        \underline{Step 2}: Define the permutation action $\widecheck{\rho}_{\bf g}^{(2)}$ as follows:
        \begin{align}
            \begin{cases}
                \widecheck{\rho}_{\bf g}^{(2)}((a,x)) = (\widecheck{\rho}^{(1)}_{\bf g}(a),x) \text{ if } (a,x) \text{ does not split}\\
                \widecheck{\rho}_{\bf g}^{(2)}((a,x)_\pm) = (\widecheck{\rho}^{(1)}_{\bf g}(a),x)_\pm \times \psi^{z_a({\bf g})} \text{ if } (a,x) \text{ splits}
            \end{cases}.
        \end{align}
        where $z_a({\bf g})\in \{0,1\} \simeq \Z_2$ must be defined for each $a \in \C^{(1)}_\sigma$.
        It is immediate that $\widecheck{\rho}_{\bf g}^{(2)}$ lifts the permutation action of $\rho_{\bf g}$ as long as $\widecheck{\rho}_{\bf g}^{(1)}$ does as well.
        
        In order to define $z_a({\bf g})$, first calculate directly
        \begin{equation}
            \widecheck{\rho}_{\bf g}(S^{\C_{\nu_1},(\psi)}_{ab}) = U_{\bf g}(\,^{\bf g}a,\,^{\bf g}\overline{a};\psi)U^{\ast}_{\bf g}(\,^{\bf g}a,\,^{\bf g}\overline{a};1)U^{\ast}_{\bf g}(\,^{\bf g}b,\,^{\bf g}\overline{b};\psi)U_{\bf g}(\,^{\bf g}b,\,^{\bf g}\overline{b};1)S^{\C_{\nu_1},(\psi))}_{\,^{\bf g}a,\,^{\bf g}b}= S^{\C_{\nu_1},(\psi))}_{a,b}
            \label{eqn:rhoOnPuncturedS}
        \end{equation}
        In our particular gauge, the punctured $S$-matrix is symmetric in $a$ and $b$. Hence we can write
        \begin{equation}
            \widecheck{\rho}_{\bf g}(S^{\C_{\nu_1},(\psi)}_{ab}) =  \widecheck{\rho}_{\bf g}(S^{\C_{\nu_1},(\psi)}_{ba}) = U_{\bf g}(\,^{\bf g}b,\,^{\bf g}\overline{b};\psi)U^{\ast}_{\bf g}(\,^{\bf g}b,\,^{\bf g}\overline{b};1)U^{\ast}_{\bf g}(\,^{\bf g}a,\,^{\bf g}\overline{a};\psi)U_{\bf g}(\,^{\bf g}a,\,^{\bf g}\overline{a};1)S^{\C_{\nu_1},(\psi)}_{\,^{\bf g}b,\,^{\bf g}a} 
            \label{eqn:rhoOnPuncturedS2}
        \end{equation}
        But $S^{\C_{\nu_1},(\psi)}_{\,^{\bf g}b,\,^{\bf g}a}=S^{\C_{\nu_1},(\psi)}_{\,^{\bf g}a,\,^{\bf g}b}$, so Eqs.~\ref{eqn:rhoOnPuncturedS},~\ref{eqn:rhoOnPuncturedS2} combine to
        \begin{equation}
        \left[U_{\bf g}(\,^{\bf g}b,\,^{\bf g}\overline{b};\psi)U^{\ast}_{\bf g}(\,^{\bf g}b,\,^{\bf g}\overline{b};1)\right]^2 = \left[U_{\bf g}(\,^{\bf g}a,\,^{\bf g}\overline{a};\psi)U^{\ast}_{\bf g}(\,^{\bf g}a,\,^{\bf g}\overline{a};1)\right]^2
        \end{equation}
        Equivalently, if we fix a reference $\sigma$-type vortex $r_\sigma$, then we must have
        \begin{equation}
            z_a({\bf g}) = \frac{U_{\bf g}(\,^{\bf g}a,\,^{\bf g}\overline{a};1)U^{\ast}_{\bf g}(\,^{\bf g}a,\overline{\,^{\bf g}a};\psi)}{U_{\bf g}(\,^{\bf g}r_\sigma,\,^{\bf g}\overline{r}_\sigma;1)U^{\ast}_{\bf g}(\,^{\bf g}r_\sigma,\,^{\bf g}\overline{r}_\sigma;\psi)} \in \Z_2
        \end{equation}
        where we are slightly abusing notation; the above defines $z_a({\bf g}) \in \{1,-1\} \simeq \Z_2$ instead of $\{0,1\} \simeq \Z_2$. We emphasize that the entire analysis above requires $a$ be a $\sigma$-type vortex. 
        
        The above proof actually does not apply if $S_{ab}^{\C_{\nu_1},(\psi)} = 0$. As in Sec.~\ref{subsec:kerR}, we can break the $S-$matrix of $(\C_{\nu_2})_v$ into $k$ blocks and apply the above argument to each block separately; we have a separate choice of $r_\sigma$ for each block.
        
        We choose $r_\sigma$ to be ${\bf g}$-independent. As we will see, such a choice will lead to a permutation $o^{(2)}({\bf g,h})$ which is exactly the identity; instead choosing a ${\bf g}$-dependent $r_\sigma$ in each block will amount to modifying the permutation action of $o^{(2)}({\bf g,h})$ by a $\ker r$-valued coboundary.
        
        \underline{Step 3:} It is immediately obvious that $\widecheck{\rho}_{\bf g}^{(2)}$ preserves the $T$-matrix provided $\widecheck{\rho}_{\bf g}^{(1)}$ does. Also, given the fact that $\widecheck{\rho}_{\bf g}^{(1)}$ preserves the $S$-matrix of $\C_{\nu_1}$, Eqs.~\ref{eqn:condensationSmatrixNoSplit},\ref{eqn:condensationSmatrixOneSplit} immediately imply that $\widecheck{\rho}_{\bf g}^{(2)}$ preserves $S_{(a,x),(b,y)}^{\C_{\nu_2}}$ when at most one of $(a,x)$ and $(b,y)$ split. The case where both split requires some calculation.
        
        From Eq.~\ref{eqn:rhoOnPuncturedS}, we see that $\widecheck{\rho}_{\bf g}^{(1)}(S^{\C_{\nu_1},(\psi)}_{a,b})=z_a({\bf g})z_b({\bf g})S^{\C_{\nu_1},(\psi)}_{\,^{\bf g}a,\,^{\bf g}b}=S^{\C_{\nu_1},(\psi)}_{a,b}$.  Hence
        \begin{equation}
        \widecheck{\rho}^{(2)}_{\bf g}(S^{\C_{\nu_2}}_{(a,x)_\pm,(b,y)_\pm}) =\frac{e^{-2\pi i \delta\nu/8} }{\sqrt{2}} \times z_a({\bf g})z_b({\bf g}) S_{\,^{\bf g}a \,^{\bf g}b}^{\C_{\nu_1},(\psi)} \times \begin{pmatrix}
                1 & -1 \\
                -1 & 1
        \end{pmatrix}
        \end{equation}
        
        On the other hand, $\widecheck{\rho}^{(2)}_{\bf g}((a,\sigma)_\pm)=(a,\sigma)_{\pm z_a({\bf g})}$. Hence,
        \begin{equation}
            S^{\C_{\nu_2}}_{\widecheck{\rho}_{\bf g}^{(2)}((a,x)_\pm)\widecheck{\rho}_{\bf g}^{(2)}((b,y)_\pm)} = S^{\C_{\nu_2}}_{(\,^{\bf g}a,x)_{\pm z_a({\bf g})}(\,^{\bf g}b,y)_{\pm z_b({\bf g})}} = \frac{e^{-2\pi i \delta\nu/8} }{\sqrt{2}} S_{\,^{\bf g}a \,^{\bf g}b}^{\C_{\nu_1},(\psi)} \times z_a({\bf g})z_b({\bf g})\begin{pmatrix}
                1 & -1 \\
                -1 & 1
        \end{pmatrix}= \widecheck{\rho}^{(2)}_{\bf g}(S^{\C_{\nu_2}}_{(a,x)_\pm,(b,y)_\pm})
        \end{equation}
        where the factor of $z_a({\bf g})$ is interpreted as permuting the rows of the matrix if it is $-1$ and the factor of $s_b({\bf g})$ should be interpreted as permuting the columns if it is $-1$. Hence $\widecheck{\rho}_{\bf g}^{(2)}$ indeed preserves the $S$-matrix.
        
        \underline{Step 4:} By direct computation,
        \begin{align}
            \begin{cases}
              \left(\widecheck{\rho}_{\bf gh}^{(2)}\right)^{-1}\widecheck{\rho}_{\bf g}^{(2)} \widecheck{\rho}_{\bf h}^{(2)} ((a,x)) &= (o_2^{(1)}({\bf g,h})(a), x) \text{ if }(a,x) \text{ does not split}\\
              \left(\widecheck{\rho}_{\bf gh}^{(2)}\right)^{-1}\widecheck{\rho}_{\bf g}^{(2)} \widecheck{\rho}_{\bf h}^{(2)} ((a,x)_\pm) &= (a, x)_\pm \times \psi^{z_a({\bf g})+z_a({\bf h})-z_a({\bf gh})} \text{ if }(a,x) \text{ splits}\\
            \end{cases}
        \end{align}
        If $\delta \nu$ is even, then no particles split, and $(o_2^{(1)}({\bf g,h})(a),x) = \psi \times (a,x)$ if and only if $o_2^{(1)}({\bf g,h})(a) = \psi \times a$. Hence $o_2^{(1)}$ and $o_2^{(2)}$, as permutations, are identical on the cochain level.
        
        If $\delta \nu$ is odd, then $(o_2^{(1)}({\bf g,h})(a),x) \sim (a,x)$ for all $a \in \C_{\nu_1}$. In particular, if $o_2^{(1)}$ acts nontrivially on $a$, then $a \in (\C_{\nu_1})_v$, and we saw that in such a case $(a,\sigma) \sim (a \times \psi, \sigma)$. Hence the only permutation action comes from split particles, and by inspection, modifying the reference vortex $r_\sigma$ in a ${\bf g}$-dependent way changes $o^{(2)}({\bf g,h})$ by a $\ker r$-valued coboundary. We can directly calculate from Eq.~\ref{eqn:o2Def} that
        \begin{align}
            o^{(2)}({\bf g,h})((a_\sigma,\sigma)_\pm) &= \frac{U^{\ast}_{\bf g}(\,^{\bf gh}a_{\sigma},\,^{\bf gh}\overline{a_\sigma};1)U_{\bf g}(\,^{\bf gh}a_{\sigma},\,^{\bf gh}\overline{a_\sigma};\psi)}{U^{\ast}_{\bf g}(\,^{\bf gh}r_\sigma,\,^{\bf gh}\overline{r}_\sigma;1)U_{\bf g}(\,^{\bf gh}r_\sigma,\,^{\bf gh}\overline{r}_\sigma;\psi)}
            \times \nonumber\\
            &\hspace{1cm}\times \frac{U^{\ast}_{\bf h}(\,^{\bf h}a_{\sigma},\,^{\bf h}\overline{a_\sigma};1)U_{\bf h}(\,^{\bf h}a_{\sigma},\,^{\bf g}\overline{a_\sigma};\psi)}{U^{\ast}_{\bf h}(\,^{\bf h}r_\sigma,\,^{\bf h}\overline{r}_\sigma;1)U_{\bf h}(\,^{\bf h}r_\sigma,\,^{\bf h}\overline{r}_\sigma;\psi)} \times \frac{U_{\bf gh}(\,^{\bf gh}a_{\sigma},\,^{\bf gh}\overline{a_\sigma};1)U^{\ast}_{\bf gh}(\,^{\bf gh}a_{\sigma},\,^{\bf gh}\overline{a_\sigma};\psi)}{U_{\bf gh}(\,^{\bf gh}r_\sigma,\,^{\bf gh}\overline{r}_\sigma;1)U^{\ast}_{\bf gh}(\,^{\bf gh}r_\sigma,\,^{\bf gh}\overline{r}_\sigma;\psi)}\\
            &= \frac{\kappa_{\bf g,h}(\,^{\bf gh}a_\sigma,\,^{\bf gh}\overline{a_{\sigma}};1)\kappa^{\ast}_{\bf g,h}(\,^{\bf gh}a_\sigma,\,^{\bf gh}\overline{a_{\sigma}};\psi)}{\kappa_{\bf g,h}(\,^{\bf gh}r_\sigma,\,^{\bf gh}\overline{r_{\sigma}};1)\kappa^{\ast}_{\bf g,h}(\,^{\bf gh}r_\sigma,\,^{\bf gh}\overline{r_{\sigma}};\psi)}\\
            &= +1
        \end{align}
        where the last line comes from decomposing $\kappa$ as a product of anyon-dependent factors $\beta$. Hence $\widecheck{\rho}_{\bf g}^{(2)}$ is a group homomorphism $G_b \rightarrow P(\C_{\nu_2})$.
        
        \section{Doubled \texorpdfstring{$\mathrm{SU}(2)_6$}{SU26}}
        \label{app:doubledSU26}
        
        We write down the UMTC data for $\C = \mathrm{SU}(2)_6 \times \mathrm{SU}(2)_6 \times \mathrm{Ising}_{-9/2}/\{\psi \psi \sim 1\}$. Here $\mathrm{Ising}_{-9/2}$ is the minimal modular extension of $\{1,\psi\}$ with central charge $c_-=-9/2$. The quotient means that we condense pairs of preferred fermions in these three spin modular theories, i.e. we condense $(6,6,0)$, $(0,6,\psi)$, and $(6,0,\psi)$, where the first two labels label particles in the two copies of $\mathrm{SU}(2)_6$ and the third labels particles in the Ising theory. 
        
        $\C$ contains 14 particles. Labeling the particles of $\mathrm{SU}(2)_6$ by integers from $0$ to $6$ in the usual way, and labeling elements of Ising as $\{0,\sigma,\psi\}$, the deconfined particles of the theory, their topological twists, and quantum dimensions are given in Table~\ref{tab:doubledSU26}.
        
        The particle $(6,0,0)$ is the preferred fermion of this spin modular theory. All labels are redundant under fusion with $(6,6,0)$, $(0,6,\psi)$, and $(6,0,\psi)$ in the product theory before condensation. We have
        \begin{align}
            \C_v &= \{(1,3,\sigma)_\pm,(3,1,\sigma)_\pm\}\\
            \C_\sigma &= \{(1,1,\sigma),(3,3,\sigma)\}
        \end{align}
        and all other particles are in $\C_0$. The total quantum dimension is $\mathcal{D} = 4(2+\sqrt{2})$.
        
        With $d=1+\sqrt{2}$ and the quasi-particles ordered as in Table~\ref{tab:doubledSU26}, the modular data of $\C$ is:
        \begin{align}
            S &= \frac{1}{\mathcal{D}}
            \begin{pmatrix}
                    1 & d & d & d & d & d^2 & d^2 & 1 & \sqrt{2}d &\sqrt{2}d &\sqrt{2}d &\sqrt{2}d & 2d & 2d\\
                    d & -1 & d^2 & -1 & d^2 & -d & -d & d& \sqrt{2}d &\sqrt{2}d &-\sqrt{2}d &-\sqrt{2}d & 2d & -2d\\
                    d & d^2 & -1 & d^2 & -1 & -d & -d & d& -\sqrt{2}d &-\sqrt{2}d &\sqrt{2}d &\sqrt{2}d & 2d & -2d\\
                    d & -1 & d^2 & -1 & d^2 & -d & -d & d& -\sqrt{2}d &-\sqrt{2}d &\sqrt{2}d &\sqrt{2}d & -2d & 2d\\
                    d & d^2 & -1 & d^2 & -1 & -d & -d & d& \sqrt{2}d &\sqrt{2}d &-\sqrt{2}d &-\sqrt{2}d & -2d & 2d\\
                    d^2 & -d & -d & -d & -d & 1 & 1 & d^2& -\sqrt{2}d &-\sqrt{2}d &-\sqrt{2}d &-\sqrt{2}d & 2d & 2d\\
                    d^2 & -d & -d & -d & -d & 1 & 1 & d^2& \sqrt{2}d &\sqrt{2}d &\sqrt{2}d &\sqrt{2}d & -2d & -2d\\
                    1 & d & d & d & d & d^2 & d^2 & 1& -\sqrt{2}d &-\sqrt{2}d &-\sqrt{2}d &-\sqrt{2}d & -2d & -2d\\
                    \sqrt{2}d & \sqrt{2}d & -\sqrt{2}d & -\sqrt{2}d & \sqrt{2}d & -\sqrt{2}d & \sqrt{2}d & -\sqrt{2}d  & 2\sqrt{2}d & -2\sqrt{2}d & 0 & 0 & 0 & 0\\
                    \sqrt{2}d & \sqrt{2}d & -\sqrt{2}d & -\sqrt{2}d &\sqrt{2}d & -\sqrt{2}d & \sqrt{2}d & -\sqrt{2}d & -2\sqrt{2}d & 2\sqrt{2}d & 0 & 0 & 0 & 0\\
                    \sqrt{2}d & -\sqrt{2}d & \sqrt{2}d & \sqrt{2}d & -\sqrt{2}d & -\sqrt{2}d & \sqrt{2}d & -\sqrt{2}d & 0 & 0 & 2\sqrt{2}d & -2\sqrt{2}d & 0 & 0 \\
                    \sqrt{2}d & - \sqrt{2}d & \sqrt{2}d & \sqrt{2}d & -\sqrt{2}d & -\sqrt{2}d & \sqrt{2}d & -\sqrt{2}d & 0 & 0 & -2\sqrt{2}d & 2\sqrt{2}d & 0 & 0 \\
                    2d & 2d & 2d & -2d & -2d & 2d & -2d & -2d & 0 & 0 & 0 & 0 & 0 & 0\\
                    2d & -2d & -2d & 2d & 2d & 2d & -2d & -2d& 0 & 0 & 0 & 0 & 0 & 0
            \end{pmatrix}\\
            T &= \mathrm{diag}(1,i,i,-i,-i,-1,1,-1,1,1,1,1,e^{5\pi i/4},e^{3\pi i/4})
        \end{align}
        
        There are two permutation actions on the anyons which lift the action of ${\bf T}$ on $\mathcal{C}$, complex conjugate the modular data, and square to the identity. Listed as a 14x14 matrix acting on anyon labels, these permutations are
        \begin{align}
        P_1 &=
            \begin{pmatrix}
                    1 & 0 & 0 & 0 & 0 & 0 & 0 & 0 & 0 & 0 & 0 & 0 & 0 & 0\\
                    0 & 0 & 0 & 1 & 0 & 0 & 0 & 0 & 0 & 0 & 0 & 0 & 0 & 0\\
                    0 & 0 & 0 & 0 & 1 & 0 & 0 & 0 & 0 & 0 & 0 & 0 & 0 & 0\\
                    0 & 1 & 0 & 0 & 0 & 0 & 0 & 0 & 0 & 0 & 0 & 0 & 0 & 0\\
                    0 & 0 & 1 & 0 & 0 & 0 & 0 & 0 & 0 & 0 & 0 & 0 & 0 & 0\\
                    0 & 0 & 0 & 0 & 0 & 1 & 0 & 0 & 0 & 0 & 0 & 0 & 0 & 0\\
                    0 & 0 & 0 & 0 & 0 & 0 & 1 & 0 & 0 & 0 & 0 & 0 & 0 & 0\\
                    0 & 0 & 0 & 0 & 0 & 0 & 0 & 1 & 0 & 0 & 0 & 0 & 0 & 0\\
                    0 & 0 & 0 & 0 & 0 & 0 & 0 & 0 & 0 & 0 & 1 & 0 & 0 & 0\\
                    0 & 0 & 0 & 0 & 0 & 0 & 0 & 0 & 0 & 0 & 0 & 1 & 0 & 0\\
                    0 & 0 & 0 & 0 & 0 & 0 & 0 & 0 & 1 & 0 & 0 & 0 & 0 & 0\\
                    0 & 0 & 0 & 0 & 0 & 0 & 0 & 0 & 0 & 1 & 0 & 0 & 0 & 0\\
                    0 & 0 & 0 & 0 & 0 & 0 & 0 & 0 & 0 & 0 & 0 & 0 & 0 & 1\\
                    0 & 0 & 0 & 0 & 0 & 0 & 0 & 0 & 0 & 0 & 0 & 0 & 1 & 0
            \end{pmatrix},\\
            P_2 &=
            \begin{pmatrix}
                    1 & 0 & 0 & 0 & 0 & 0 & 0 & 0 & 0 & 0 & 0 & 0 & 0 & 0\\
                    0 & 0 & 0 & 1 & 0 & 0 & 0 & 0 & 0 & 0 & 0 & 0 & 0 & 0\\
                    0 & 0 & 0 & 0 & 1 & 0 & 0 & 0 & 0 & 0 & 0 & 0 & 0 & 0\\
                    0 & 1 & 0 & 0 & 0 & 0 & 0 & 0 & 0 & 0 & 0 & 0 & 0 & 0\\
                    0 & 0 & 1 & 0 & 0 & 0 & 0 & 0 & 0 & 0 & 0 & 0 & 0 & 0\\
                    0 & 0 & 0 & 0 & 0 & 1 & 0 & 0 & 0 & 0 & 0 & 0 & 0 & 0\\
                    0 & 0 & 0 & 0 & 0 & 0 & 1 & 0 & 0 & 0 & 0 & 0 & 0 & 0\\
                    0 & 0 & 0 & 0 & 0 & 0 & 0 & 1 & 0 & 0 & 0 & 0 & 0 & 0\\
                    0 & 0 & 0 & 0 & 0 & 0 & 0 & 0 & 0 & 0 & 0 & 1 & 0 & 0\\
                    0 & 0 & 0 & 0 & 0 & 0 & 0 & 0 & 0 & 0 & 1 & 0 & 0 & 0\\
                    0 & 0 & 0 & 0 & 0 & 0 & 0 & 0 & 0 & 1 & 0 & 0 & 0 & 0\\
                    0 & 0 & 0 & 0 & 0 & 0 & 0 & 0 & 1 & 0 & 0 & 0 & 0 & 0\\
                    0 & 0 & 0 & 0 & 0 & 0 & 0 & 0 & 0 & 0 & 0 & 0 & 0 & 1\\
                    0 & 0 & 0 & 0 & 0 & 0 & 0 & 0 & 0 & 0 & 0 & 0 & 1 & 0
            \end{pmatrix}.
        \end{align}
        Two other anyon permutations lift the action on $\mathcal{C}$, complex conjugate the modular data, but yield a $\Z_4$ action; they are
        \begin{align}
            P_3 &=
            \begin{pmatrix}
                    1 & 0 & 0 & 0 & 0 & 0 & 0 & 0 & 0 & 0 & 0 & 0 & 0 & 0\\
                    0 & 0 & 0 & 1 & 0 & 0 & 0 & 0 & 0 & 0 & 0 & 0 & 0 & 0\\
                    0 & 0 & 0 & 0 & 1 & 0 & 0 & 0 & 0 & 0 & 0 & 0 & 0 & 0\\
                    0 & 1 & 0 & 0 & 0 & 0 & 0 & 0 & 0 & 0 & 0 & 0 & 0 & 0\\
                    0 & 0 & 1 & 0 & 0 & 0 & 0 & 0 & 0 & 0 & 0 & 0 & 0 & 0\\
                    0 & 0 & 0 & 0 & 0 & 1 & 0 & 0 & 0 & 0 & 0 & 0 & 0 & 0\\
                    0 & 0 & 0 & 0 & 0 & 0 & 1 & 0 & 0 & 0 & 0 & 0 & 0 & 0\\
                    0 & 0 & 0 & 0 & 0 & 0 & 0 & 1 & 0 & 0 & 0 & 0 & 0 & 0\\
                    0 & 0 & 0 & 0 & 0 & 0 & 0 & 0 & 0 & 0 & 0 & 1 & 0 & 0\\
                    0 & 0 & 0 & 0 & 0 & 0 & 0 & 0 & 0 & 0 & 1 & 0 & 0 & 0\\
                    0 & 0 & 0 & 0 & 0 & 0 & 0 & 0 & 1 & 0 & 0 & 0 & 0 & 0\\
                    0 & 0 & 0 & 0 & 0 & 0 & 0 & 0 & 0 & 1 & 0 & 0 & 0 & 0\\
                    0 & 0 & 0 & 0 & 0 & 0 & 0 & 0 & 0 & 0 & 0 & 0 & 0 & 1\\
                    0 & 0 & 0 & 0 & 0 & 0 & 0 & 0 & 0 & 0 & 0 & 0 & 1 & 0
            \end{pmatrix},\\
            P_4 &=
            \begin{pmatrix}
                    1 & 0 & 0 & 0 & 0 & 0 & 0 & 0 & 0 & 0 & 0 & 0 & 0 & 0\\
                    0 & 0 & 0 & 1 & 0 & 0 & 0 & 0 & 0 & 0 & 0 & 0 & 0 & 0\\
                    0 & 0 & 0 & 0 & 1 & 0 & 0 & 0 & 0 & 0 & 0 & 0 & 0 & 0\\
                    0 & 1 & 0 & 0 & 0 & 0 & 0 & 0 & 0 & 0 & 0 & 0 & 0 & 0\\
                    0 & 0 & 1 & 0 & 0 & 0 & 0 & 0 & 0 & 0 & 0 & 0 & 0 & 0\\
                    0 & 0 & 0 & 0 & 0 & 1 & 0 & 0 & 0 & 0 & 0 & 0 & 0 & 0\\
                    0 & 0 & 0 & 0 & 0 & 0 & 1 & 0 & 0 & 0 & 0 & 0 & 0 & 0\\
                    0 & 0 & 0 & 0 & 0 & 0 & 0 & 1 & 0 & 0 & 0 & 0 & 0 & 0\\
                    0 & 0 & 0 & 0 & 0 & 0 & 0 & 0 & 0 & 0 & 1 & 0 & 0 & 0\\
                    0 & 0 & 0 & 0 & 0 & 0 & 0 & 0 & 0 & 0 & 0 & 1 & 0 & 0\\
                    0 & 0 & 0 & 0 & 0 & 0 & 0 & 0 & 0 & 1 & 0 & 0 & 0 & 0\\
                    0 & 0 & 0 & 0 & 0 & 0 & 0 & 0 & 1 & 0 & 0 & 0 & 0 & 0\\
                    0 & 0 & 0 & 0 & 0 & 0 & 0 & 0 & 0 & 0 & 0 & 0 & 0 & 1\\
                    0 & 0 & 0 & 0 & 0 & 0 & 0 & 0 & 0 & 0 & 0 & 0 & 1 & 0
            \end{pmatrix}.
        \end{align}
            
    \section{Change of $[o_3]$ under change in lift}
    \label{app:o3Changes}
    
    Fix a minimal modular extension $\C$, and assume $\ker r = \Z_2$. Given symmetry fractionalization data on $\mathcal{C}$, suppose we have a lift $\widecheck{\rho}_{\bf g}$ for which $o_2=0$. Then if $\pi : G \rightarrow \Z_2 \simeq \{0,1\}$ is a group homomorphism,  all other valid lifts are obtained (up to locality-respecting natural isomorphism) by writing
    \begin{align}
        \widecheck{\rho}_{\bf g}^V &= \alpha_\psi^{\pi({\bf g})} \circ \widecheck{\rho}_{\bf g}
    \end{align}
    In this appendix, we reproduce~\cite{aasen21ferm} the calculation~\footnote{We thank Parsa Bonderson for sharing additional unpublished notes.} for how  $[o_3^{(\nu,\widecheck{\rho}^V)}]$ is related to $[o_3^{(\nu,\widecheck{\rho})}]$, which is a special case of Eq.~\ref{eqn:o3ChangeModExt}. We will also show that $o_3$ is invariant under various gauge choices.
    
    \subsection{Change of $o_3$}
    
    Recall our setup from Sec.~\ref{sec:H3}; we choose a gauge where $r(\widecheck{\rho}_{\bf g})=\rho_{\bf g}$ on the nose, and we can choose a gauge where the phases $\beta_a$ on $\mathcal{C}$ are just the restriction of $\widecheck{\beta}_a$. We are given symmetry fractionalization data $\omega_a({\bf g,h}) \in C^2(G_b,K(\mathcal{C}))$ on $\mathcal{C}$, which is guaranteed in this gauge to lift to some $\widecheck{\omega}_a \in C^2(G_b,K(\C))$.
    
    Now let us compute how this data changes for $\widecheck{\rho}_{\bf g}^V$. 
    
    First, we need the following facts about $\alpha_\psi$. Let $\varphi$ be an arbitrary topological (anti-)autoequivalence. Then from Eq.~\ref{eqn:alphaPsiUSymbols}, one can compute that
    \begin{align}
        \alpha_\psi \circ \varphi &= \Upsilon_\varphi \circ \varphi \circ \alpha_\psi\\
        \alpha_\psi^2 = \xi
    \end{align}
    where $\Upsilon_\varphi$ and $\xi$ are natural isomorphisms defined by their actions on anyons
    \begin{align}
        \gamma_{\Upsilon_\varphi,a_x} &= U_\varphi(\psi^x,a_x,a_x\psi^x)\\
        \gamma_{\xi,a_x} &= i^x.
    \end{align}
    Here we are denoting $a_x \in \C_x$ for $x \in \{0,1\}$. We also note that when treating $\pi$ as an \textit{integer}-valued function, it must obey
    \begin{equation}
        \pi({\bf g})+\pi({\bf h})=\pi({\bf gh})+2\pi({\bf g})\pi({\bf h}).
    \end{equation}
    
    Next, we need to compute how $\widecheck{\kappa}_{\bf g,h}$ changes.
    
    \begin{align}
    \widecheck{\kappa}^V_{\bf g,h} &= \alpha_\psi^{\pi({\bf gh})}\widecheck{\rho}_{\bf gh}\widecheck{\rho}_{\bf h}^{-1}\alpha_\psi^{-\pi({\bf h})}\widecheck{\rho}_{\bf g}^{-1}\alpha_\psi^{-\pi({\bf g})}\\
    &= \alpha_\psi^{\pi({\bf gh})}\widecheck{\rho}_{\bf gh}\widecheck{\rho}_{\bf h}^{-1}\widecheck{\rho}_{\bf g}^{-1}\alpha_\psi^{-\pi({\bf h})}\Upsilon_{\bf g}^{\pi({\bf h})}\alpha_\psi^{-\pi({\bf g})}\\
    &= \alpha_\psi^{\pi({\bf gh})} \widecheck{\kappa}_{\bf g,h} \alpha_\psi^{-\pi({\bf h})}\alpha_\psi^{-\pi({\bf g})}\Upsilon_{\bf g}^{\pi({\bf h})}\Upsilon_{\Upsilon_{\bf g}^{-\pi({\bf h})}}^{-\pi({\bf g})}\\
    &= \Upsilon_{\widecheck{\kappa}_{\bf g,h}}^{\pi({\bf gh})}\widecheck{\kappa}_{\bf g,h}\alpha_\psi^{\pi({\bf gh})} \xi^{-\pi(\bf g)\pi(\bf h)}\alpha_\psi^{-\pi({\bf gh})}\Upsilon_{\bf g}^{\pi({\bf h})}\Upsilon_{\Upsilon_{\bf g}^{-\pi({\bf h})}}^{-\pi({\bf g})}\\
    &= \Upsilon_{\widecheck{\kappa}_{\bf g,h}}^{\pi({\bf gh})} \xi^{-\pi(\bf g)\pi(\bf h)}\Upsilon_{\bf g}^{\pi({\bf h})}\Upsilon_{\Upsilon_{\bf g}^{-\pi({\bf h})}}^{-\pi({\bf g})}\widecheck{\kappa}_{\bf g,h}
\end{align}
where we have repeatedly commuted $\alpha_\psi$ through other maps, used the fact that
\begin{equation}
    \alpha_\psi^{\pi({\bf g})}\alpha_\psi^{\pi({\bf h})} = \alpha_\psi^{\pi({\bf gh})+2\pi({\bf g})\pi({\bf h})},
\end{equation}
and used the fact that $\alpha_\psi^2 = \xi$. Using the expressions for $\Upsilon_\varphi$, this change in $\widecheck{\kappa}_{\bf g,h}$ can be used to calculate the change in $\widecheck{\beta}_a$:
\begin{equation}
    \widecheck{\beta}_{a_x}' = (-i)^{x\pi({\bf g})\pi({\bf h})}\widecheck{\kappa}_{\bf g,h}(\psi^x,a_x,\psi^x a_x)^{\pi({\bf gh})}\widecheck{U}_{\bf g}(\psi^x,a_x,\psi^x a_x)^{\pi({\bf h})} \left[\frac{\widecheck{U}_{\bf g}(\psi^x,a_x,\psi^x a_x)}{\widecheck{U}_{\bf g}(\psi^x,\psi^x a_x,a_x)}\right]^{\pi({\bf g})\pi({\bf h})} \widecheck{\beta}_{a_x}
    \label{eqn:betaChange}
\end{equation}

    Note in particular that if $a \in \C_0$, then $\widecheck{\beta}_a$ is completely unchanged. Furthermore, $\alpha_\psi$ is strictly the identity when restricted to $\mathcal{C}_0$. Hence there is no gauge transformation on any of the data of $\mathcal{C}_0$, and we can choose the lift
    \begin{equation}
        \widecheck{\omega}_a^V = \widecheck{\omega}_a.
    \end{equation}
    We can now calculate the change in $o_3$ directly. The calculation proceeds differently for $\sigma$-type and $v$-type vortices:
    
    \subsubsection{$\sigma$-type vortices}
    
    First note that $d\widecheck{\omega}_a^V = d\widecheck{\omega}_a$ since the permutation action of the symmetry on $\sigma$-type vortices is unaffected. All change in $o_3$ must therefore come from a change in $\widecheck{\Omega}_a$.
    
    If there are any $\sigma$-type vortices, then $\widecheck{\Upsilon}_\psi$ must violate locality. Hence, since $o_2=0$, $\widecheck{\beta}_\psi = \omega_2$. Simplifying Eq.~\ref{eqn:betaChange} with the fact that $a\times \psi = a$ for $\sigma$-type vortices, we can calculate directly that
    \begin{align}
    \frac{\widecheck{\Omega}_a^V}{\widecheck{\Omega}_a} &= (-1)^{s_1 \stdcup \pi \stdcup \pi} \times \frac{\widecheck{\beta}_\psi({\bf h,k})^{\sigma({\bf g})\pi({\bf hk})} \widecheck{\beta}_\psi({\bf g,hk})^{\pi({\bf ghk})}}{\widecheck{\beta}_\psi({\bf g,h})^{\pi({\bf gh})}\widecheck{\beta}_\psi({\bf gh,k})^{\pi({\bf ghk})}}\times \frac{\widecheck{U}_{\bf h}(\psi,\,^{\overline{\bf g}}a, \,^{\overline{\bf g}}a)^{\sigma({\bf g})\pi({\bf k})} \widecheck{U}_{\bf g}(\psi,a, a)^{\pi({\bf hk})}}{\widecheck{U}_{\bf g}(\psi,a, a)^{\pi({\bf h})}\widecheck{U}_{\bf gh}(\psi,a, a)^{\pi({\bf k})}}.
\end{align}
    The factors of $\widecheck{\beta}$ come from decomposing $\widecheck{\kappa}$. Manipulating the various factors of $\pi$, we obtain
\begin{align}
    \frac{\widecheck{\Omega}_a^V}{\widecheck{\Omega}_a} &= (-1)^{s_1 \stdcup \pi \stdcup \pi} \times \frac{\omega_2({\bf h,k})^{\left[\pi({\bf hk})-\pi({\bf ghk})\right]}}{\omega_2({\bf g,h})^{\pi({\bf gh})-\pi({\bf ghk})}}\times \frac{\widecheck{U}_{\bf h}(\psi,\,^{\overline{\bf g}}a, \,^{\overline{\bf g}}a)^{\sigma({\bf g})\pi({\bf k})} \widecheck{U}_{\bf g}(\psi,a, a)^{\pi({\bf hk})-\pi({\bf h})}}{\widecheck{U}_{\bf gh}(\psi,a, a)^{\pi({\bf k})}}\\
    &= (-1)^{s_1 \stdcup \pi \stdcup \pi} \times \frac{\omega_2({\bf h,k})^{-\pi({\bf g})}}{\omega_2({\bf g,h})^{-\pi({\bf k})}} \times \widecheck{\kappa}_{\bf g,h}(\psi, a, a)^{-\pi({\bf k})}\widecheck{U}_{\bf g}(\psi,a,a)^{-2\pi({\bf h})\pi({\bf k})}
\end{align}
The canonical gauge-fixing $F^{a\psi \psi}=+1$ enforces
But $\widecheck{U}_{\bf g}(\psi,a,a) \in \{\pm 1\}$ for $\sigma$-type vortices, so so the above expression simplifies to
\begin{equation}
    \frac{\widecheck{\Omega}_a'}{\widecheck{\Omega}_a} = (-1)^{s_1 \stdcup \pi \stdcup \pi+\pi \stdcup \tilde{\omega}_2}.
    \label{eqn:changeInOmega}
\end{equation}

    \subsubsection{$v$-type vortices}

    Notice first that because of the change in permutation action, if $a \in \C_v$, then
    \begin{align}
        \frac{d\widecheck{\omega}_a^V({\bf g,h,k})}{d\widecheck{\omega}_a({\bf g,h,k})} &= \left(\frac{\widecheck{\omega}_{(\widecheck{\rho}_{\bf g}^V)^{-1}a}({\bf h,k})}{\widecheck{\omega}_{(\widecheck{\rho}_{\bf g})^{-1}a}({\bf h,k})}\right)^{\sigma({\bf g})}\\
        &=\left(\frac{\widecheck{\omega}_{(\widecheck{\rho}_{\bf g})^{-1}a\times \psi^{\pi({\bf g})}}({\bf h,k})}{\widecheck{\omega}_{(\widecheck{\rho}_{\bf g})^{-1}a}({\bf h,k})}\right)^{\sigma({\bf g})} = \omega_\psi({\bf h,k})^{\pi({\bf g})}
    \end{align}
    where we have used the fact that $\widecheck{\omega}_a$ obeys the fusion rules.

    The calculation for the change in $\widecheck{\Omega}_a$ simplifies dramatically with a convenient gauge-fixing. One can check that there always exists a gauge transformation $\gamma_a({\bf g})$ which is non-trivial only on $\C_1$ which fixes
    \begin{equation}
        \widecheck{U}_{\bf g}(\psi, a,a\times \psi) = +1
    \end{equation}
    for $v$-type vortices. This gauge-fixing leads to all of the factors of $\widecheck{U}$ and $\widecheck{\kappa}$ dropping out from Eq.~\ref{eqn:betaChange}. Now we can calculate carefully:
    \begin{align}
        \widecheck{\Omega}_a^V &= \frac{\widecheck{\beta}^V_{(\widecheck{\rho}_{\bf g}^V)^{-1}a}({\bf h,k})^{\sigma({\bf g})}\widecheck{\beta}^V_a({\bf g,hk})}{\widecheck{\beta}^V_a({\bf g,h})\widecheck{\beta}^V_a({\bf gh,k})}\\
        &= (-1)^{s_1 \stdcup \pi \stdcup \pi} \frac{\widecheck{\beta}_{(\widecheck{\rho}_{\bf g})^{-1}a\times \psi^{\pi({\bf g})}}({\bf h,k})^{\sigma({\bf g})}\widecheck{\beta}_a({\bf g,hk})}{\widecheck{\beta}_a({\bf g,h})\widecheck{\beta}_a({\bf gh,k})}\\
        &= (-1)^{s_1 \stdcup \pi \stdcup \pi} \frac{\widecheck{\beta}_{(\widecheck{\rho}_{\bf g})^{-1}a\times \psi^{\pi({\bf g})}}({\bf h,k})^{\sigma({\bf g})}}{\widecheck{\beta}_{(\widecheck{\rho}_{\bf g})^{-1}a}({\bf h,k})^{\sigma({\bf g})}} \widecheck{\Omega}_a({\bf g,h,k}) = (-1)^{s_1 \stdcup \pi \stdcup \pi} \widecheck{\beta}_\psi({\bf h,k})^{\pi({\bf g})} \widecheck{\Omega}_a({\bf g,h,k}).
    \end{align}
    Combining the above results, we see that
    \begin{equation}
        \frac{\widecheck{\Omega}_a^V\overline{d\widecheck{\omega}_a^V}}{\widecheck{\Omega_a}d\widecheck{\Omega_a}} = (-1)^{s_1 \stdcup \pi \stdcup \pi}\left(\frac{\widecheck{\beta}_\psi({\bf h,k})}{\widecheck{\omega}_\psi({\bf h,k})}\right)^{\pi({\bf g})} = (-1)^{s_1 \stdcup \pi \stdcup \pi}\omega_2({\bf h,k})^{\pi({\bf g})}
    \end{equation}
    since $\widecheck{\beta}_\psi = \beta_\psi$, $\widecheck{\omega}_\psi = \omega_\psi$, and $\beta_\psi/\omega_\psi = \eta_\psi = \omega_2$.
    
    Accordingly,
    \begin{equation}
        o_3 \rightarrow o_3 + s_1 \stdcup \pi\stdcup \pi + \pi \stdcup \tilde{\omega}_2
    \end{equation}
    as expected.
    
    \subsection{Gauge invariances of \texorpdfstring{$[o_3]$}{o3}}
    
    We discuss the invariance of $[o_3]$ under various gauge transformations.
    
    \subsubsection{Invariance under locality-respecting natural isomorphisms on \texorpdfstring{$\mathcal{C}$}{C}}
    
    If we modify $\rho_{\bf g} \rightarrow \Upsilon_{\bf g} \circ \rho_{\bf g}$ with $\Upsilon_{\bf g}$ a locality-respecting natural isomorphism, then $\omega_a$ is unchanged. Suppose that $\Upsilon_{\bf g}$ is given by the anyon-dependent factors $\gamma_a({\bf g})$ with $\gamma_\psi = +1$. Then we can lift $\Upsilon_{\bf g}$ to a locality-respecting natural isomorphism on $\C$ by defining anyon-dependent factors
        \begin{equation}
            \widecheck{\gamma}_a({\bf g}) = \begin{cases}
              \gamma_a({\bf g}) & a \in \C_0\\
              +1 & a \in \C_1
            \end{cases}.
        \end{equation}
        Under this transformation, $\widecheck{\Omega}_a$ is invariant, and the condition that $\widecheck{\beta}_a$ restricts to $\beta_a$ is respected. Hence the lift $\widecheck{\omega}_a$ is also unchanged, and $o_3$ is strictly invariant.
    
    \subsubsection{Invariance under locality-respecting natural isomorphisms on \texorpdfstring{$\C$}{Ccheck}}
    
    Since we are demanding that $\widecheck{\rho}_{\bf g}$ restricts to $\rho_{\bf g}$ on the nose, such a  locality-respecting natural isomorphism must be non-trivial only on $\C_1$ (modulo a $\nu$-type gauge transformation, which we will deal with next). It is thus clear that the condition that $\widecheck{\beta}_a$ restricts to $\beta_a$ is respected, so none of the data on $\mathcal{C}$ is modified. Hence the allowed lifts $\widecheck{\omega}_a$ are unchanged. Also, $\widecheck{\Omega}_a$ is simply invariant under locality-respecting natural isomorphisms. Hence $o_3$ is gauge-invariant.
    
    \subsubsection{Invariance under \texorpdfstring{$\nu$}{nu}-type gauge transformations on \texorpdfstring{$\C$}{Ccheck}}
    
    Recall that $\nu$-type gauge transformations modify
    \begin{align}
        \widecheck{\beta}_a({\bf g,h}) &\rightarrow \widecheck{\nu}_a({\bf g,h})\widecheck{\beta}_a({\bf g,h})\\
        \widecheck{\omega}_a({\bf g,h}) & \rightarrow \widecheck{\nu}_a({\bf g,h}) \widecheck{\omega}_a({\bf g,h})
    \end{align}
    where $\widecheck{\nu}_a({\bf g,h})$ obeys the fusion rules of $\C$. Under such a transformation, in order to maintain the condition that $\widecheck{\beta}_a$ restricts to $\beta_a$ and $\widecheck{\omega}_a$ restricts to $\omega_a$, we should also perform the restricted gauge transformation $\nu_a$ on $\mathcal{C}$, which is allowed since $\nu_a$ obeys the fusion rules of $\mathcal{C}$. We can then work directly with the gauge-transformed $\widecheck{\Omega}_a$ and $\widecheck{\omega}_a$. It is straightforward to check that $\widecheck{\Omega}_a$ and $d\widecheck{\omega}_a$ transform by the same factor under this gauge transformation, so $o_3$ is strictly invariant.

        \bibliography{TI}
        
\end{document}